%% file: svc_mult_full.tex
\documentclass[conference]{IEEEtran}

\usepackage{epsfig,graphics,graphicx,amssymb,amstext,amsmath,algorithm,
algorithmic,wrapfig,multirow}
\usepackage{cite}
\usepackage{url}
\usepackage{times}
\usepackage{float}
\usepackage{placeins}
\usepackage{textcomp}
\usepackage[font=small,bf]{caption}
\usepackage{flushend}
\usepackage{xcolor}
\usepackage{lipsum}
\usepackage{color}
\usepackage{subcaption}
\usepackage{authblk}
\usepackage{bbm}
\usepackage{dsfont}
\usepackage{balance}
\usepackage{rotating}
\AtEndDocument{\par\leavevmode}

\usepackage{tabularx}
\usepackage{perpage} 
\newtheorem{theorem}{Theorem}[section]

\newtheorem{corollary}[theorem]{Corollary}
\newtheorem{definition}[theorem]{Definition}

\newcommand{\qed}{\nobreak \ifvmode \relax
 \else
      \ifdim\lastskip<1.5em \hskip-\lastskip
      \hskip1.5em plus0em minus0.5em \fi \nobreak
      \vrule height0.75em width0.5em depth0.25em\fi}

\newenvironment{proof}[1][Proof]{\begin{trivlist}
\item[\hskip \labelsep {\bfseries #1}]}{\end{trivlist}}

\MakePerPage{footnote} 

\IEEEoverridecommandlockouts

\renewcommand{\footnoterule}{%
  \kern -3pt
  \hrule width \columnwidth height 0.5pt
  \kern 3pt
}
\begin{document}

\title{Restless Video Bandits: Optimal SVC Streaming in a Multi-user Wireless Network}

\author{S. Amir Hosseini}
\author{Shivendra S. Panwar}
\affil{Department of Electrical and Computer Engineering, NYU Tandon School of Engineering}


\pagenumbering{gobble}

\maketitle

\begin{abstract}
In this paper, we consider the problem of optimal scalable video delivery to mobile users in wireless networks given arbitrary Quality Adaptation (QA) mechanisms. In current practical systems, 
QA and scheduling are performed independently by the content provider and network operator, respectively. While most research has been focused on jointly
optimizing these two tasks, the high complexity that comes with a joint approach makes the implementation impractical. Therefore, we present
a scheduling mechanism that takes the QA logic of each user as input and optimizes the scheduling accordingly. Hence, there is no need for centralized QA and 
cross-layer interactions are minimized. We model the QA-adaptive scheduling and the jointly optimal problem as a Restless Bandit and a Multi-user Semi Markov Decision Process,
respectively in order to compare the loss incurred by not employing a jointly optimal scheme. We then present heuristic algorithms in order to achieve the optimal outcome of the Restless Bandit solution 
assuming the base station has knowledge of the underlying quality adaptation of each user (QA-Aware). We also present a simplified heuristic without the need for any higher layer knowledge at the base station (QA-Blind). We show that our QA-Aware strategy can achieve up to two times improvement in user network utilization compared to popular baseline algorithms such as Proportional Fairness.
We also provide a testbed implementation of the QA-Blind scheme in order to compare it with baseline algorithms in a real network setting.  
\end{abstract}

\input{introduction}
\input{related_work}

\input{system_model}

\input{formulation}

\input{algorithm}
\input{simulation}

\input{implementation}
\input{conclusion}
\input{appendices}

\bibliographystyle{IEEEtran}
\bibliography{svc_mult}
\end{document}

%% file: introduction.tex
\section{Introduction}\label{sec:intro}
Video and real time applications are the largest consumer of mobile wireless data (40\% during peak consumption) in North America, and it is predicted that this trend will continue \cite{sandvine2016}. This calls for more intelligent usage of the available spectrum and more bandwidth conserving techniques for video delivery. For this purpose, adaptive video delivery over HTTP has been standardized under the commercial name DASH. DASH can also be implemented using the scalable extension of the video codec H.264/SVC and H.265/SHVC. In SVC, each segment is encoded into a base layer containing the minimum quality representation, and one or more enhancement layers for additional quality. Apart from higher flexibility in segment delivery, SVC also benefits the network in terms of caching efficiency and congestion reduction at the server \cite{sanchez2011idash}. These benefits have led to efforts for commercial deployment of SVC. For instance, Vidyo and Google have begun a collaboration for implementing SVC on WebRTC using the VP9 codec \cite{vidyo_2013}. The process of delivering adaptive video using SVC in a wireless network can be broken into two separate tasks:


\emph{Quality Adaptation (QA)}: QA determines the order in which different layers of different segments must be requested by the user and is performed by an end-to-end application specified by the content provider (Netflix, Amazon, etc.). The adaptation policy is not specified in the DASH standard and therefore, depending on the user device, video application, content provider, etc., different vendors can use different policies. 

\emph{Scheduling}: In multi-user wireless networks, where the bottleneck is typically the access link, the base station determines how the time-frequency resources are shared among users. This task is referred to as scheduling, and it is a design choice of the network service provider. In general, the scheduling policy should ensure high Quality of Experience (QoE) and utilize the wireless resources efficiently. 


The above two tasks can either be implemented independent of each other as shown in Figure \ref{fig:sys_sep}, or jointly by the base station as illustrated in Figure \ref{fig:sys_joint}. A joint optimization would require considerable cross-layer functionality at the base station making it impractical and overly complex. It would also call for coordination between content providers and network operators which is undesirable because it forces the content provider to give away control over its content delivery process, which it may be reluctant to do for business reasons. On the other hand, separately optimizing the two tasks provides inferior system performance compared to the joint case. In this paper, we combine the two schemes in Figure \ref{fig:sys} and we design a scheduling policy that adapts itself to any arbitrary QA policy that is implemented on each end user (QA-adaptive scheduling). In our proposed system model, end users can deploy any QA provided by the content provider. The network then takes the QA of each user as input and optimizes the scheduling accordingly. As a result, content providers will still have full control over the adaptation process and the scheduling will be adaptive to the underlying QA. Furthermore, this separation between QA and scheduling may allow service providers, bound by the evolving net neutrality rules, a new option to maximize QoE without explicit, and therefore possibly discriminatory, cooperation with content providers. The recently developed MPEG's Server and Network Assisted DASH (SAND) technology, offers standardized messaging schemes and protocol exchanges for service providers and operators to enhance streaming experience while also improving network bandwidth utilization \cite{thomas2016applications}. The exchange of QA logic between content provider and the network can be done within the SAND framework.

\begin{figure}[h]
	\captionsetup[subfigure][h]{twoside,margin={-0.1cm,-0.1cm}}
	\centering
	\begin{subfigure}[h]{0.4\textwidth}
             \includegraphics[height=1.4in]{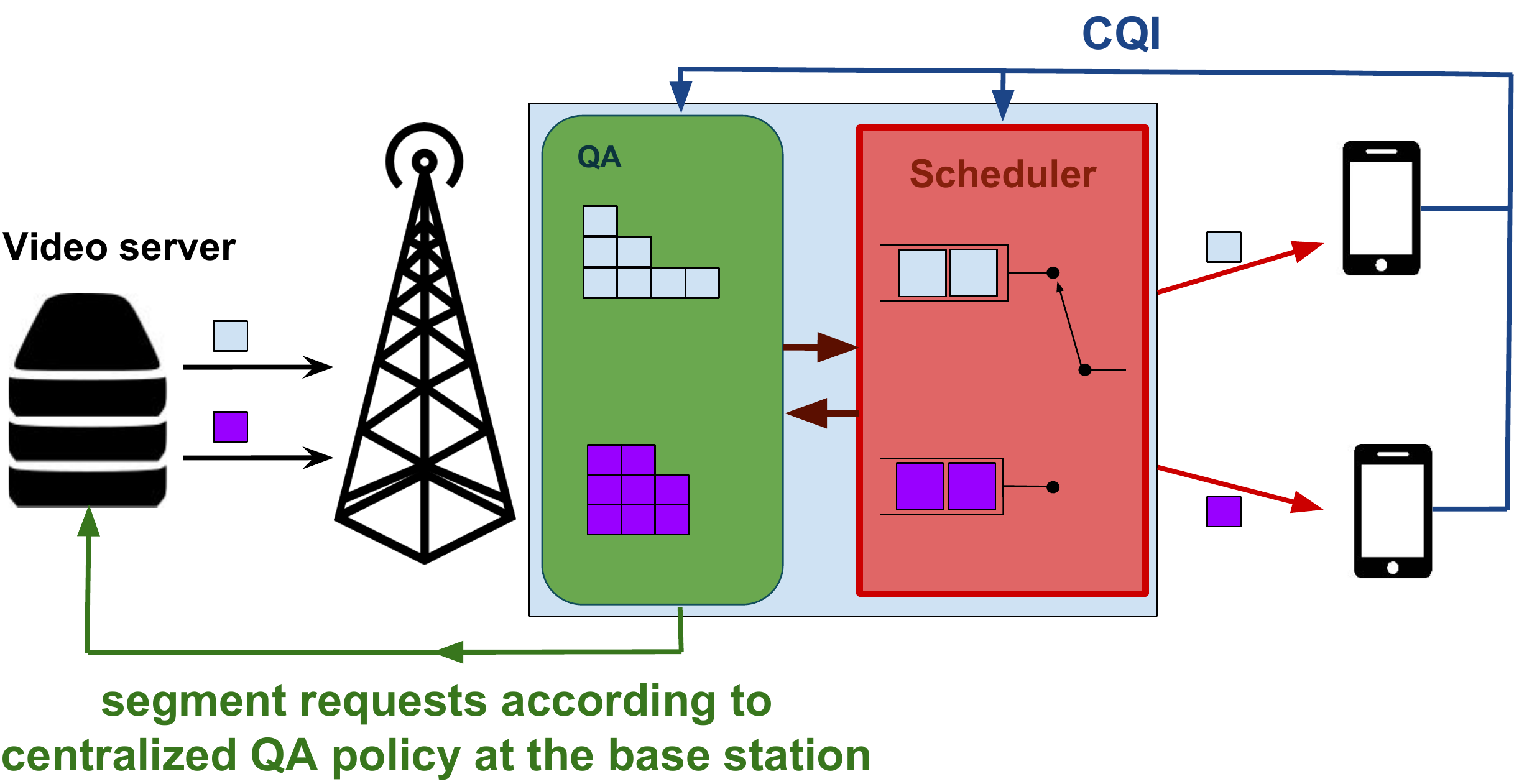}
             \caption{Joint optimal QA and scheduling.}\label{fig:sys_joint}         
        \end{subfigure}\qquad \\
        \begin{subfigure}[h]{0.4\textwidth}
             \includegraphics[height=1.5in]{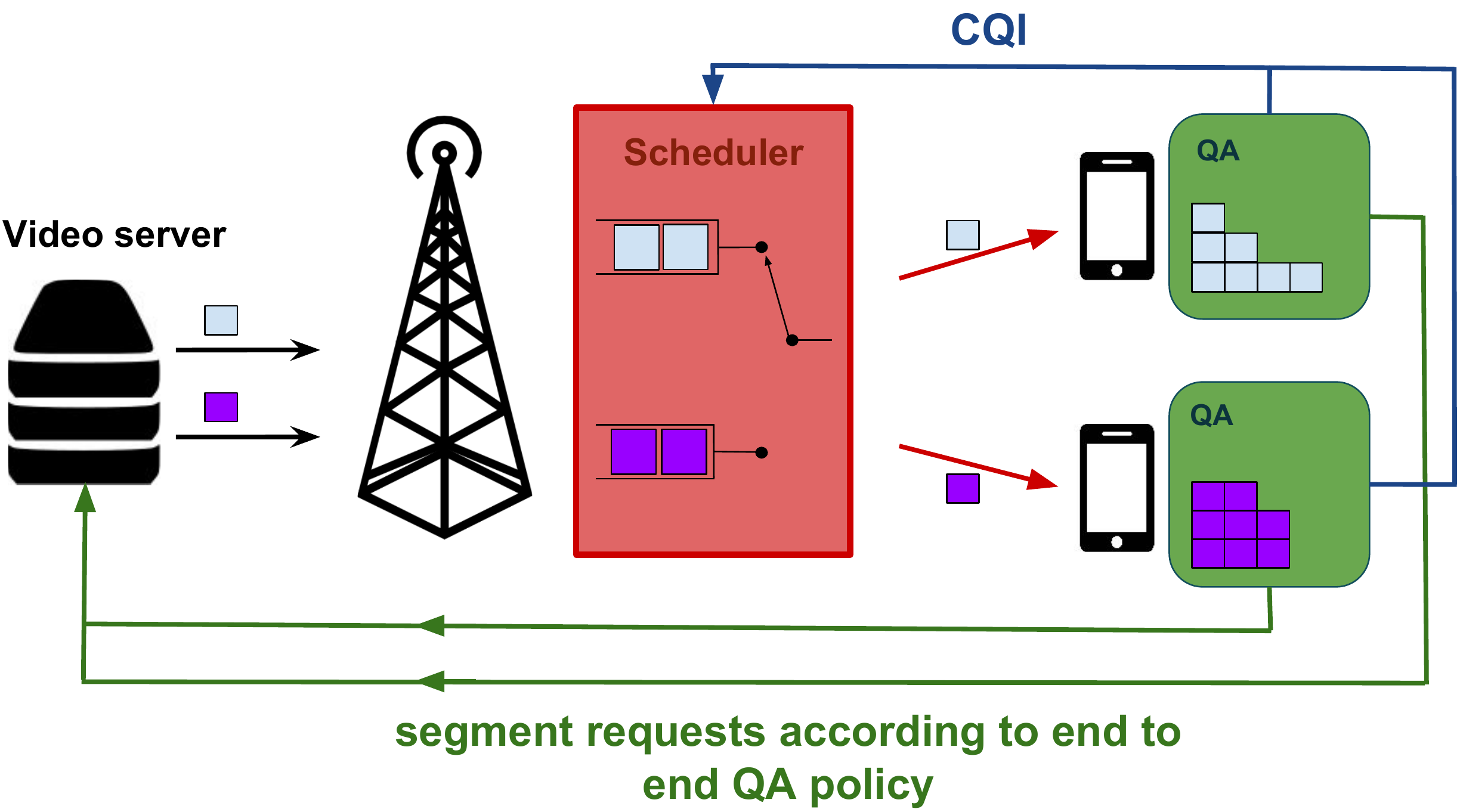}
             \caption{Separate QA and scheduling.}\label{fig:sys_sep}
        \end{subfigure}\qquad
\caption{System architecture for two SVC wireless delivery schemes.}
\label{fig:sys}         
\end{figure}

For this purpose, we first formulate the QA-adaptive scheme and the joint optimization using the concept of Restless Bandits (RB) \cite{whittle1988restless} and Multi-user Semi Markov Decision Process (MUSMDP), respectively, in order to quantify the loss in performance incurred by diverging from the jointly optimal scheme. We then develop a heuristic algorithm that perform scheduling given that the base station is aware of the QA used by the end users (QA-Aware). Furthermore, by analyzing the behavior of the scheduler for different QA schemes, we devise an even simpler scheduling heuristic in which the base station is blind to any of the users' QA (QA-Blind). 

RB is a powerful tool for optimizing sequential decision making based on forward induction. It is represented by a set of slot machines (one-armed bandits), where at any time slot a fixed number of them can be operated to receive a reward. The goal is to schedule the bandits such that their long term sum reward is maximized. This is a generalization of the traditional Multi-Armed Bandit (MAB) problem in which bandits that are not chosen in a slot do not change state and offer no reward in that slot \cite{gittins1979bandit}.

The remainder of the paper is organized as follows. In Section \ref{sec:rel_work}, we provide a summary of relevant research on the topic. Section \ref{sec:sys_model} describes the system model which is followed by the RB formulation of the problem in Section \ref{sec:prob_form}. The heuristic algorithms are presented in Section \ref{sec:alg}. Section \ref{sec:sim} and \ref{sec:implementation} contain the simulation and implementation results, respectively. Finally, Section \ref{sec:con} concludes the paper.

%% file: related_work.tex
\section{Related Work}\label{sec:rel_work}

The additional flexibility added by the multi-layer structure of SVC has triggered a substantial body of research focusing on optimizing single user QA. For this purpose, various approaches are used ranging from dynamic programming \cite{hosseini2015not,andelin2012quality,xiang2012adaptive,otwani2015optimal}, heuristics \cite{kim2005optimal,kuschnig2010evaluation,sieber2013implementation}, and experimental methods \cite{muller2012using,huang2015buffer}. 

Many others have investigated the joint problem of QA and scheduling for scalable video in wireless networks \cite{freris2013distortion,freris2010resource,ji2009scheduling,khan2016qoe,zhao2015qoe,talebi2011quasi,zhang2010cross,xiao2015optimal,fu2009systematic,cicalo2014distortion}. Among these papers, some have deployed network utility maximization techniques for solving the joint problem \cite{freris2013distortion,freris2010resource,talebi2011quasi,zhao2015qoe,cicalo2014distortion}. In \cite{ji2009scheduling}, a gradient based method is used in which, in every time slot, the base station solves a weighted rate maximization problem to update the gradient. The majority of these schemes are myopic and obtain optimality in a real time fashion which makes them suitable for live streaming events. 

The authors of \cite{xiao2015optimal,fu2009systematic}, model the problem first as a Multi-user Markov Decision Process and solve it using an iterative sub-gradient method. Unlike the previous papers, the proposed schemes are foresighted, i.e., the effect of each decision on future actions is taken into account. However, they require complex iterative computations in every time slot. In all the above papers, scheduling and QA are jointly optimized and therefore suffer the shortcomings discussed in Section \ref{sec:intro}. 

Other papers have proposed simple collaboration mechanisms between content providers and network operators with the goal of improving existing QA schemes for DASH\cite{zahran2017sap,chen2013scheduling,georgopoulos2013towards,petrangeli2016qoe}. The main argument in these papers is that current QA mechanisms that fully rely on client based adaptation fail to deliver acceptable performance in terms of fairness, stability, and resource utilization. By providing network assistance through the exchange of system statistics between the network and the client, the QA policies can be improved. In \cite{georgopoulos2013towards}, an OpenFlow assisted control plane orchestrates this functionality. An in-network system of coordination proxies for facilitating resource sharing among clients is proposed in \cite{petrangeli2016qoe}. The authors of \cite{zahran2017sap} develop a scheme that leverages both network and client state information to optimize the pacing of different video flows. In \cite{chen2013scheduling}, the bitrate of each requested stream is throttled to a certain range and a proportional fair scheduler shares the resources among the streams. In our work, we deploy a similar collaboration mechanism between the network and the content provider but for the purpose of optimizing the scheduling policy given that users may deploy any arbitrary QA.

To the best of our knowledge, no prior work has considered optimal QA-adaptive scheduling for scalable video in wireless networks. Furthermore, our proposed solution aims at optimizing the scheduling procedure in a foresighted manner and is suitable for video on demand, where buffering of content is possible.

%% file: system_model.tex
\section{System Model}\label{sec:sys_model}
In this section, we start by describing the network and video models used in the formulation. We then present a matrix representation to model arbitrary QA.

\subsection{Network Model:}
The network under consideration consists of $N$ users from the set $\mathcal N = \{1,2,\cdots,N\}$ and a base station. The total bandwidth is denoted by $W_{tot}$ and is divided into $M$ equal subchannels, as in OFDMA. At each time slot, the base station chooses $M$ ($M \leq N$) users and allocates time-frequency resources to each of them. We assume that for each user $n$, the channel has flat fading and the capacity of each subchannel follows a Markov chain with transition matrix $\mathbf{C}_n$, where $C_{n,i,j} = P(c_{n,t+1} = j|c_{n,t} = i)$. The states of this Markov chain $c_{n,t}$ are taken from a finite set $\mathcal{C}$ and represent the maximum achievable data rate per subchannel, which is a function of the available modulation and coding schemes in the network. We also assume that the channel variation is slow enough so that the download rate remains constant over one time slot. 

\subsection{Video Model:}\label{sec:video_model}
The users are streaming scalable video, each encoded into equal length segments of $\tau_{seg}$ seconds. The segments are encoded into $L$ quality layers. Throughout the paper, we refer to each individual layer of a segment as a sub-segment. We assume that sub-segments of the same layer are of equal rate and the layer rates are denoted by $\mathcal{Q} = \{q_1,\cdots,q_L\}$. At each time slot, users receive rewards based on the quality of the video segments that are played back in that time slot. As measure for QoE, we use a reward function that maps the rate of the video that is played back to the perceived quality as follows \cite{hu2012qoe}:
\begin{equation}\label{eq:qoe}
R =
\begin{cases}
e^{-\phi\left(\frac{R_p}{R_{max}}\right)^{-\theta}+\phi}, &\quad\text{no re-buffering} \\
r_{pen}, &\quad\text{re-buffering}   
\end{cases}
\end{equation}
where $R_p$ is the rate of the video that is being played back and $R_{max} = \sum_{l=1}^Lq_l$ is the maximum rate of that segment when all layers are present. The constants $\phi$ and $\theta$ are video-specific parameters of the quality model, and it is shown in \cite{hu2012qoe} that after averaging over numerous video sequences, their values is equal to $0.16$ and $0.66$, respectively. If the playback header reaches a segment for which the base layer is not delivered, playback stalls and re-buffering occurs. In order to account for this in the reward function, we assign a penalty for all instances of re-buffering denoted by $r_{pen}$, with a value depending on the sensitivity to re-buffering. By setting the value of $r_{pen}$, we implicitly determine our desired delay-quality trade-off. Needless to say, in order to penalize re-buffering, $r_{pen}$ should be set to a value less than the reward obtained by only having the base layer. The lower the value of $r_{pen}$, the higher the penalty. The main purpose of using adaptive video instead of constant rate video is the ability to decrease the quality of the video whenever there is risk of re-buffering. Hence, we suggest a low value for $r_{pen}$ throughout our simulation study in order to avoid re-buffering as much as possible.

\subsection{Quality Adaptation:}

In order to model delivery scenarios with arbitrary QA, we develop a matrix representation of the end user buffer and call it the \emph{policy matrix}. In this section, we describe how the policy matrix for each QA is derived. In Section \ref{sec:prob_form}, the policy matrix is used for the formulation of the optimization problem.  

We define the policy matrix $\mathbf P^{\pi_n}(c_{n,t})$ as a binary transition matrix representing the QA policy $\pi_n$ that is applied when user $n$ is in channel state $c_{n,t} \in \mathcal{C}$ at time $t$. Assuming the policy to be stationary, we can drop the time index from now on. For a buffer limit of $b_{max}$, the policy matrix determines all possible sub-segment deliveries that are allowed by policy $\pi_n$ in channel state $c_n$ in one time slot. Since each layer can have any number of sub-segments between 0 and $b_{max}$, the size of the policy matrix is $(b_{max}+1)^L\times (b_{max}+1)^L$. Each row of this matrix represents a particular buffer state at any time slot prior to selecting the next sub-segments to deliver, and each column represents the state of the buffer right after policy $\pi_n$ is applied. Hence, if the element in the $i^{th}$ row and $j^{th}$ column of $\mathbf{P}(c_n)$ is 1, it means that the policy chooses to download those sub-segments for which the buffer state changes from $i$ to $j$. 

Figure \ref{fig:buf} illustrates the concept of policy matrix with a simple example. It shows an end user buffer streaming a video that is encoded into a base and two enhancement layers. Suppose that under the current channel conditions, the user can receive one sub-segment in the current time slot ($c=1$Mbps). The current buffer state is denoted by $i = (6,4,1)$, showing the number of sub-segments per layer. Assume that the next sub-segment to be requested is from the second enhancement layer. However, since the slot is one second, one segment of the video will be played back and the final state of the buffer will be $j = (5,3,1)$. Therefore, the $i^{th}$ row of the policy matrix is constructed as follows (other rows are constructed in a similar fashion):

\begin{equation}\label{eq:pol_mat}
\mathbf{P}(c_n)_{i,j}=
\begin{cases}
1, &\quad\text{if}~j=(5,3,1) \\
0, &\quad\text{otherwise}   
\end{cases}
\end{equation}

With this technique, any arbitrary QA mechanism can be modeled as a set of policy matrices, each representing a particular channel state. For the remainder of our analysis, we consider three different QA policies:
\begin{enumerate}
\item \emph{Diagonal Buffer Policy (DBP)}: Results from existing research \cite{hosseini2015not, andelin2012quality} suggest that it is optimal to pre-fetch lower layers first, and fill higher layers after. In this policy, which we call the \emph{diagonal policy}, the user starts pre-fetching sub-segments from the lowest layer until the difference between the sub-segments of that layer and the one above reaches a certain \emph{pre-fetch threshold}, at which point it switches to the layer above, and this continues for all layers. This way, the difference between the buffer occupancy for each layer with the subsequent upper layer is kept at the fixed pre-fetch threshold. The policy depicted in Figure \ref{fig:buf} is an example of the diagonal policy. It should be noted that each two neighboring layers can have different pre-fetch thresholds depending on their respective segment sizes, the additional video quality they provide and design preferences. 
\item \emph{Channel Based Policy (CBP)}: In this scheme, users conservatively request more base layers whenever they are in bad channel conditions and gradually become more aggressive and request more enhancement layers as the channel condition improves \cite{kuschnig2010evaluation}. 
\item \emph{Base layer Priority Policy (BPP)}: In this scheme, base layer sub-segments are requested while buffer occupancy is low. After buffer is filled beyond a certain limit, the policy switched to full quality segments. This method has been proposed for single layered DASH video delivery \cite{joseph2014nova}. 
\end{enumerate}
There are many different ways to design a CBP or BPP policy. In Section \ref{sec:prob_form} we describe the particular CBP and BPP policies we used for the simulations. 

\begin{figure}[t!]
\includegraphics[width = 0.47\textwidth]{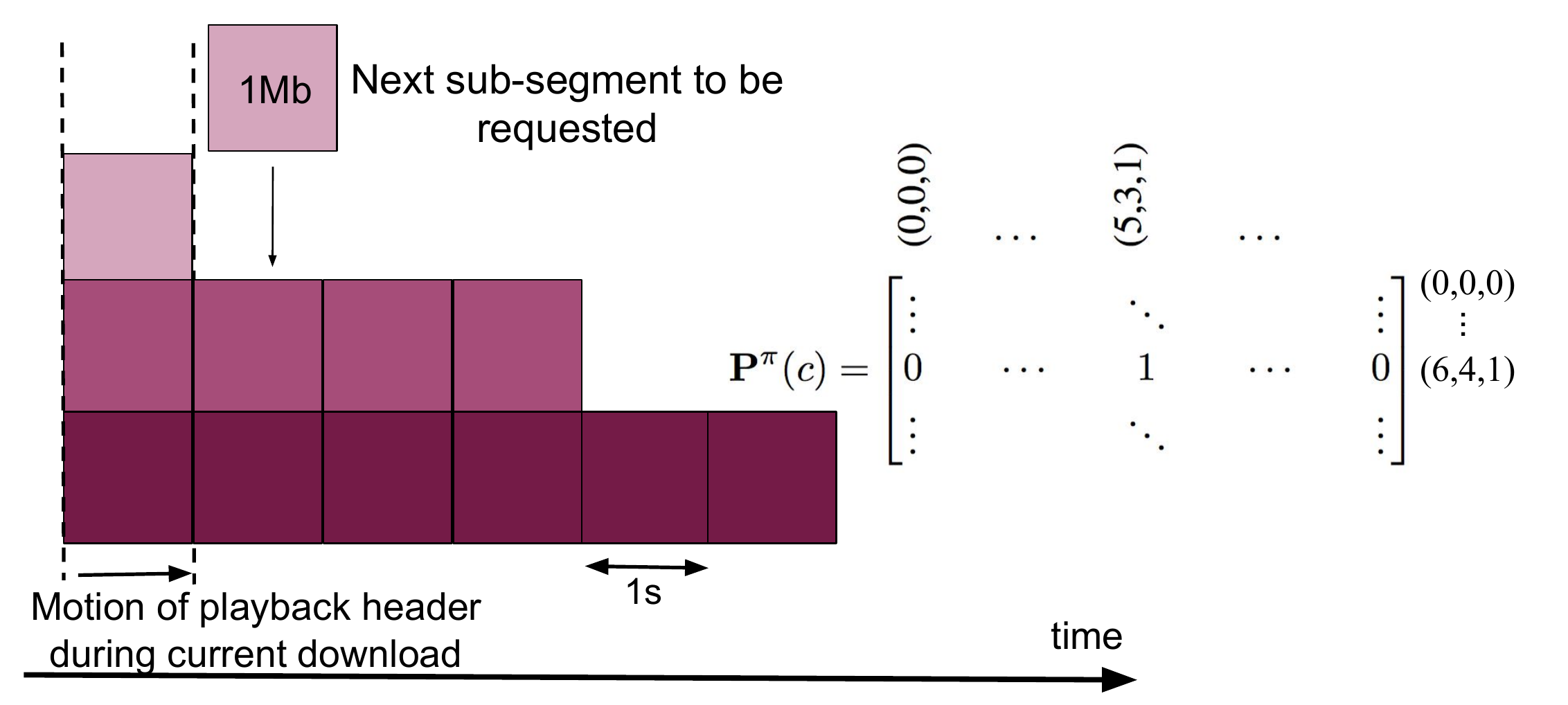}
\caption{Policy matrix and corresponding receiver buffer under DBP with pre-fetch threshold equal to two. The QA policy chooses the next sub-segment from the top enhancement layer.}\label{fig:buf}
\end{figure}

%% file: formulation.tex
\section{Problem Formulation}\label{sec:prob_form}
In this section, we first formulate the QA-adaptive scheduling as a Restless Bandit (RB). In order to compare the optimal solution of this formulation with a jointly optimal QA and scheduling scheme, we also formulate the latter using a MUSMDP. 

\subsection{QA-adaptive scheduling}

We assume that similar to Figure \ref{fig:sys_sep}, QA is determined by the content provider. The scheduler takes the QA of each user as input prior to the start of the streaming process and optimizes the schedule accordingly. Each time the previously requested segments are delivered, users request new segments from the video server. They also send Channel Quality Indicator (CQI) messages to the base station every time slot. The video server sends the requested segments to the base station where they are buffered and scheduled for delivery. 

In order to formulate the RB problem, we first model the state space of any user $n$, denoted by $\mathcal{S}_n$. $\mathcal{S}_n$ is defined as the combination of the instantaneous channel state $c_n$ and the current state of the buffer. We define the state of the buffer as a vector $\mathbf {b_n}$ representing the number of sub-segments the user has currently stored in the buffer for each layer, i.e., $\mathbf{b_n} = (b_{n,l})_{1:L}$. Therefore, the state space can be represented as $\mathcal{S}_n=\{(c_n,\mathbf{b_n})|c_n\in\mathcal{C},b_{n,l}\in\{0,\cdots,b_{max}\}\}$ and is of size $|\mathcal{C}|(1+b_{max})^L$. At each time slot $k$, the policy taken by the scheduler is in the form of an action vector $\mathbf{a_k} = (a_{n,k})_{1:N}$ of size $N$, where $a_{n,k}$ is set to one for scheduled users and zero for the rest. Hence, each user can be modeled as a bandit that is either operated in a slot or not. Since the scheduler does not control the users action due to the arbitrary QA, the bandits are uncontrolled and therefore, satisfy all necessary conditions for RB \cite{gittins1979bandit}. 

The transition from one state to the next depends on the channel transition matrix and the policy matrix as well as if the user was scheduled (active) in that time slot or not (passive). The structure of the transition matrix is similar to the policy matrix with the difference that here we also include the instantaneous channel state. We define two state transition matrices for the active and passive users and denote them as $\mathbf{H_n^1}$ and $\mathbf{H_n^0}$, respectively. In the passive case, the user cannot request any new sub-segments and can only play back the existing segments in the buffer. Therefore, the passive policy matrix $\mathbf{P_n^0}$ indicates transitions for which the occupancy of each layer is decremented by one. If no base layer is left in the buffer, no playback is possible, and therefore, no change occurs in the state of the buffer. We count this as an instance of re-buffering. The same procedure is followed for all channel states and we can write the state transition matrix $\mathbf{H_n^0}$ as follows:
\begin{equation}\label{eq:h0}
\mathbf{H_n^0} = \mathbf{C_n}\otimes\mathbf{P_n^0},
\end{equation}
where $\otimes$ is the Kronecker product. For the active state transition matrix $\mathbf{H_n^1}$, we need to create the policy matrices for all channel states ($\mathbf{P}^{\pi_n}(c_n), c_n\in\mathcal{C}$) since, depending on the available data rate, a different number of sub-segments can be delivered in every time slot. After determining the policy matrices for all channel states, the active state transition matrix can be derived as:

\begin{eqnarray}\label{eq:h1}
\mathbf{H_n^1} = 
\begin{bmatrix}
C_{n,1,1}\mathbf{P}^{\pi_n}(1)&\cdots&C_{n,1,|\mathcal{C}|}\mathbf{P}^{\pi_n}(1)\\
\vdots&\ddots&\vdots\\
C_{n,|\mathcal{C}|,1}\mathbf{P}^{\pi_n}(|\mathcal{C}|)&\cdots&C_{n,|\mathcal{C}|,|\mathcal{C}|}\mathbf{P}^{\pi_n}(|\mathcal{C}|)
\end{bmatrix}
\end{eqnarray}

The objective function is the expected discounted sum of rewards received by the users throughout the streaming process and is expressed as:
\begin{equation}\label{eq:objrb}
\max_u\mathbb{E}_u\left[\sum_{k=0}^\infty\sum_{n \in \mathcal{N}}R_{s_{n,k}}^{a_{n,k}}\beta^k\right],
\end{equation} 
where $\beta$ is the discount factor ($0 < \beta < 1$) and $R_{s_{n,k}}^{a_{n,k}}$ is the reward received by user $n$ if it is in state $s_{n,k}$ in time slot $k$ and chooses action $a_{n,k}$, and is calculated according to (\ref{eq:qoe}). The goal of the optimal scheduler is to determine the optimal policy $u$ in such a way that the expected sum of received rewards is maximized with respect to the resource constraint bounding the number of active users at each time slot to $M$. We assume that segments that are downloaded in each time slot cannot be played back in the same slot. This is also the case for real video delivery in which after a segment is received, it takes some time for decoding and processing before it becomes available for playback. With this assumption, the immediate reward of a user is independent of the immediate action of the scheduler, and we can write $R_{s_{n,k}}^0 = R_{s_{n,k}}^1= R_{s_{n,k}}$.

Next, we define for every user $n$ and scheduling policy $u$, the performance measures $x_{s_n}^a(u)$, where $a$ is either zero or one for the passive and active case, respectively. These performance measures are then defined as follows:
\begin{equation}\label{eq:perf_measure}
x_{s_n}^a(u) = \mathbb{E}_u\left[\sum_{k = 0}^\infty I^a_{s_n}(k)\beta^k\right],
\end{equation}
where
\begin{equation}\label{eq:indicator}
    I^a_{s_n}(k)=\left\{
                \begin{array}{ll}
                  1& \text{If in slot $k$, user $n$ is in state $s_n$}\\
                  & ~\text{and is assigned action $a$}\\
                  0 & \text{otherwise}
                  \end{array}
              \right.
\end{equation}
Essentially, $x_{s_n}^a(u)$ is the expected discounted amount of time that policy $u$ assigns action $a$ to user $n$ whenever the user is in state $s_n$. It is proved in \cite{d1960probleme} that the set of all Markovian policies for user $n$ following the transition matrix $\mathbf{H}^a_n$, forms a polytope which can be represented as follows:
\begin{eqnarray}\label{eq:polytope}
\mathcal{Q}_n &=& \Bigg\{\mathbf{x_n}\in\mathcal R_+^{|\mathcal{S}_n\times\{0,1\}|} \Bigg |x_{j_n}^0+x_{j_n}^1 \nonumber\\ 
 &=& \alpha_{j_n}+\beta\sum_{i_n\in\mathcal{S}_n}\sum_{a\in\{0,1\}}h_{i_nj_n}^ax_{i_n}^a,j_n\in\mathcal{S}_n\Bigg\},
\end{eqnarray}
where $\alpha_{j_n}$ represents the probability of $j_n$ being the initial state for user $n$.

Consequently, the RB can be formulated as the following linear program:
\begin{eqnarray}\label{eq:opt_1}
&\max_\mathbf{x} &\sum_{n\in\mathcal N}\sum_{s_n\in\mathcal S_n}\sum_{a\in\{0,1\}}R_{s_n}x_{s_n}^a \label{eq:objective}\\\nonumber
&&\text{subject to:} \\\label{eq:depinoux}
&&\mathbf{x_n}\in\mathcal{Q}_n, ~~n\in\mathcal N \\\label{eq:whittle}
&&\sum_{n \in \mathcal N}\sum_{s_n\in\mathcal S_n}x_{s_n}^1 = \frac{M}{1-\beta},
\end{eqnarray}
where $\mathbf{x} = (\mathbf{x}_1,\cdots,\mathbf{x}_N)$ and $\mathbf{x}_i = (x_{s_i}^0,x_{s_i}^1)_{s_i\in\mathcal{S}_i}$. The objective function is derived by simply replacing the performance measure from (\ref{eq:perf_measure}) into the original objective function (\ref{eq:objrb}). It should be noted that since the quality adaptation is pre-determined, the users are modeled as uncontrolled agents, i.e., the action space only determines if the user is active or passive without specifying what users do in the active mode. This is by definition the classic Restless Bandit (RB) problem including a set of agents (bandits) where a fixed number of them are activated in every time slot. All bandits, whether active or not, change state and receive a reward for the next slot. 

Although, based on our network model, the number of active users per slot is kept at a constant $M$, RB fixes the average number of active users per slot instead (see the resource constraint (\ref{eq:whittle})). According to RB theory \cite{weber1990index}, if the size and capacity of the network grow infinitely large ( $N,M \rightarrow \infty$) while $\frac{M}{N}$ remains fixed, the solution to RB asymptotically converges to the case with a constant number of users per slot. Therefore, for a fixed $\frac{M}{N}$, the larger the network, the closer RB will get to our desired solution.

We can simplify the problem for the cases in which all users are \emph{homogeneous}, i.e., they have the same buffer limit, use the same QA, and have similar video and channel characteristics. In this case, the polytope constraint (\ref{eq:polytope}) becomes identical for all users.
\begin{theorem}\label{prop:simpl}
If the users in a network are homogeneous, an allocation policy that results in equal active and passive service time in each state for every user is an optimal allocation policy for RB. In other words, $\mathbf{x}_i^0 = \mathbf{x}_j^0$ and $\mathbf{x}_i^1 = \mathbf{x}_j^1, \forall i,j\in\mathcal{N}$ in the optimal point, where $\mathbf{x}_i^0 = (x_{s_i}^0)_{s_i\in\mathcal{S}_i} $ and $\mathbf{x}_i^1 = (x_{s_i}^1)_{s_i\in\mathcal{S}_i} $ . 
\end{theorem}
\begin{proof}
Refer to Appendix \ref{app_1}.
\qed
\end{proof}
\begin{corollary} \label{cor:groups}
If the homogeneity conditions hold, RB can be simplified to the following linear program: 
\begin{eqnarray}\label{eq:opt_2}
&\max_\mathbf{x}&\sum_{s\in\mathcal S}\sum_{a\in\{0,1\}}R_{s}x_{s}^a \\
&&\text{subject to:} \\
&&\mathbf{x}\in\mathcal{Q} \\
&&\sum_{s\in\mathcal S}x_{s}^1 = \frac{M}{N(1-\beta)}\label{eq:simp_res_const}
\end{eqnarray}
The number of variables in the above linear program is equal to $2|\mathcal S|$ and it has $|\mathcal S| + 1$ constraints. 
Hence, we can model a network consisting of multiple homogeneous users using the state space of a single user, and therefore significantly decrease the number of variables and constraints of RB. We can optimize the scheduling for heterogeneous (not homogeneous) users, by grouping them into multiple groups each comprising of homogeneous users, and including only one sample user per group in the optimization. In this case, there will be one polytope constraint for each group and the resource constraint changes to:
\begin{equation}
\sum_{g\in\{1,\cdots.G\}}N_{g}\sum_{s_g\in\mathcal S_g}x_{s_g}^1 = \frac{M}{1-\beta},\label{eq:simp_res_const_multclass}
\end{equation}
where $g$ is the index of the groups and $N_g$ and $\mathcal{S}_g$ represent the number of users and the state space of user in group $g$, respectively. The complexity of the problem, therefore only depends on the number of groups and not the number of users per group.
\end{corollary}

\subsection{Joint Optimal QA and Scheduling}

In this section, we formulate the joint problem of optimally requesting new segments and scheduling users using the same network and video models described earlier. Similar to the RB, every user is modeled as an independent agent that changes state in a Markovian fashion. In this case, the action space for each user $n$ contains the passive action as well as the index of the layer of the sub-segment to be downloaded in the active case and can be represented as $a_n = \{0,1,\cdots,L\}$. The key difference between this formulation and the RB is that previously, we assumed that actions are taken at every time slot, and that within each slot multiple sub-segments can be downloaded if the user is active. Now, we assume that if the user is active, each decision is made once the previous action is fully executed. Therefore, actions will have different durations depending on the state of the user and the problem becomes a MUSMDP. 

In order to properly formulate the MUSMDP, we define the duration of a time slot to be the duration of the shortest possible action $\tau_{slot} = \frac{\min_l q_l}{\max(\mathcal{C})}$, where $q_l$ is the size of layer $l \in\{1,\cdots,L\}$. Since the time slot duration is generally shorter than the duration of a segment, at every time slot we will have partially played back segments and the fraction of the playback needs to be included in the state space representation as well. Hence, the state space of user $n$ is denoted by $\mathcal{S}_n=\{(c_n,\mathbf{b_n},u_n)|c_n\in\mathcal{C},b_{n,l}\in\{0,\cdots,b_{max}\},u_n\in\{0,\cdots,\frac{\tau_{seg}}{\tau_{slot}}-1\}\}$, where $c_n$ and $\mathbf{b}_n$ are defined as before and $u_n$ is the number of time slots that have passed since the current segment started playing back. For simplicity of notation we assume that $\tau_{seg}$ is an integer multiple of $\tau_{slot}$. Figure \ref{fig:smdp} illustrates the state space under the new setting in a simple way. In this figure, $u$ is initially zero. After the next sub-segment is delivered, playback continues and the playback header points at $u = \Delta t$.

\begin{figure}[t!]
\centering
\includegraphics[width = 0.35\textwidth]{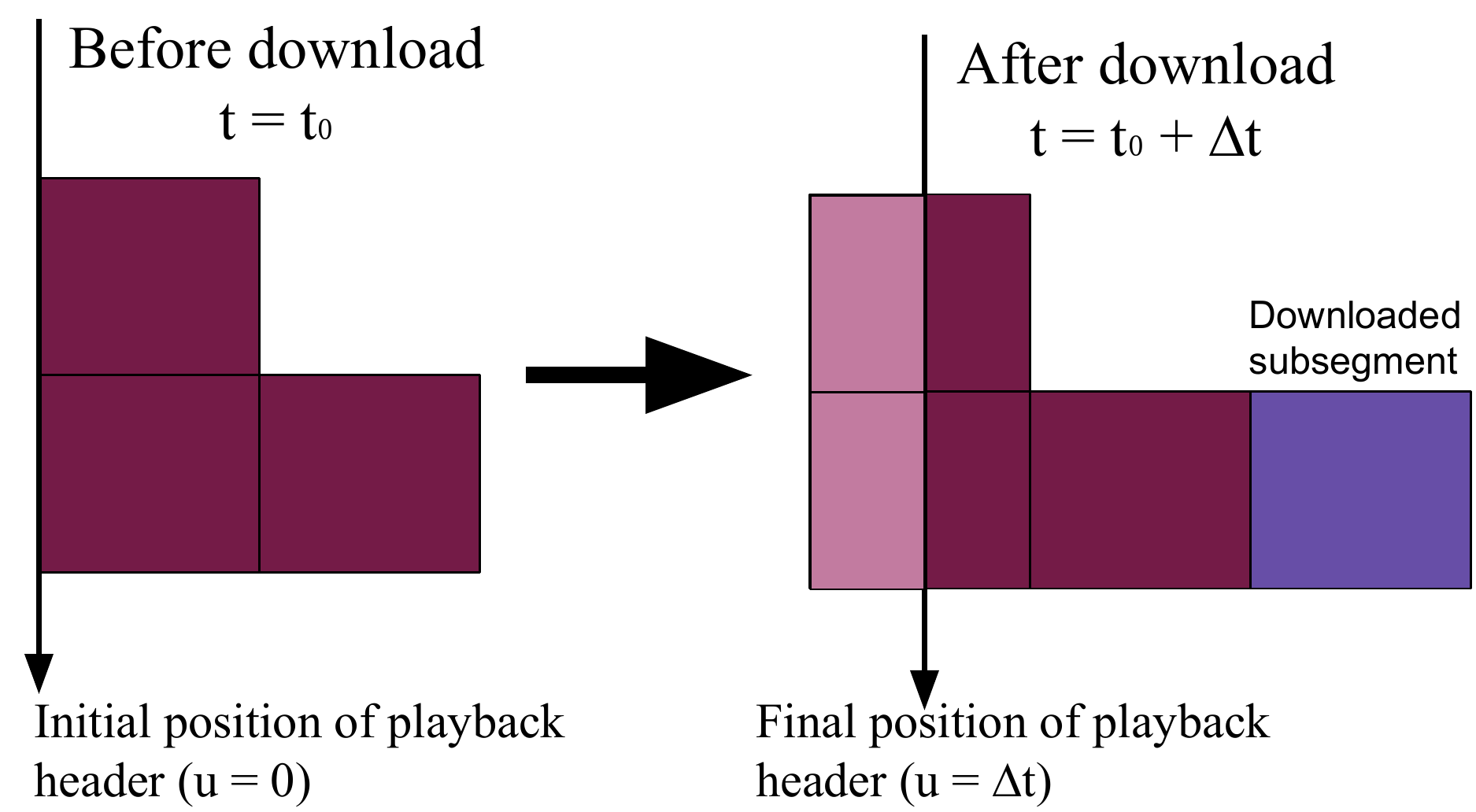}
\caption{Buffer evolution model for the MUSMDP. Decisions are made at unequal time intervals, hence, the fraction of playback for each segment ($u$) is included in the state space.}\label{fig:smdp}
\end{figure}

We can turn a discounted SMDP into a discounted MDP by modifying the transition probabilities and reward function \cite{puterman2014markov}. First, we need to develop policy matrices for each of the actions. For the passive action, we assume that for the duration of one time slot, the user plays back the video in the buffer without adding any segment to it. For the active cases, we generate one matrix per layer $l$ represented by $\mathbf{H}^l$. Since the channel transitions occur at every time slot, it is not known beforehand how many time slots each action takes. In order to determine the probability distribution of the duration of action $l$, we first define the random variable $\tau_{l}$ as the minimum number of time slots required to fully execute action $l$ given that the initial channel state is $c$ as follows:
\[
f^l_c(\tau_l) = P\left.\left(\sum_{k=1}^{\tau_l}c_k \geq \frac{q_l}{\tau_{slot}}\right| c_1 = c\right),
\]
where $c_k$ is the available rate in time slot $k$ and it transitions according to the channel matrix $\mathbf{C}$. By considering all possible trajectories of channel state transitions, we can determine the joint probability distribution of $\tau_l$ and the final channel state, given the initial channel state, as $f^l_{c}(t,j) = P\left(\tau_l = t ,c_{\tau_l + 1}=j| c_1 = c \right)$. 

In order to generate the policy matrix for user $n$ for the active cases, we consider the next state as the one in which $b_{n,l}$ is incremented by one and $u_n$ is incremented by the duration of the action. Therefore, the policy matrices $\mathbf{P}^l(\tau,c_n)$ is also a function of the number of time slots required for executing action $l$. Given that we apply discounting on a slot by slot basis, the policy matrix $\mathbf{H}^l_n$ can be written as:

\begin{eqnarray}\label{eq:hl}
\mathbf{H}_n^l = 
\begin{bmatrix}
\mathbf{V}^l_{1,1}&\cdots&\mathbf{V}^l_{1,|\mathcal{C}|}\\
\vdots&\ddots&\vdots\\
\mathbf{V}^l_{|\mathcal{C}|,1}&\cdots&\mathbf{V}^l_{|\mathcal{C}|,|\mathcal{C}|}
\end{bmatrix}
\end{eqnarray}

where:

\begin{equation}\label{eq:vl}
\mathbf{V}^l_{i,j} = \sum_{k=1}^\infty e^{-sk}f^l_{i}(k,j)\mathbf{P}^l(k,i),
\end{equation}

where $e^{-s}$ is the discount factor and $l = \{0,1,\cdots,L\}$. For the passive case where $l=0$, we have $f^0_{i,j}(k)$ equal to one for $k=1$ and zero for all higher values since we define the passive action duration to be one time slot.
  
Suppose that user $n$ makes its $m^{th}$ decision at time $\sigma_n^m$. Then, the objective function of the MUSMDP is the following:
\begin{equation}\label{eq:musmdp_new_objective}
\max_u \mathbb{E}_u \left(\sum_{n\in\mathcal{N}}\sum_{m=0}^\infty e^{-s\sigma_n^m} \sum_{k=0}^{\tau_l} e^{-sk}R^k_{s_n}\right), 
\end{equation}
where $e^{-s}$ is the discount factor, $\tau_l$ is the duration of action $l$ under policy $u$ and $R^k_{s_n}$ is the reward obtained by user $n$ in state $s_n$ after $k$ time slots have passed. The above expression can be turned into an equivalent discounted RB using the theorem below:
 
\begin{theorem}\label{theorem:musmdp}
The MUSMDP can also be formulated as a linear program as follows:
 \begin{eqnarray}\label{eq:opt_1}
&\max_\mathbf{y} &\sum_{n\in\mathcal N}\sum_{s_n\in\mathcal S_n}\sum_{a\in\{0,1,\cdots,L\}}\bar r_{s_n}^{a_n}y_{s_n}^{a_n} \label{eq:objective_smdp}\\\nonumber
&&\text{subject to:} \\\label{eq:depinoux_smdp}
&&\mathbf{y_n}\in\mathcal{P}_n, ~~n\in\mathcal N \\\label{eq:whittle_smdp}
&&\sum_{n \in \mathcal N}\sum_{a_n=1}^L\sum_{s_n\in\mathcal S_n}y_{s_n}^{a_n}\bar\tau_{s_n}^{a_n} = \frac{M}{1-e^{-s}}, \\\label{eq:resource_smdp}
\end{eqnarray}

where 
\[
y_{s_n}^{a_n}(u) = \mathbb{E}_u\left[\sum_{k = 0}^\infty I^{a_n}_{s_n}(k)e^{-sk}\right],
\]
and $I^{a_n}_{s_n}(k)$ is defined as (\ref{eq:indicator}) and $\bar{r}_{s_n}^{a_n} = \sum_{t=0}^\infty \sum_{c\in\mathcal{C}}\left(\sum_{k=0}^te^{-sk}R_{s_n}^k f_{s_n}^{a_n}(t,c)\right)$. Furthermore, the equivalent polytope constraint turns into the following:

\begin{eqnarray}\label{eq:polytope_smdp}
\mathcal{P}_n&=& \Bigg\{\mathbf{y_n}\in\mathcal R_+^{|\mathcal{S}_n\times\{0,1,\cdots,L\}|} \Bigg |\sum_{{a_n}\in\{0,1,\cdots,L\}} y_{j_n}^l \\ 
 &=& \alpha_{j_n}+\sum_{i_n\in\mathcal{S}_n}\sum_{{a_n}\in\{0,1,\cdots,L\}}h_{i_nj_n}^{a_n}y_{i_n}^{a_n},j_n\in\mathcal{S}_n\Bigg\}, \nonumber
\end{eqnarray}

Finally, $\bar\tau_{s_n}^{a_n}$ is the expected discounted duration of action $a_n$ if the user is in state $s_n$ when the action is taken and it can be written as $\bar\tau_{s_n}^{a_n} =\sum_{t=0}^\infty \sum_{c\in\mathcal{C}}\left(\sum_{k=0}^te^{-sk}f_{s_n}^{a_n}(t,c)\right)$.
\end{theorem}
\begin{proof}
Refer to Appendix \ref{app_2}.
\qed
\end{proof}

\subsection{Evaluation}
Now, we evaluate the optimal performance of each of the above scenarios in order to determine the loss incurred by abandoning the jointly optimal scheme for the more practical QA-adaptive method. In order to show this, we solve the above problems with identical system settings described in Table \ref{table:params}. 

\begin{table}[h] 
\centering
\begin{tabular}{|c|c|}
\hline
\textbf{Parameter} & \textbf{Value} \\
\hline
number of users & 20\\
\hline
subchannel bandwidth & 2MHz \\
\hline
channel states per subchannel & (1,2,5,10) Mbps\\ 
\hline
number of layers & 2\\
\hline
video rate per layer & (1,1) Mbps \\
\hline
buffer limit & 20s \\
\hline
segment duration & 1s \\
\hline
discount factor per second & 0.99 \\
\hline
$r_{pen}$ & 0 \\
\hline
\end{tabular}\caption{Parameters}\label{table:params}
\end{table}

\begin{figure*}[t]
	\captionsetup[subfigure][h]{twoside,margin={0.2cm,0.2cm}}
	\centering
	\begin{subfigure}[h]{0.3\textwidth}
             \includegraphics[height=1.6in]{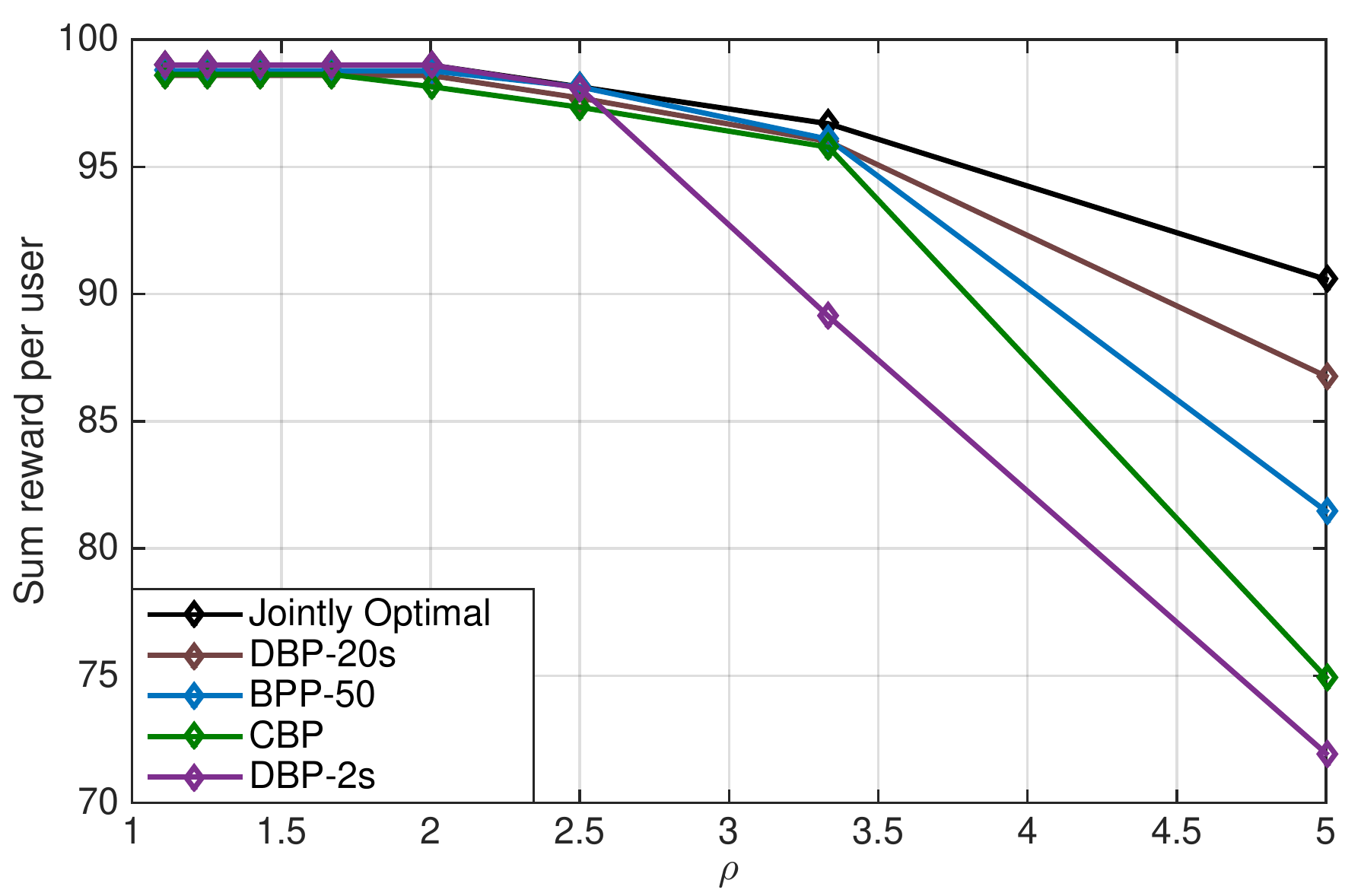}
             \caption{Average rate per subchannel: 4.5 Mbps}\label{fig:joint_comp_midchan}         
        \end{subfigure}\qquad 
        \begin{subfigure}[h]{0.3\textwidth}
             \includegraphics[height=1.6in]{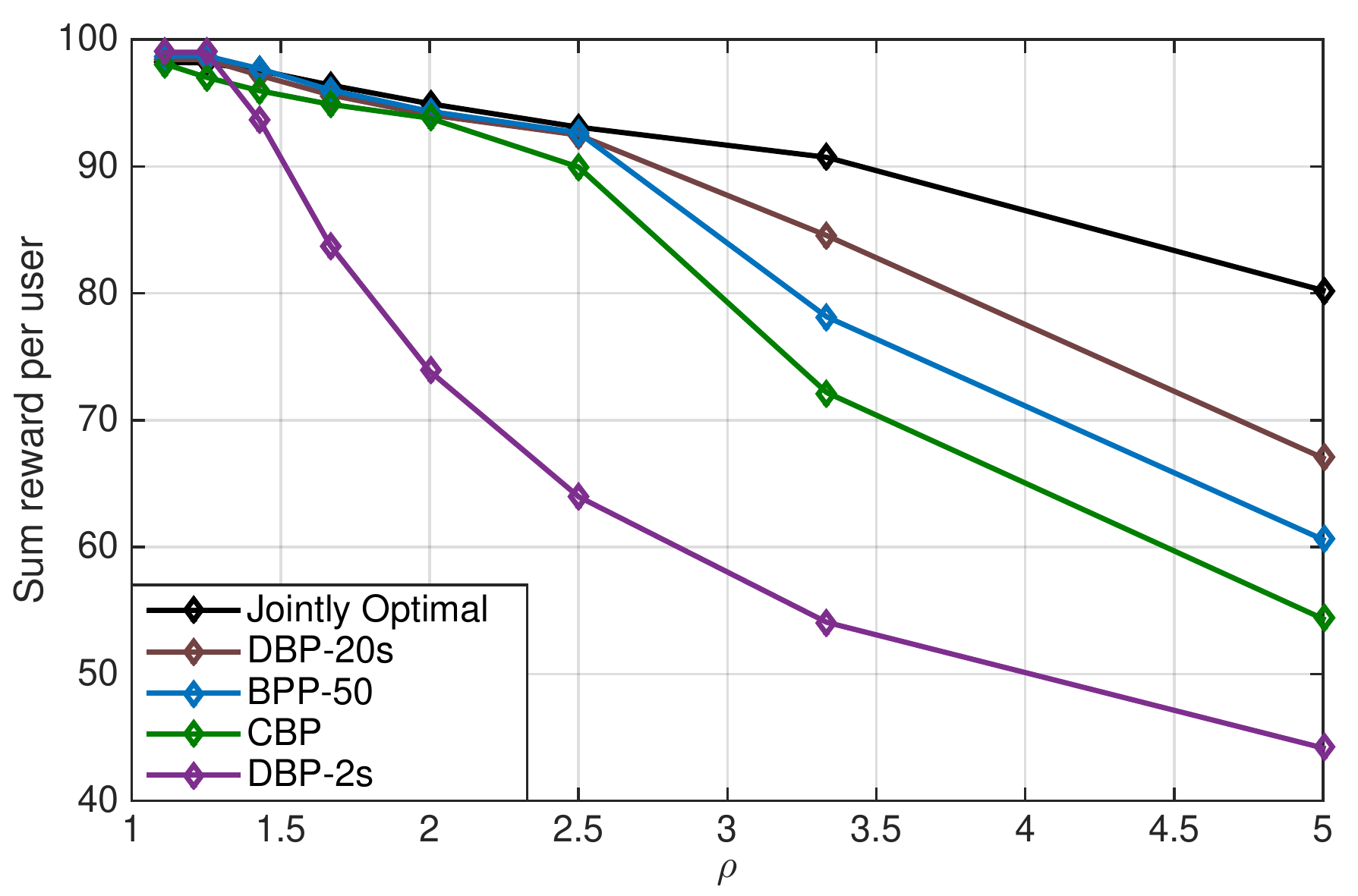}
             \caption{Average rate per subchannel: 2.5 Mbps}\label{fig:joint_comp_lowchan}
        \end{subfigure}\qquad
        \begin{subfigure}[h]{0.3\textwidth}
             \includegraphics[height=1.6in]{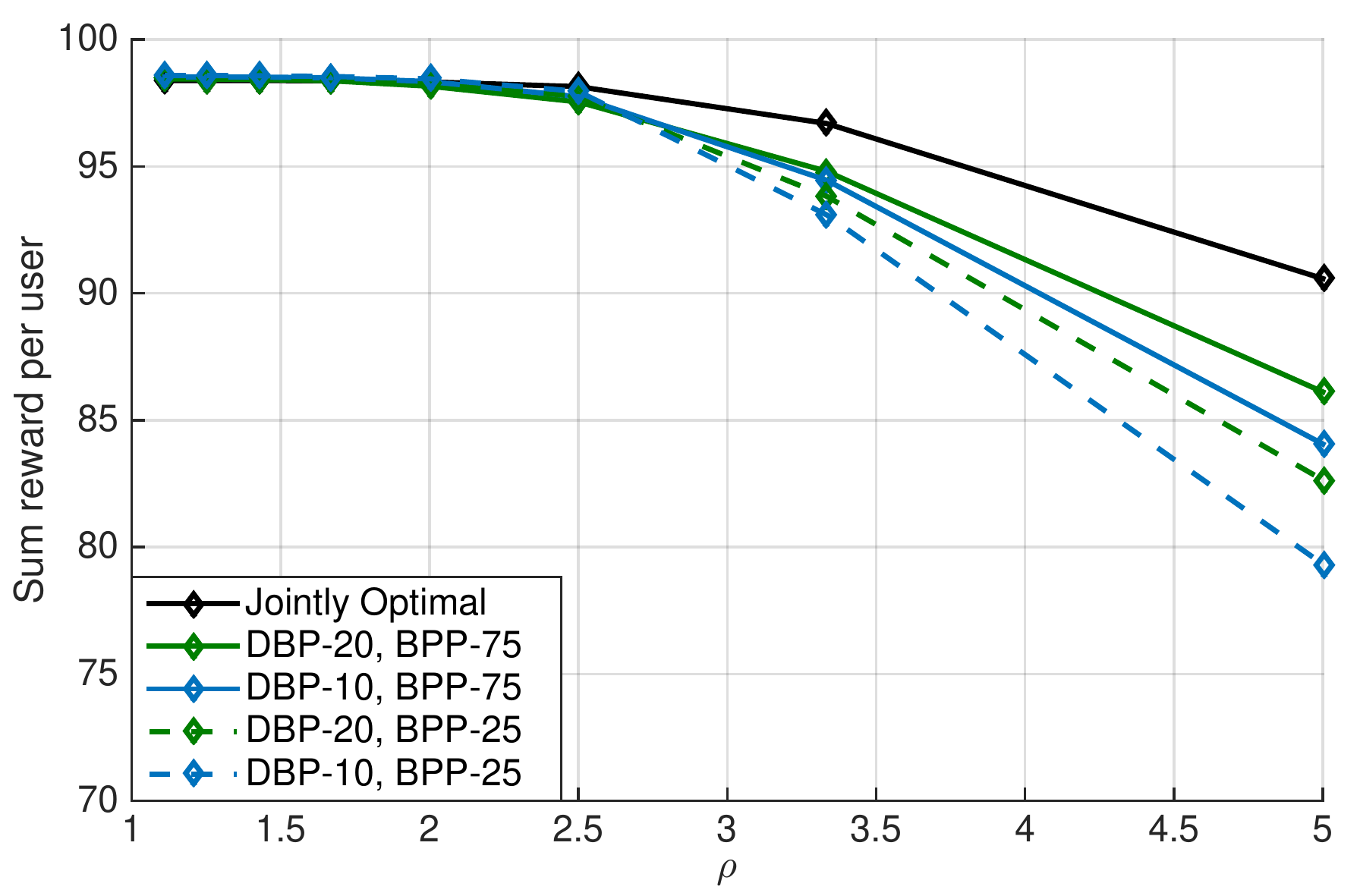}
             \caption{Average rate per subchannel: 4.5 Mbps}\label{fig:joint_comp_bpp}
        \end{subfigure}\qquad
\caption{Performance comparison between jointly optimal QA and scheduling and QA-adaptive scheduling for different QA schemes.}
\label{fig:joint_comp}         
\end{figure*}

Figure \ref{fig:joint_comp} shows the sum reward per user, which is the value of the objective function for the RB for each user. In these figures, the reward is plotted as a function of the load on the network, which we define as the average number of users that compete for one subchannel, $\rho = \frac{N}{M}$. In order to vary the load on the network, we vary the number of available subchannels from 4 to 18 in increments of 2. We perform this evaluation for two network settings, with average rates of 4.5 Mbps and 2.5 Mbps, respectively. The QA schemes used are different variations of BPP, CBP, and two DBP schemes. In the BPP-x scheme, base layer segments are requested until the occupancy reaches x\% of the buffer limit. After that, full quality segments are requested. For CBP, if the channel is in the two low rate states, the user only requests base layer segments. In the third channel state, two-thirds of the resources are spent on requesting base layers and the rest is reserved for enhancement layers. In the best state, only full quality segments are requested. The DBP-x QA policy represents the diagonal policy with a pre-fetch threshold of x seconds.

In Figures \ref{fig:joint_comp_midchan} and \ref{fig:joint_comp_lowchan}, the network is assumed to be homogeneous and all users in the network use the same QA. We observe that especially for low load scenarios, choosing a proper QA along with QA-adaptive scheduling will perform close to the jointly optimal case. For high load scenarios, the choice of QA becomes more important and using DBP with a large pre-fetch threshold that fills the buffer with base layer sub-segments up to the buffer limit and then starts downloading enhancement layers performs best, since by pre-fetching base layers, we lower the risk of re-buffering. Based on these results we argue that if an optimal QA-adaptive scheduler is used, an arbitrary QA can perform relatively close to the global optimum of the system. Figure \ref{fig:joint_comp_bpp} illustrates the performance of the QA-adaptive scheduling mechanism in a heterogeneous system where one half of the users deploy a DBP policy and the other half use a BPP policy. This figure is represents networks in which some content providers use SVC while others use single layered video. This model is useful because once content providers start offering adaptive video with SVC, they have to coexist with services that will still rely on single layered video. Our proposed QA-adaptive model is capable of devising scheduling policies for these mixed environments. Figure \ref{fig:joint_comp_bpp} shows that by using proper QA schemes, our scheduler performs within 85\% of the jointly optimal scheme.

In the next section, we describe the scheduler operation for the QA-adaptive scenario.

%% file: algorithm.tex
\section{Online Algorithm}\label{sec:alg}
The solution of RB gives the long term performance measures, representing the total discounted time each action is applied in each state without explicitly stating what action should be taken in each time slot. In this section, we propose two heuristic algorithms for the QA adaptive scheme that are based on ranking the states of the users in terms of the scheduling priority, based on the optimal solution of RB.  

The first algorithm is designed for the case in which the base station knows what particular QA is being used by each user beforehand, which we denote as \emph{QA Aware Scheduling}. The second algorithm is designed to further simplify this procedure and imitate the functionality of the QA Aware Scheduling without actually knowing the QA of each user. We call this policy \emph{QA Blind Scheduling}.

\subsection{QA Aware Scheduling Algorithm}

We start by defining the dual problem of RB as shown below. Without loss of generality, we show the dual of the simplified problem for the case with homogeneous users. 
\begin{eqnarray}\label{eq:opt_dual}
&\min_\mathbf{\lambda} &\sum_{s\in\mathcal S}\alpha_s\lambda_s + \frac{M}{N(1-\beta)}\lambda\label{eq:objective}\\
&&\text{subject to:} \\
&&\lambda_s-\beta\sum_{j\in\mathcal{S}}h^0_{s,j}\lambda_j\geq NR^0_s,~\forall s\in\mathcal{S},\\
&&\lambda_s-\beta\sum_{j\in\mathcal{S}}h^1_{s,j}\lambda_j+\lambda\geq NR^1_s,~\forall s\in\mathcal{S},
\end{eqnarray}
where the set of variables is denoted by $\mathbf{\lambda} = \{(\lambda_s)_{s\in\mathcal{S}},\lambda\}$. We define \emph{reduced cost coefficients} $\gamma^a_s$ for $\forall s \in \mathcal{S}, a \in \{0,1\}$ as follows:
\begin{equation}
\gamma^0_s = \lambda^*_s-\beta\sum_{j\in\mathcal{S}}h^0_{s,j}\lambda^*_j -  NR^0_s,
\end{equation}
\begin{equation}
\gamma^1_s = \lambda^*_s-\beta\sum_{j\in\mathcal{S}}h^1_{s,j}\lambda^*_j +\lambda-  NR^1_s,
\end{equation}
where $\lambda^*_s$ is the optimal value for $\lambda_s$. 

The reduced cost coefficient $\gamma^a_s$ represents the rate of decrease in the objective function in RB per unit increase in the variable $x^a_s$ \cite{murty1983linear}. For example, if in a particular time instant, two users are in states $s_1$ and $s_2$, respectively, such that $\gamma^1_{s_1} > \gamma^1_{s_2}$, scheduling the user in $s_2$ will cause less reduction in the objective function and should be prioritized over the other. By using this characteristic and the fact that due to complementary slackness, either $\gamma^a_s = 0$ or $x^a_s = 0$ for $\forall s,a$, we can derive a ranking scheme for all states as shown in Algorithm \ref{alg:ranking}. 

\begin{algorithm}
\caption{QA Aware Scheduling}
\label{alg:ranking}
\textbf{Step 0:} Solve RB and its dual and determine $x^{*0}_s$, $x^{*1}_s$, $\mathbf{\gamma}^{0}_s$ and $\mathbf{\gamma}^{1}_s$ $\forall s \in \mathcal{S}$.\\
\textbf{Step 1:} Set $Q_1 = \{s\in\mathcal{S}|x^{*1}_s > 0\}$.\\
\textbf{Step 2:} Sort $Q_1$ in descending order of $\mathbf{\gamma}^{0}_s$.\\
\textbf{Step 3:} Set $Q_0 = \{s\in\mathcal{S}|x^{*0}_s > 0\}$.\\
\textbf{Step 4:} Sort $Q_0$ in ascending order of $\mathbf{\gamma}^{1}_s$.\\
\textbf{Step 5:} Define ranking vector $Q = [Q_1,Q_0]$
\end{algorithm} 

It is easy to argue that states for which $x^{*1}_s$ is positive should have priority over those for which $x^{*1}$ is zero. Therefore, we prioritize states with $x^{*1}_s > 0$ and sort them in descending order of their respective $\mathbf{\gamma}^{0}_s$. Then, at the secondary level of priority, we take states with $x^{*0}_s > 0$ and sort them is ascending order of their $\mathbf{\gamma}^{0}_s$. Thereby, we have a priority list of all states,  where in every time slot, the $M$ users that appear highest in the list are scheduled. As a simple example, consider the case where a network has three users. At a particular time slot, the users are in states $s_1$, $s_2$, and $s_3$, respectively. Also, suppose that the value of the reduced cost coefficients and long term performance measures in this time slot are as follows:

\begin{table}[h] 
\centering
\begin{tabular}{|c|c|c|c|c|}
\hline
state&$x_s^1$&$\mathbf{\gamma}_s^1$&$x_s^0$&$\mathbf{\gamma}_s^0$ \\
\hline
$s_1$&1&0&0.2&0\\
\hline
$s_2$&0&10&0.1&0\\
\hline
$s_3$&0.5&0&0&20\\
\hline
\end{tabular}\caption{A sample instance of the QAA algorithm for three users. Note that the numbers are not derived from an actual experiments and serve only as am example.}\label{table:alg_example}
\end{table}

By following Algorithm \ref{alg:ranking} for this sample case, we first prioritize user 1 and 3 over user 2 because their respective $x_s^1$ is non-zero. Among users 1 and 3, we pick user 3 because it has a lower $\mathbf{\gamma}_s^0$. Therefore, the resulting ranking of users will become $3-1-2$. 

In order to implement QAA, the base station needs to know the channel matrix, the QA of each user and the video characteristics such as number of layers and segment size. Since segment sizes are not equal throughout the video, an estimate for average size can be used for the calculation. The RB linear program is solved along with its dual and all variables that are needed for QAA are determined prior to the start of the stream. However, the wireless channel is non-stationary and the channel matrix might change over time. Therefore, in order to have a more accurate estimate for the channel dynamics, the channel matrix should be updated at periodic intervals and the RB is recalculated in order to update the variables in Algorithm \ref{alg:ranking}. Frequently updating the variables will increase the accuracy of the algorithm as well as its computational complexity. It is experimentally shown in \cite{fund2013performance} that the channel matrix can be assumed to be stationary for a duration in the order of tens of seconds. Updating the RB variables should also be performed whenever a user enters or exits the network.

\subsection{QA Blind Scheduling Algorithm}

In this section, we derive a simple QA Blind heuristic policy for scheduling users without the base station knowing any of the mentioned system characteristics in the previous section. To that end, we start by studying the outcome of the RB problem for a variety of scenarios in order to find common trends.

\begin{figure}
	\centering
             \includegraphics[width=0.4\textwidth]{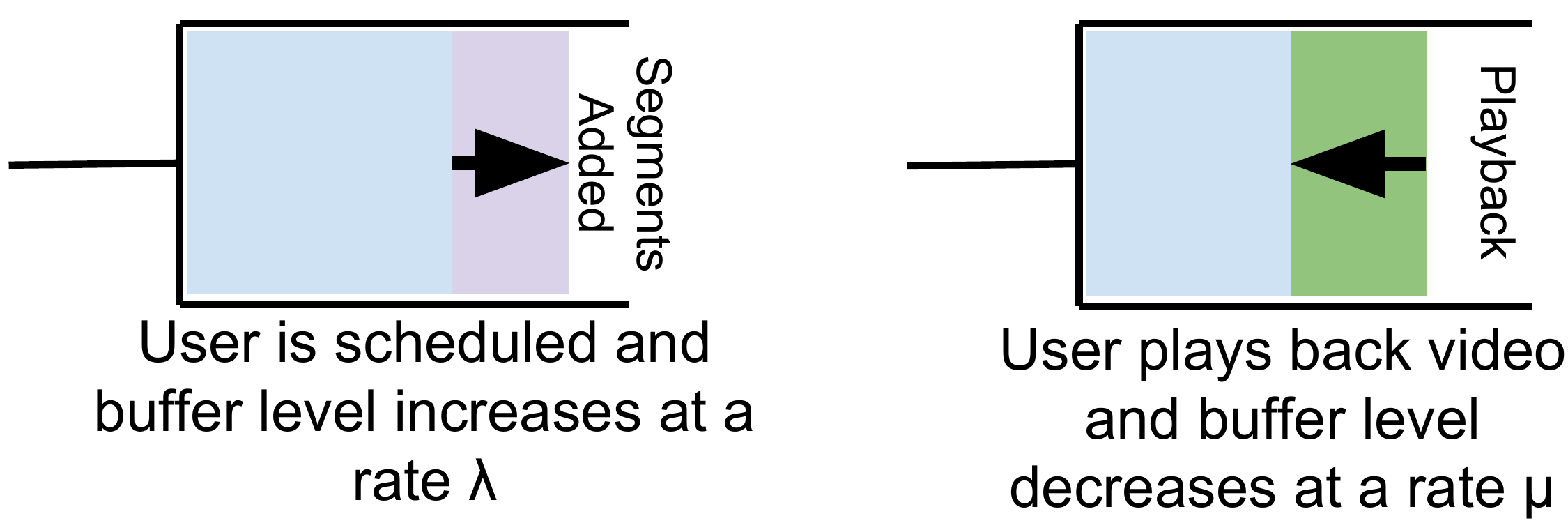}
             \caption{User buffer level increases with rate $\lambda$ when video data is delivered and decreases with rate $\mu$ due to continuous playback.}\label{fig:alg_buf}
\end{figure}

\begin{figure}
	\centering
             \includegraphics[width=0.4\textwidth]{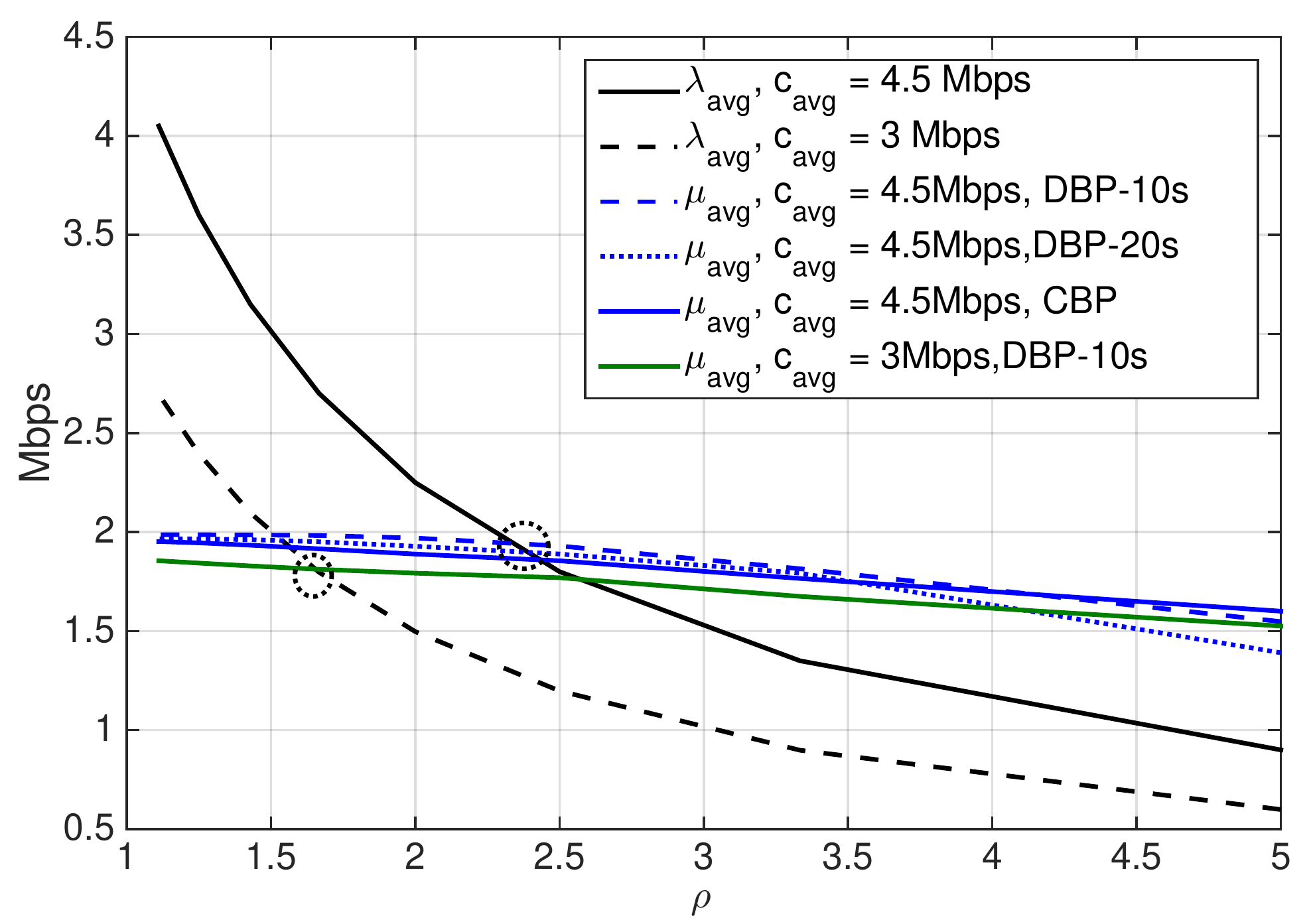}
             \caption{$\lambda_{avg}$ and $\mu_{avg}$ for different QA and network capacities as a function of load. The encircled points refer to the critical load $\rho^*$ for which the buffer level is constant on average.}\label{fig:alg_lambda_mu}
\end{figure}

\begin{figure*}[t]
	\captionsetup[subfigure][h]{twoside,margin={0cm,0cm}}
	\centering
        	\begin{subfigure}[h]{0.24\textwidth}
             \includegraphics[height=1.5in]{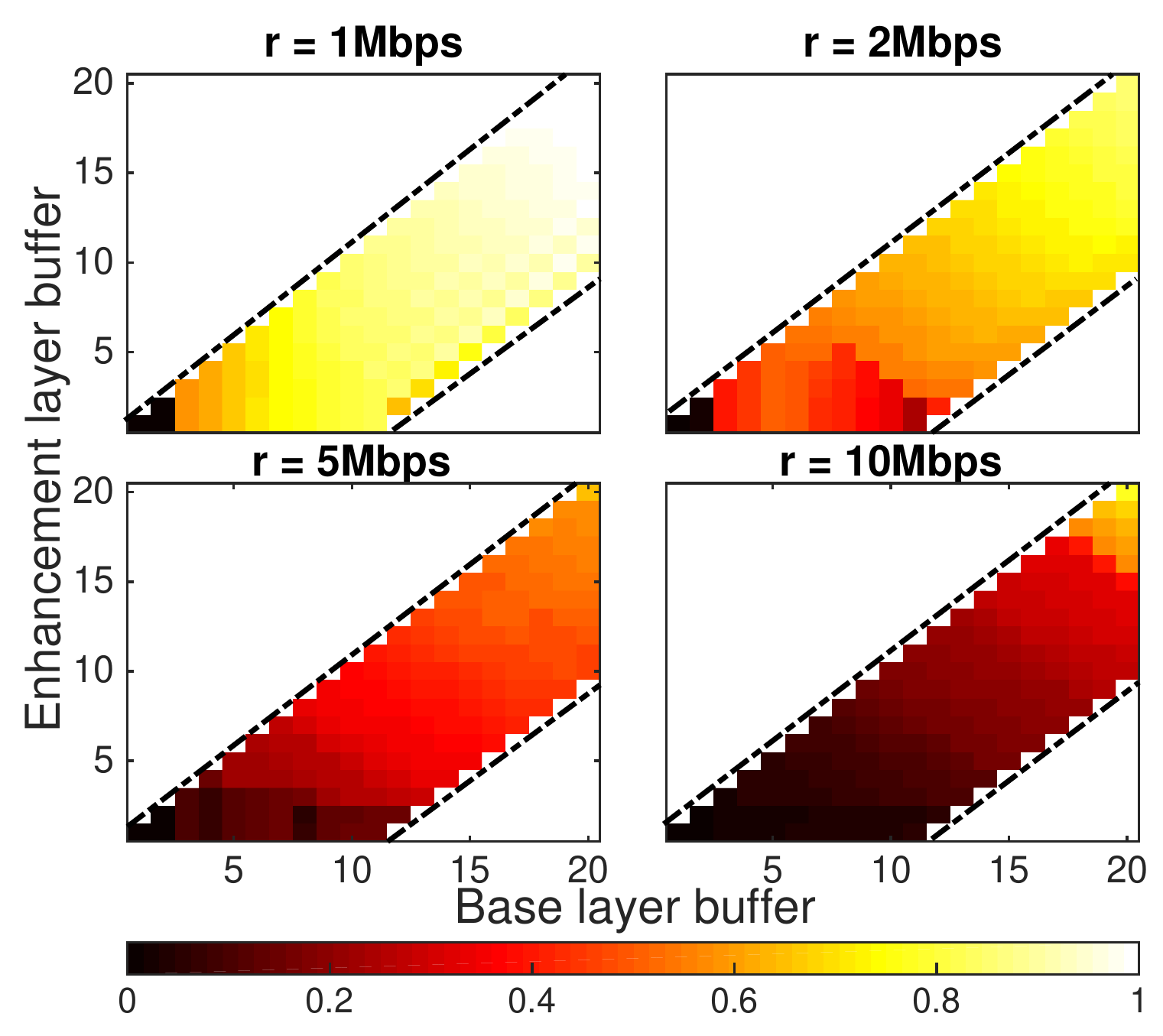}
             \caption{$c_{avg} =  4.5$ Mbps, $\rho = 2.5$}\label{fig:heatmap_thresh10_450_highload}
        \end{subfigure}
        \begin{subfigure}[h]{0.24\textwidth}
             \includegraphics[height=1.5in]{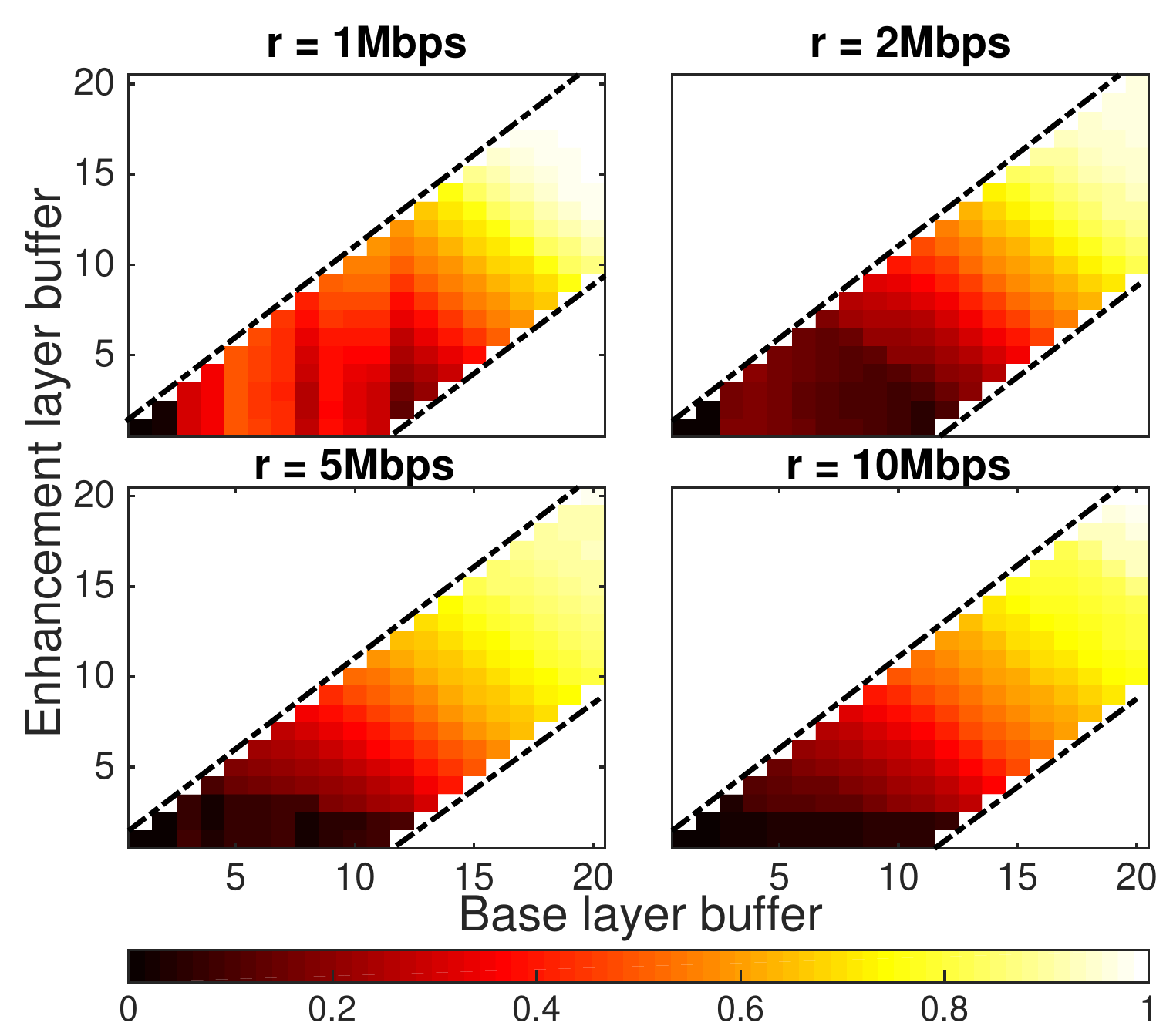}
             \caption{$c_{avg} =  4.5$ Mbps, $\rho = 2$}\label{fig:heatmap_thresh10_450_lowload}
        \end{subfigure}
        	\begin{subfigure}[h]{0.24\textwidth}
             \includegraphics[height=1.5in]{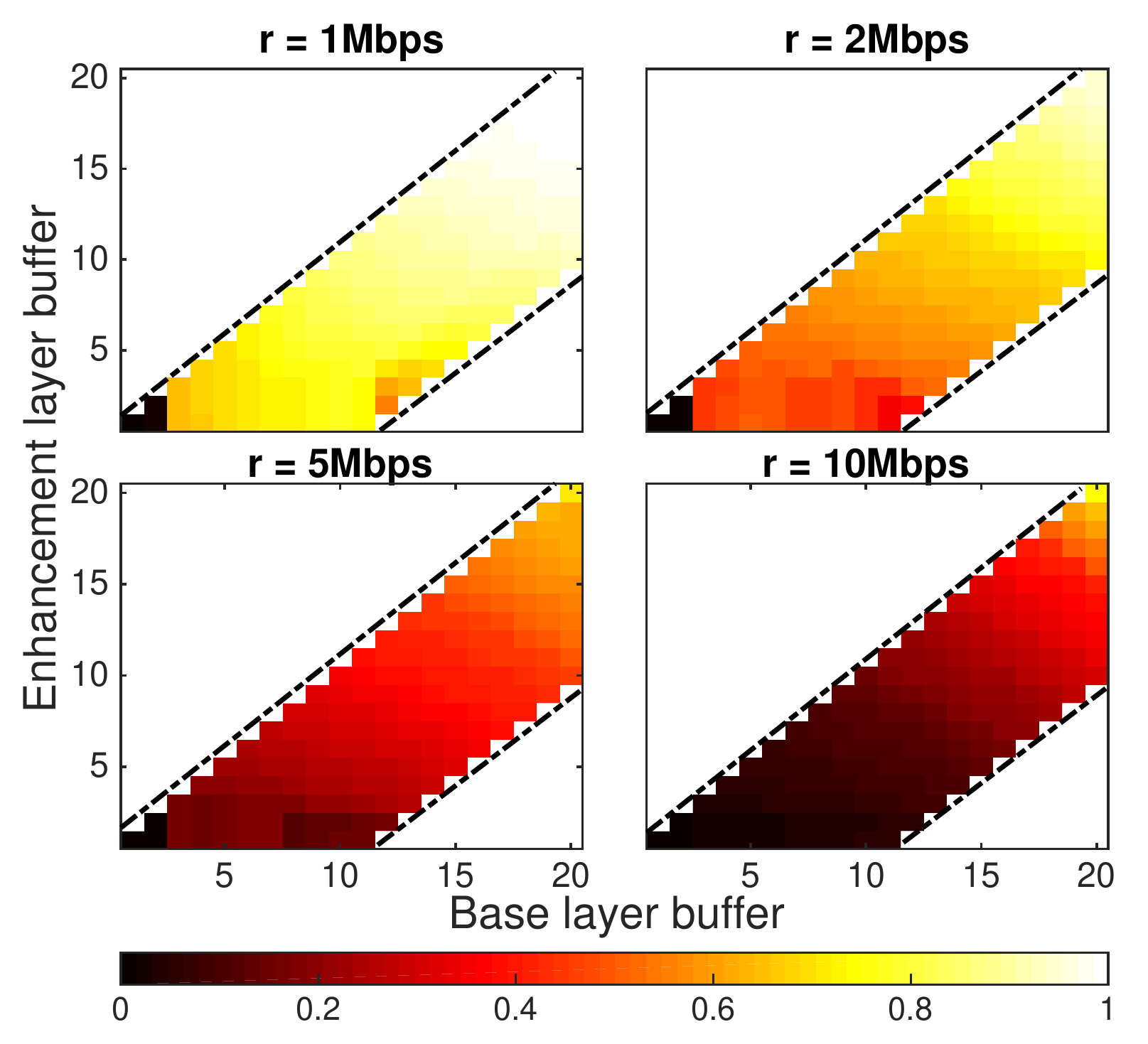}
             \caption{$c_{avg} =  3$ Mbps, $\rho = 1.88$}\label{fig:heatmap_thresh10_300_highload}
        \end{subfigure}
        \begin{subfigure}[h]{0.24\textwidth}
             \includegraphics[height=1.5in]{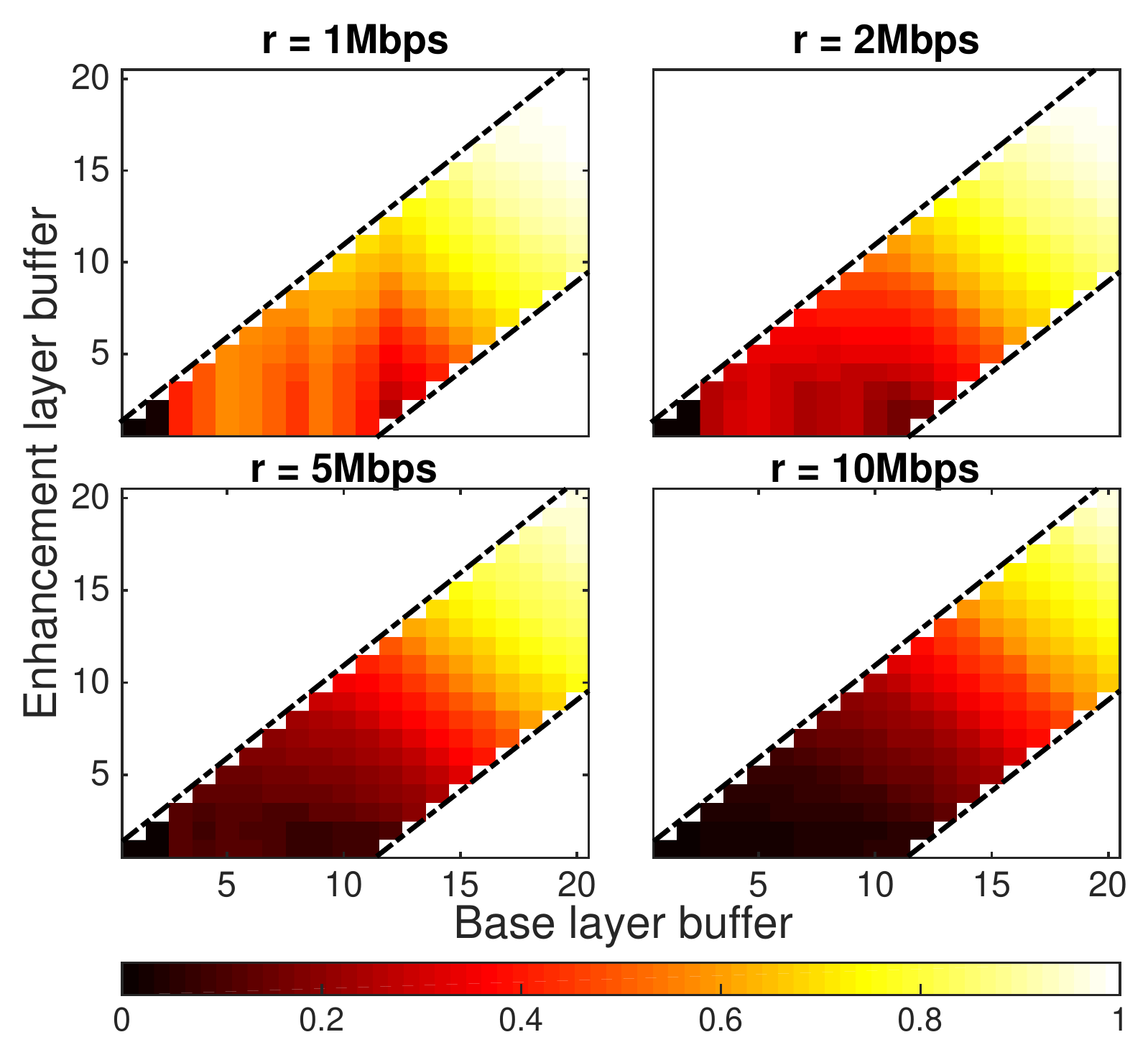}
             \caption{$c_{avg} =  3$ Mbps, $\rho = 1.25$}\label{fig:heatmap_thresh10_300_lowload}
        \end{subfigure}
\caption{Scheduling priority comparison for users with DBP-10s. For the case with $c_{avg} = 4.5$ Mbps, the critical load $\rho^*$ is equal to 2.3 and for the case with $c_{avg}= 3$ Mbps, $\rho^* = 1.58$. The region beyond the dashed lines corresponds to buffer values that are not possible.}
\label{fig:heatmap_thresh10}         
\end{figure*}

\begin{figure*}[t]
	\captionsetup[subfigure][h]{twoside,margin={0cm,0cm}}
	\centering
	\begin{subfigure}[h]{0.24\textwidth}
             \includegraphics[height=1.5in]{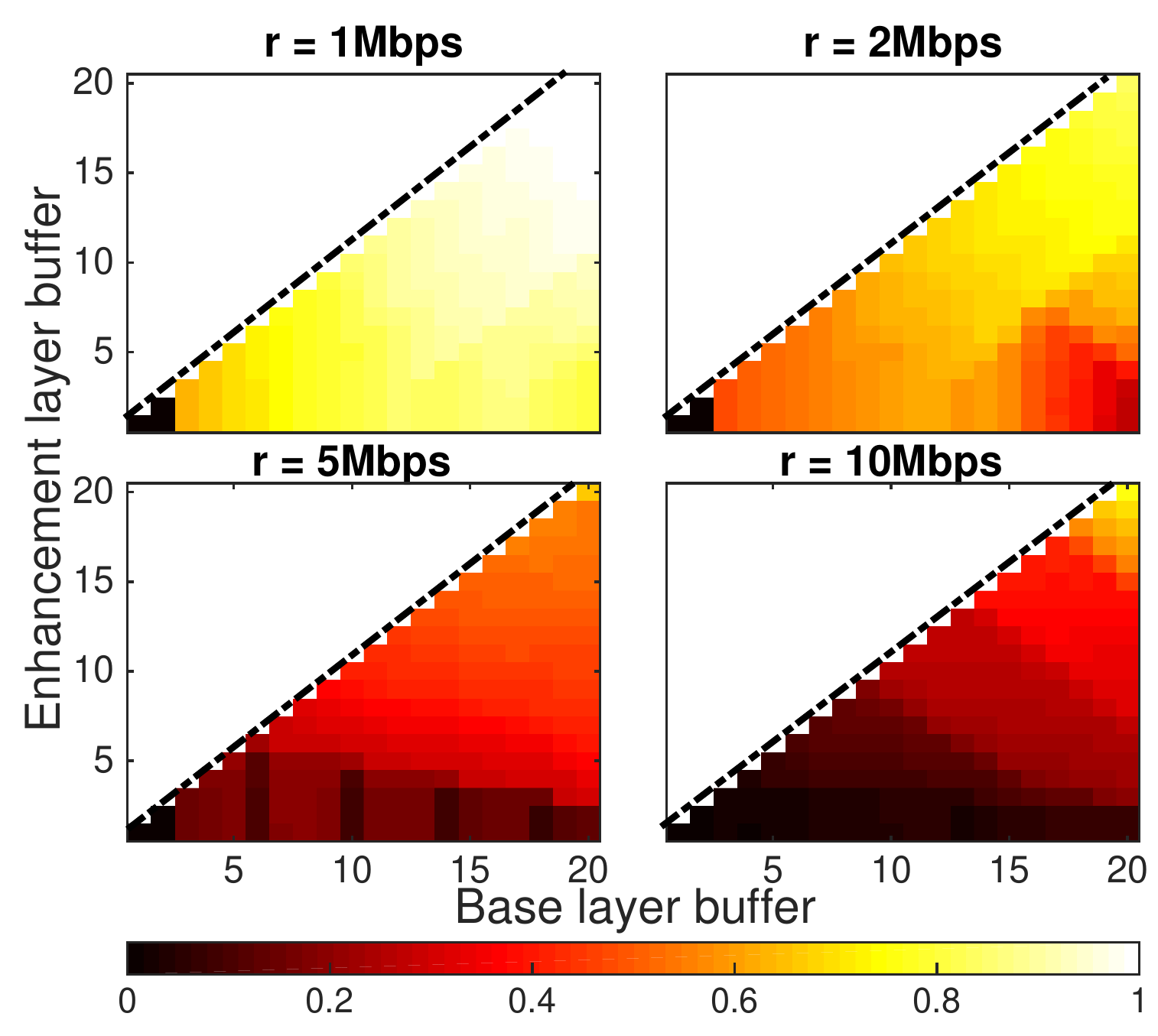}
             \caption{DBP-20s, Mbps, $\rho = 2.5$}\label{fig:heatmap_thresh20_450_highload}
        \end{subfigure}
        \begin{subfigure}[h]{0.24\textwidth}
             \includegraphics[height=1.5in]{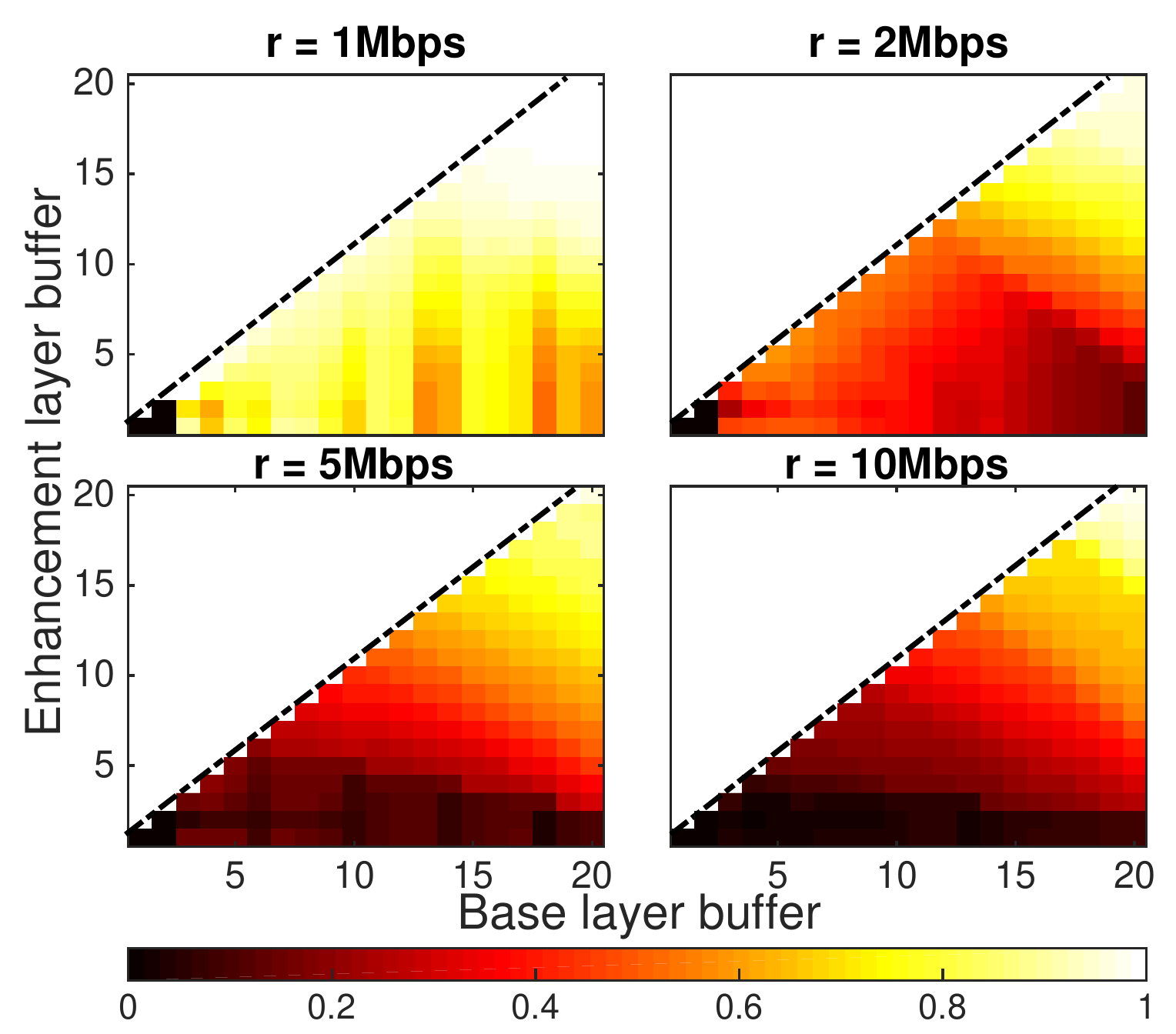}
             \caption{DBP-20s, Mbps, $\rho = 2$}\label{fig:heatmap_thresh20_450_lowload}
        \end{subfigure}
        	\begin{subfigure}[h]{0.24\textwidth}
             \includegraphics[height=1.5in]{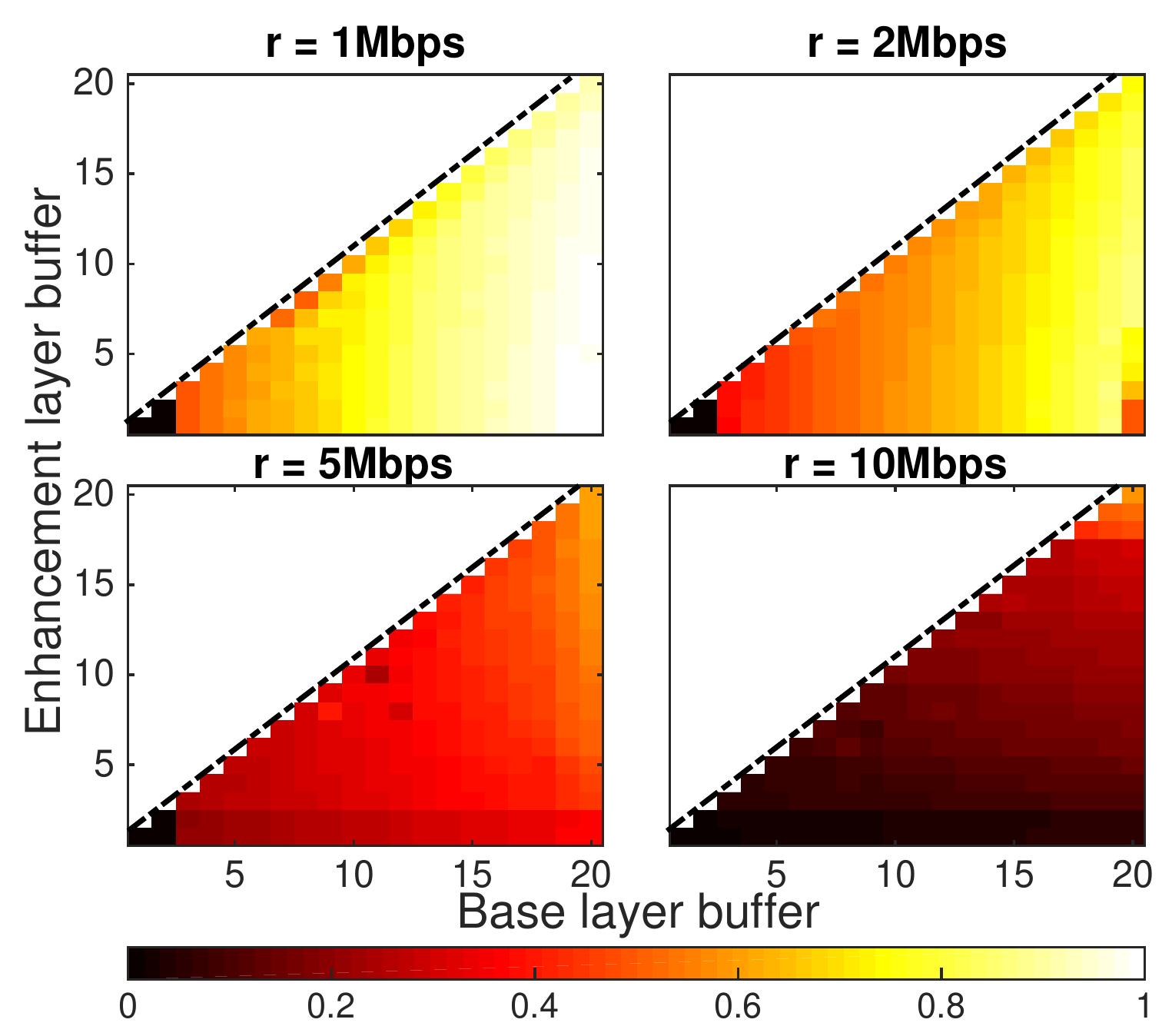}
             \caption{CBP, $\rho = 2.5$}\label{fig:heatmap_cbp_450_highload}
        \end{subfigure}
        \begin{subfigure}[h]{0.24\textwidth}
             \includegraphics[height=1.5in]{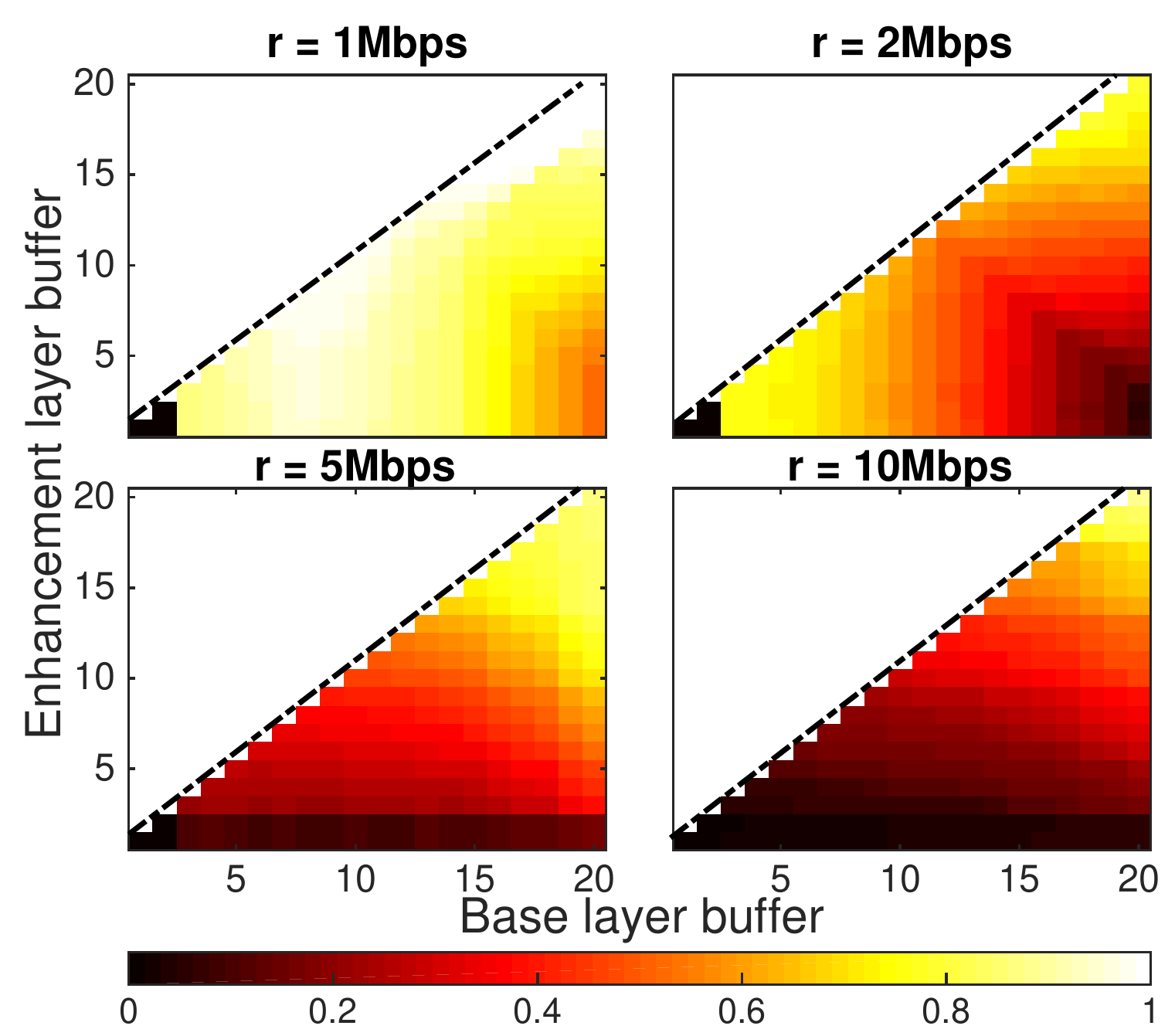}
             \caption{CBP, $\rho = 2$}\label{fig:heatmap_cbp_450_lowload}
        \end{subfigure}
\caption{Scheduling priority comparison for a network with $c_{avg} = 4.5$ Mbps. The figures on the left side correspond to DBP-20s and the ones on the right correspond to CBP. The region beyond the dashed lines corresponds to buffer values that are not possible.}
\label{fig:heatmap_difq_450}         
\end{figure*}

In a QA Blind scheduling policy, the base station does not know what quality layer each user requests at each time slot. Therefore, from the base station's point of view, the user buffer can be modeled as Figure \ref{fig:alg_buf}. In this figure, we illustrate a sample buffer of a user regardless of what layer and segment the data belongs to. Whenever a user is scheduled, the data is delivered and the buffer level increases at a rate of $\lambda$. Because the video is being continuously played back, the buffer level decreases at a rate of $\mu$. The average value of $\lambda$ in a homogeneous network, which we denote as $\lambda_{avg}$ is the average throughput of each user. We can calculate $\lambda_{avg}$ as $\frac{c_{avg}}{\rho}$, where $\rho$ is the load on the network as defined in previous sections and $c_{avg}$ is the average capacity of each subchannel. Also, $\mu_{avg}$ is defined as the average rate of draining the buffer, which depends on the average rate of the video segments being played back. In Appendix \ref{app_3}, we derive an expression for $\mu_{avg}$ given the optimal variables derived from the RB problem.

Figure \ref{fig:alg_lambda_mu} illustrates the values for $\lambda_{avg}$ and $\mu_{avg}$ for different QA and $c_{avg}$ as a function of network load with settings similar to Table \ref{table:params}. We can observe that for each pair of $\lambda_{avg}$ and $\mu_{avg}$, their values coincide at a specific network load, which we call critical load $\rho^*$. This is the network load for which on average, the buffer level remains stable. For load values larger than $\rho^*$, the buffer level will decrease and vice versa. We will use the concept of critical load to derive conclusions regarding the QA Blind scheduling policy.

We now turn our attention to the QAA algorithm in order to determine its outcome at the critical load. Here, in order to clean out states that have no significance in the scheduling policy, we add a sub step between step 0 and step 1 of Algorithm \ref{alg:ranking} in which we remove all states $s$ for which $x^{*0 }_s=x^{*1}_s = 0$ based on the following argument.

\begin{definition}
We define a particular state $s \in \mathcal{S}$ to be reachable from state $l$ under policy $u$, if either $h_{ls}^0 > 0$ and $x(u)_l^{0} > 0$ or $h_{ls}^1 > 0$ and $x(u)_l^{1} > 0$. In other words, $s$ is reachable from $l$ under a given policy, if there is a path from $l$ to $s$ suggested by the policy.
\end{definition}

\begin{theorem}
A particular state $s$ satisfies $x^{*0 }_s=x^{*1}_s = 0$ if, and only if, that state is unreachable from the initial state and any state on the trajectory determined by the optimal policy $u^*$.
\end{theorem}
\begin{proof}
Refer to Appendix \ref{app_4}.
\qed
\end{proof}

We run the QAA algorithm for a network with settings similar to Table \ref{table:params}. In order to determine the relation between scheduling priority and state space attributes. For this purpose, we run QAA and rank all states in an ordered list where the head of list is the state with the highest scheduling priority. For each state $s$, we assign an index $i_s = \frac{p_s}{|\mathcal{S}|}$, where, $p_s$ is the position of state $s$ in the ordered list. Using this index representation, we generate heatmaps that illustrate the scheduling priority of each state. 

Figure \ref{fig:heatmap_thresh10} illustrates the scheduling priority heatmap with respect to the instantaneous channel state and the buffer occupancy of both layers. The darker the color, the higher the state appears in the priority list. Also, all users deploy DBP-10s and $c_{avg} = 4.5$ Mbps. From Figure \ref{fig:alg_lambda_mu}, the critical load for this case is equal to $\rho^* = 2.3$. For a load value larger than 2.3, we observe in Figure \ref{fig:heatmap_thresh10_450_highload} that the scheduling priority is highly channel dependent, with users with the highest channel capacity getting the highest priority. Figure \ref{fig:heatmap_thresh10_450_lowload} shows that for load values smaller than $\rho^*$, the policy begins to become buffer dependent prioritizing users that have less buffer occupancy. A similar trend is observed in Figures \ref{fig:heatmap_thresh10_300_highload} and \ref{fig:heatmap_thresh10_300_lowload} where $c_{avg} = 3$ Mbps and therefore, $\rho^* = 1.58$.
From Figure \ref{fig:heatmap_difq_450} we can conclude that the above observation is not limited to DBP-10s and applies to a great extent also to cases with DBP-20s and CBP, as shown in Figures \ref{fig:heatmap_thresh20_450_highload}-\ref{fig:heatmap_thresh20_450_lowload} and \ref{fig:heatmap_cbp_450_highload}-\ref{fig:heatmap_cbp_450_lowload}, respectively. Therefore, a scheduling mechanism that leverages this trend can be used for a variety of QA policies which the scheduler does not need to know in advance.

Due to the heterogeneity of wireless networks, users will face rising buffer levels at some times and draining buffer levels at others. From the above analysis we can conclude that these fluctuations in buffer level can be exploited to devise scheduling algorithms without the scheduler knowing the underlying system parameters. Such an algorithm should first provide a measure to quantify fluctuations in the buffer level of each user. This measure will then be used to perform buffer dependent scheduling when buffer is filling, and channel dependent scheduling, when the buffer is draining. Given these guidelines, we devise a simple QA Blind scheduling policy called \emph{Buffer Evolution Aware Scheduling (BEAS)} shown in Algorithm \ref{alg:buf}.

\begin{algorithm}
\textbf{Initialization:} Let $\epsilon>0$, for 
$i\in\mathcal{N}$, let $b^k_i = b_0$.\\
The number of layer $l$ segments that are delivered to user $j$ in time slot $k$ is denoted by $n^k_{j,l}$.
\begin{algorithmic}
\FOR{all time slots $k$}
\STATE$b^{k+1}_i = (1-\epsilon)b^k_i - \epsilon \tau_{slot}$
\STATE \textbf{SCHEDULE:} 
\STATE $\mathcal{B} = \{i | b^k_i < b_{thresh}\}$
\IF{$|\mathcal{B}| < M$}
\STATE Schedule $M$ users from $\mathcal{B}$ with the best channel.
\ELSE 
\STATE Schedule all users in $\mathcal{B}$.
\STATE Schedule $M - |\mathcal{B}|$ users from $\mathcal{N}\setminus\mathcal{B}$ with the lowest base layer occupancy.
\ENDIF
\STATE \textbf{UPDATE:} 
\FOR{all scheduled users $j$:}
\STATE $b^{k+1}_j = (1-\epsilon)b^k_j + \epsilon \tau_{seg} h(\sum_{l = 1}^L  n^k_{j,l})$
\ENDFOR
\ENDFOR
\end{algorithmic}
\caption{Buffer Evolution Aware Scheduling}
\label{alg:buf}
\end{algorithm} 

In Algorithm \ref{alg:buf}, we use an auxiliary variable $b_i$ as a measure to represent buffer fluctuations. By starting from an initial value $b_0$ and updating it at each time slot, we can quantify whether the buffer is draining, (decreasing $b_i$) or filling (increasing $b_i$). As a scheduling rule, we first consider users with $b_i$ less than a pre-determined threshold $b_{thresh}$. Among these users, those with better channel are prioritized. If any resources are left, we move to the rest of the users and schedule them by prioritizing users that have fewer base layer segments in the buffer. The update rule for $b_i$ is based on an exponential filter with a smoothing factor $\epsilon$. A larger $\epsilon$ reacts faster to buffer fluctuations while a smaller value results in a smooth representation for the buffer fluctuations. Also, $h(\cdot)$ is a function of the total number of sub-segments delivered in each time slot. In our simulations, we have determined by trial and error that a linear function in the form $h(x)=\alpha x + \beta$ results in the best performance. Algorithm \ref{alg:buf} describes this heuristic. In Section \ref{sec:implementation}, we discuss the practical implications of this algorithm in more detail.

%% file: simulation.tex
\section{Simulation Results}\label{sec:sim}
In this section, we perform an extensive simulation study evaluating the performance of the algorithms presented in the previous section and compare them with several baseline schemes. Unless mentioned otherwise, the simulation parameters are similar to Table \ref{table:params} and the video length is 10 minutes. The QA schemes used in the simulations are designed similar to Section \ref{sec:prob_form}. For the BEAS algorithm we use $b_{thresh}=0$. 

We start by studying the effect of buffer limit on the system performance in Figure \ref{fig:buflim}. It can be seen that increasing the buffer limit beyond 20s will not significantly improve the delivered video quality. Therefore, for the remainder of the simulations, we set the buffer limit to 20s in order to gain a suitable trade-off between computational complexity and average video quality.

\begin{figure}[h]
\hspace*{-0.6cm}
	\captionsetup[subfigure][h]{twoside,margin={0cm,0cm}}
	\centering
	\begin{subfigure}[h]{0.25\textwidth}
             \includegraphics[height=1.38in]{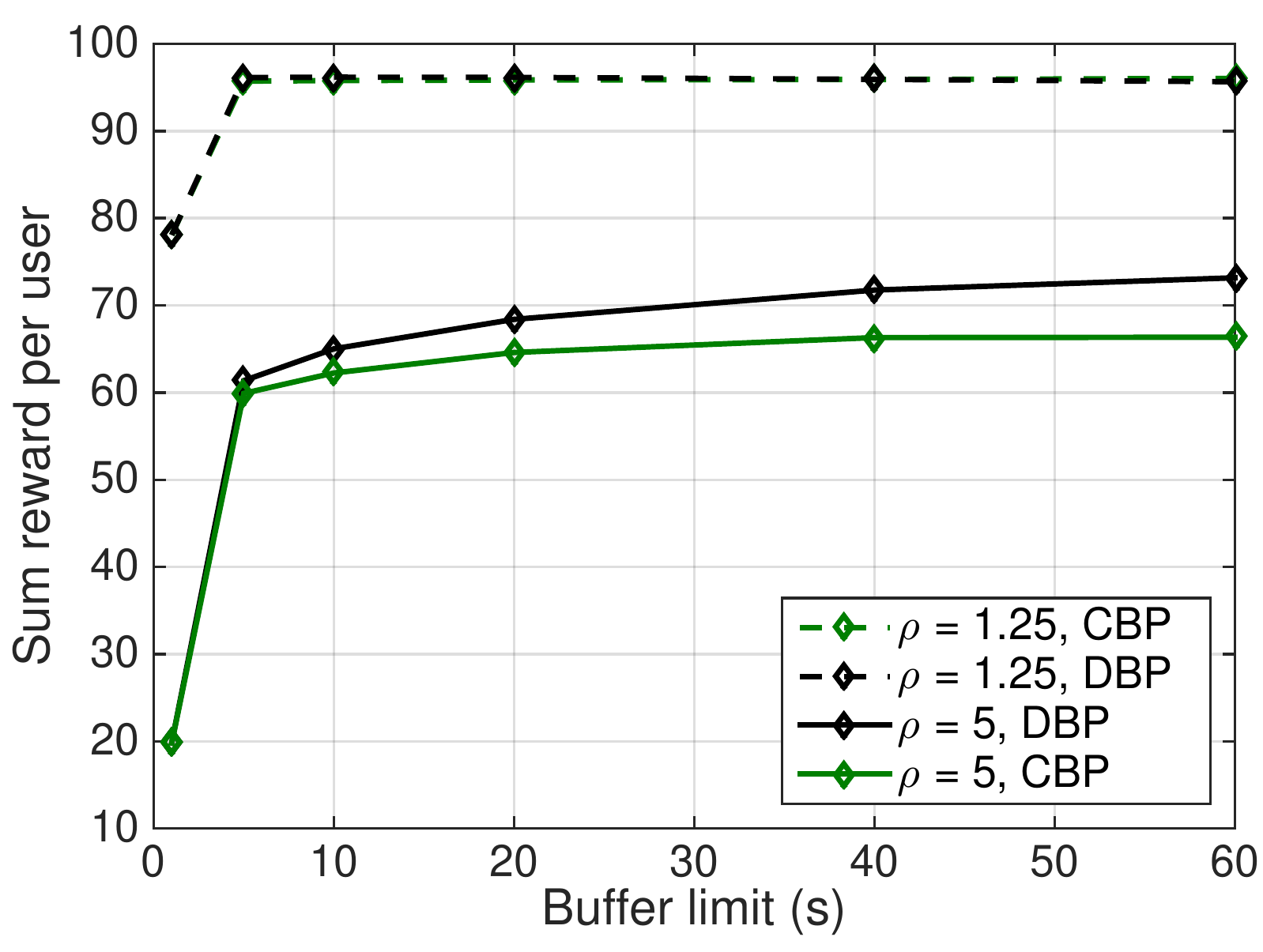}
             \caption{$c_{avg} =  4.5$ Mbps}\label{fig:buflim_high}
        \end{subfigure}
        \begin{subfigure}[h]{0.25\textwidth}
             \includegraphics[height=1.38in]{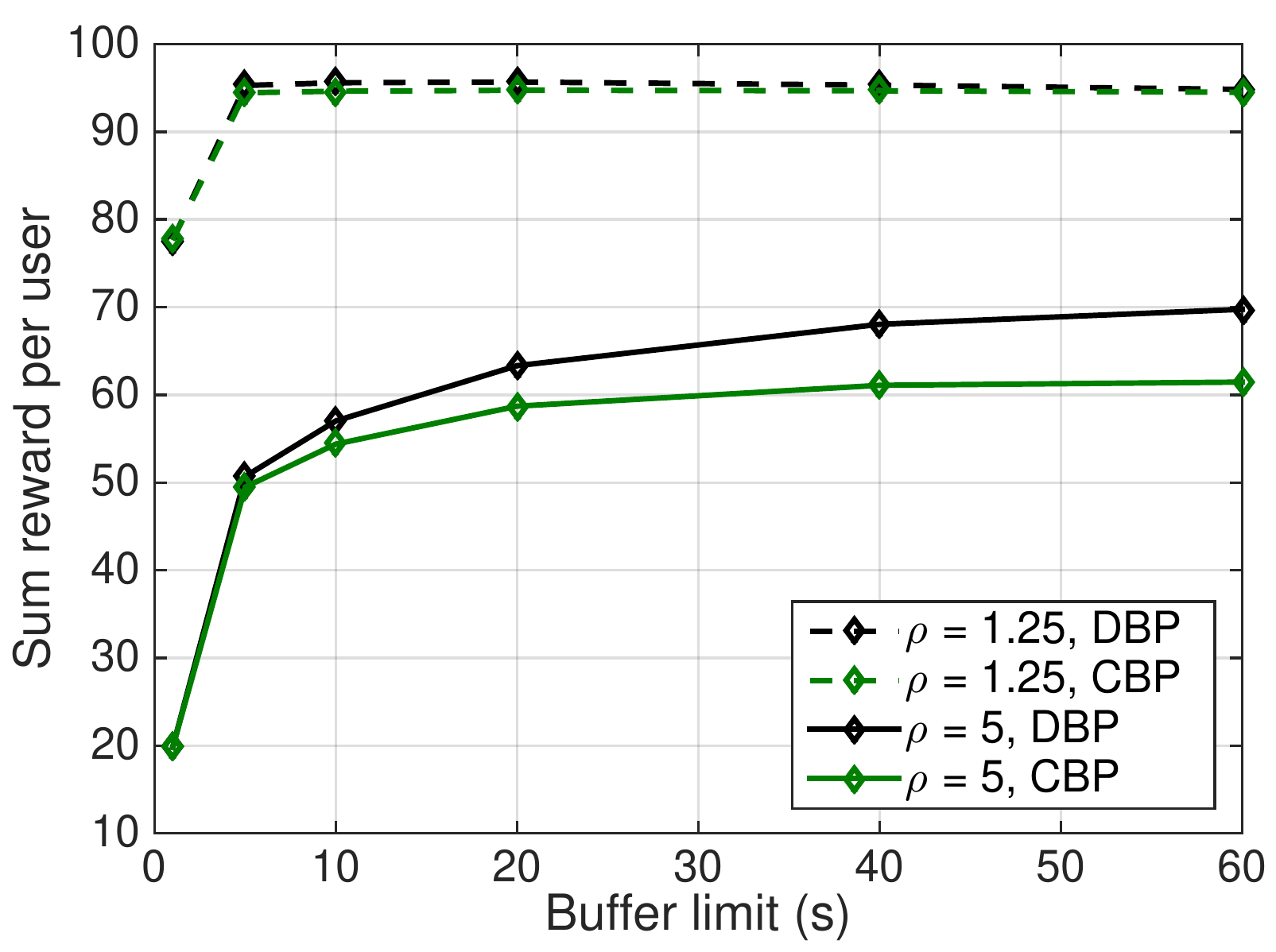}
             \caption{$c_{avg} =  2.55$ Mbps}\label{fig:buflim_low}
        \end{subfigure}        	
\caption{Reward gained by users for different buffer limit values. The QAs used for the experiment are CBP and DBP with the pre-fetch threshold being the buffer limit.}
\label{fig:buflim}         
\end{figure}

Next, we move on to comparing the QAA and BEAS algorithms with three baseline algorithms, namely Proportional Fairness (PF) \cite{kelly1998rate}, Best Channel First (BCF), and Lowest Buffer First (LBF) \cite{yeh2014bandwidth} algorithms. PF is a very popular scheduling scheme for wireless networks in which users are scheduled based on their current channel state normalized by their long term average throughput. BCF is a purely channel dependent scheduling method that only takes the current link conditions of each user and schedules users with the best channel. LBF is a purely buffer dependent scheme in which users that have fewer base layers in the buffer are prioritized. 

An important measure for comparison is to determine how each of these algorithms implement the quality-delay trade-off explained in Section \ref{sec:sys_model}. Therefore, for each scenario under consideration, we show the average fraction of time that each user spends re-buffering, the average fraction of segments that are delivered with only the base layer, and the value of the sum reward per user, which combines both QoE measures into one. It should be noted that while for the re-buffering and reward plots, the x-axis represents the load on the network, the same axis for the video layer plots shows the number of subchannels.

\subsection{Homogeneous System}\label{sec:hom}
We first consider the case of homogeneous users in Figures \ref{fig:midchan} and \ref{fig:lowchan} for channels with $c_{avg} = 4.5$ Mbps and $c_{avg} = 2.55$ Mbps, respectively. By looking at Figures \ref{fig:midchan_reward} and \ref{fig:lowchan_reward}, we observe that in terms of reward, QAA performs very close to the optimum illustrated by the black line, especially for highly loaded networks. Furthermore, we observe that LBF performs better than PF and BCF in high capacity networks with low load while for heavily loaded or low capacity networks, the reverse occurs. 

From Figures \ref{fig:midchan_rebuf} and \ref{fig:lowchan_rebuf}, we see that the two channel based schemes have poor delay performance. Also, in the low capacity network with high load, LBF also has poor delay performance which is due to the fact that by always scheduling the user with the smallest buffer, it might choose users with very poor channel conditions for which the download takes long. In other words, LBF has poor spectrum utilization which is detrimental in low capacity scenarios and high loads where resources are very scarce. These findings suggest that in order to provide satisfactory delay performance, algorithms should take both channel, and buffer state into account.

\begin{figure}[h]
\centering
\hspace*{-0.6cm}
	\captionsetup[subfigure][h]{twoside,margin={0cm,0cm}}
	\centering
	\begin{subfigure}[h]{0.25\textwidth}
             \includegraphics[height=1.38in]{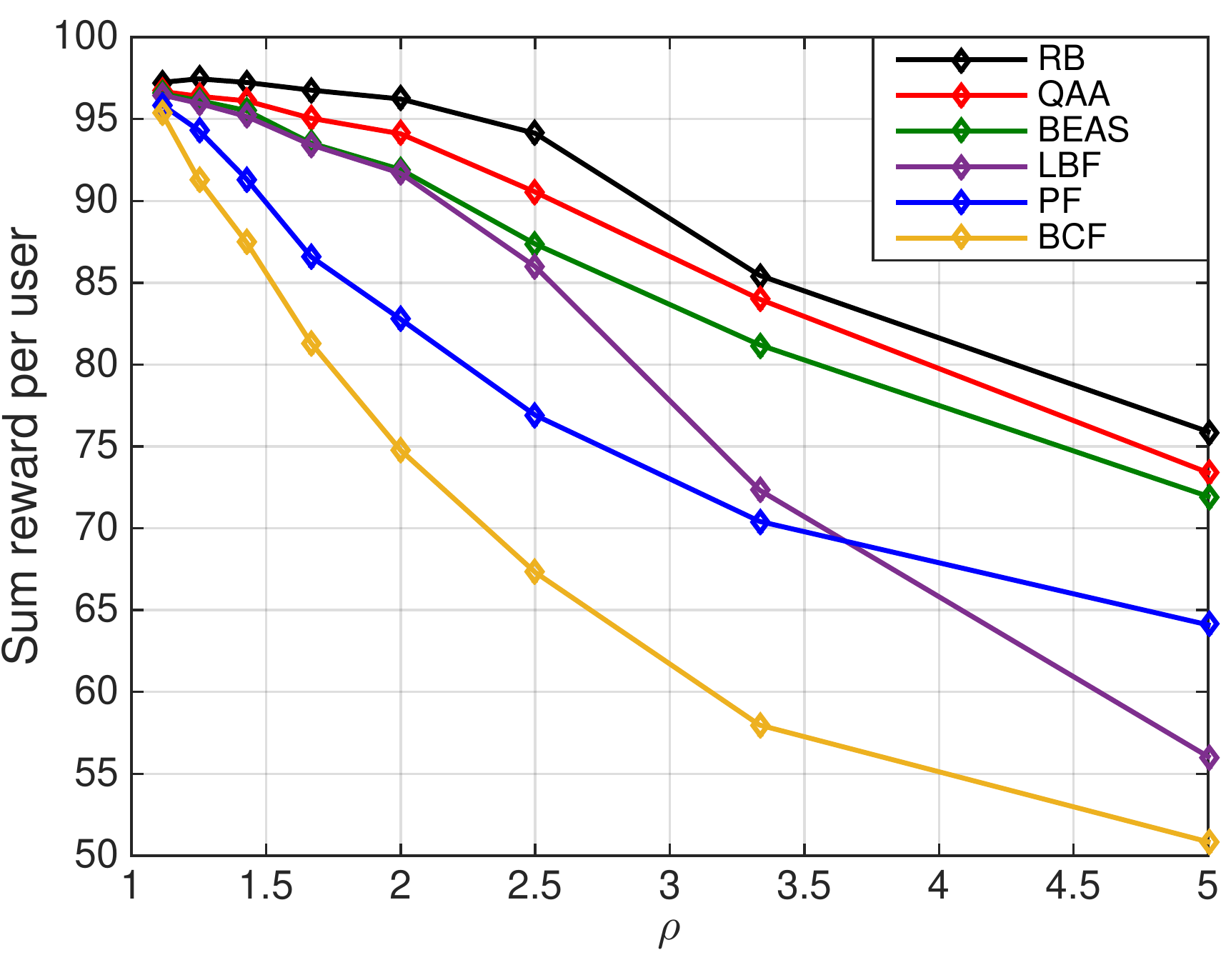}
             \caption{Reward per user}\label{fig:midchan_reward}
        \end{subfigure}
        \begin{subfigure}[h]{0.25\textwidth}
             \includegraphics[height=1.38in]{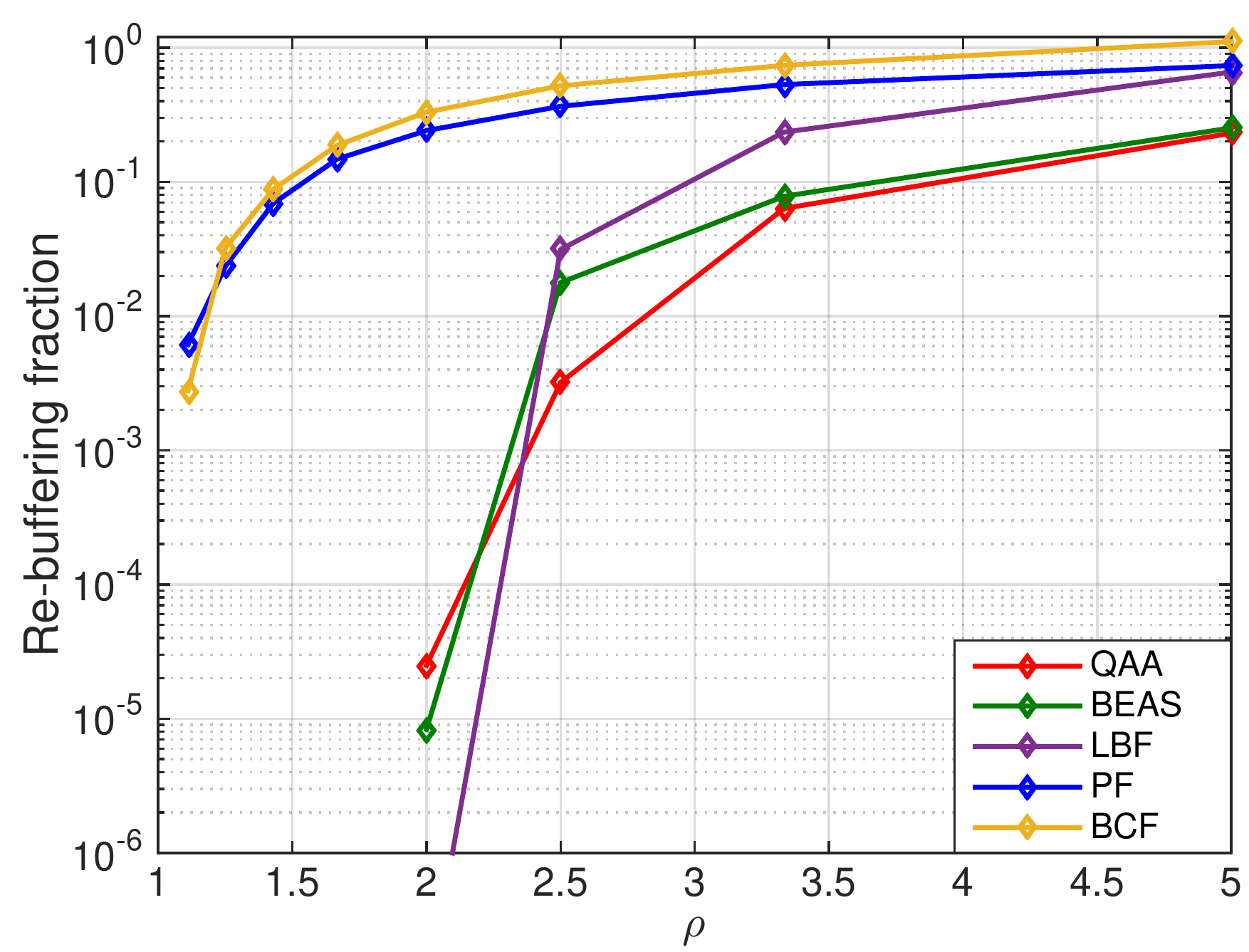}
             \caption{Re-buffering fraction}\label{fig:midchan_rebuf}
        \end{subfigure}   \\
       \begin{subfigure}[h]{0.35\textwidth}
             \includegraphics[height=1.47in]{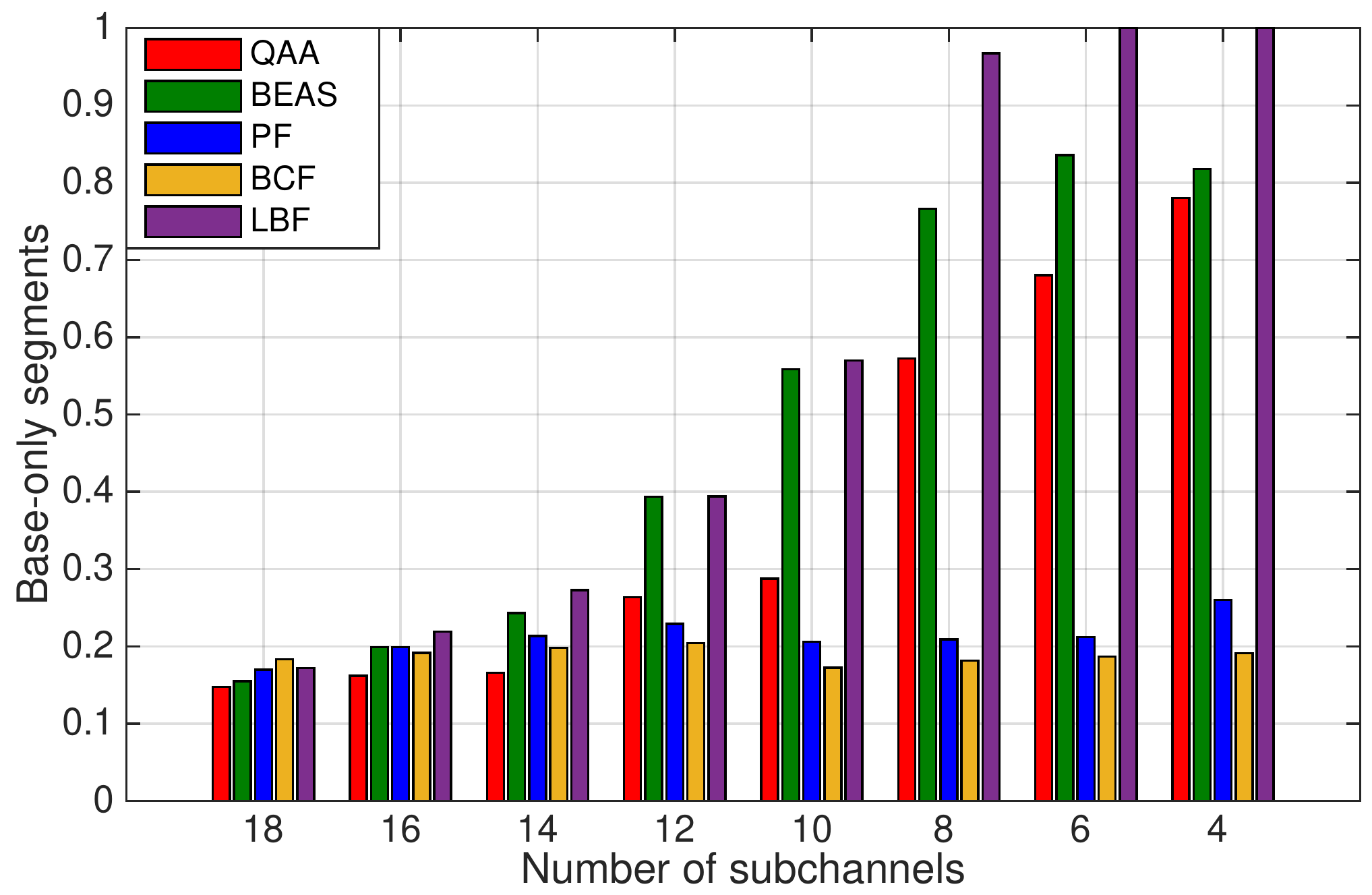}
             \caption{Base-only segment fraction}\label{fig:midchan_layer}
        \end{subfigure}        	
\caption{20 homogeneous users streaming a 10 minute video with DBP-20s, $c_{avg} = 4.5$ Mbps.}
\label{fig:midchan}         
\end{figure}

Providing enhanced delay performance comes at the cost of delivering segments with fewer higher layers in order to avoid re-buffering. Figures \ref{fig:midchan_layer} and \ref{fig:lowchan_layer} show the average fraction of segments that were delivered with only the base layer. We can see from these figures that PF and BCF provide on average more segments with maximum quality than the other schemes. This result, together with the delay performance shows that these two schemes tend to over-serve some users, thereby being able to deliver more full quality segments, and under-serve the rest and cause large re-buffering. The QAA scheme adjusts the base layer only fraction with the load on the network, hence, when the load is large, fewer full quality segments are delivered and vice versa. LBF has a poor performance in these figures which is due to the fact that by preferring small buffer users without taking into account the channel conditions, barely any user can get beyond the initial base layer build up phase of DBP-20s. We can also see that by decreasing the average capacity of the network in Figure \ref{fig:lowchan_layer}, all scheduling schemes are more prone to delivering base layer only segments. 

In all cases discussed above, BEAS performs closest to QAA in terms of average reward per user which is mostly due to its ability to efficiently avoid re-buffering. However, for the video quality, it is sometimes not able to effectively mimic QAA. This is the penalty of not knowing the users' QA. 

\begin{figure}[h]
\centering
\hspace*{-0.6cm}
	\captionsetup[subfigure][h]{twoside,margin={0cm,0cm}}
	\centering
	\begin{subfigure}[h]{0.25\textwidth}
             \includegraphics[height=1.38in]{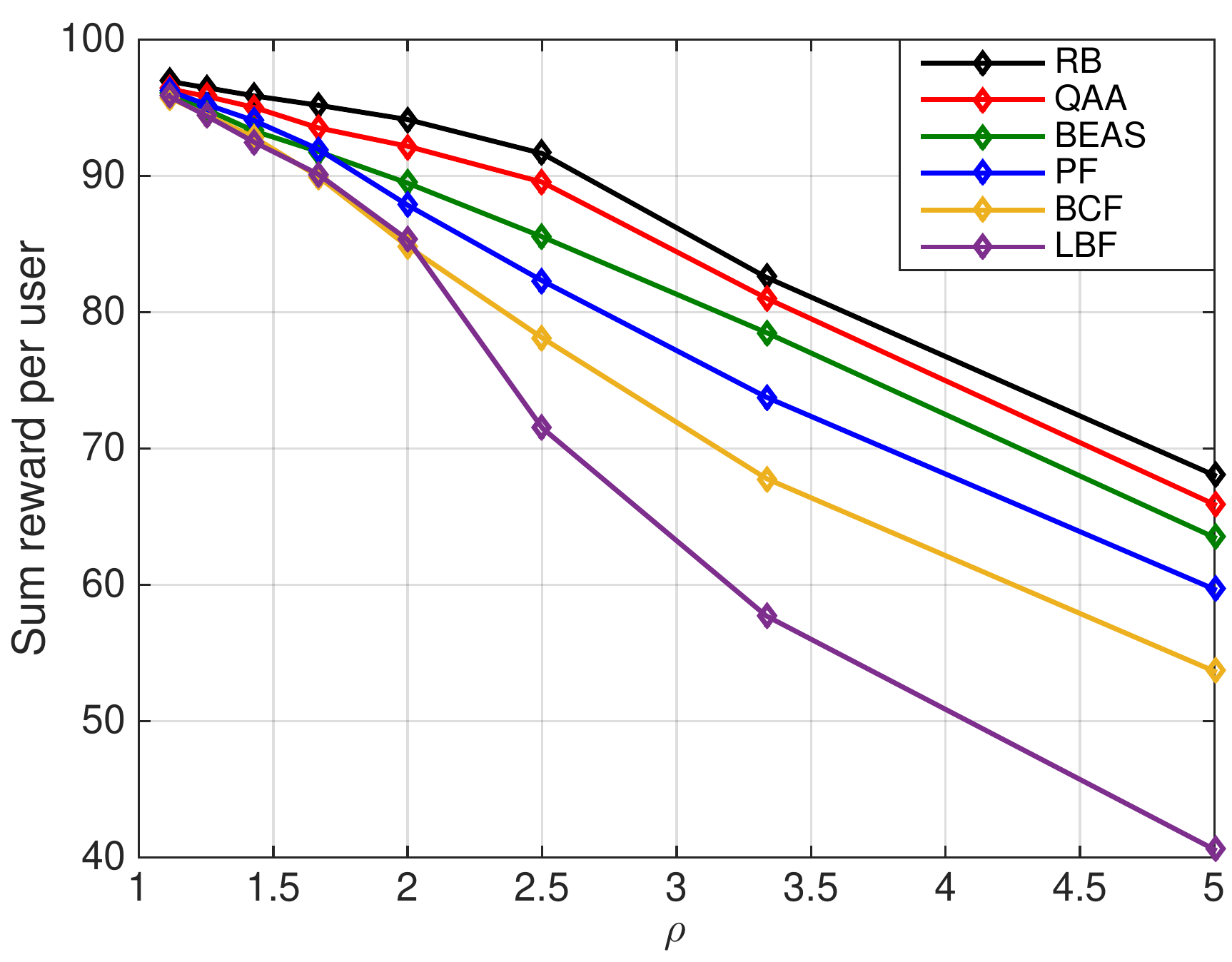}
             \caption{Reward per user}\label{fig:lowchan_reward}
        \end{subfigure}
        \begin{subfigure}[h]{0.25\textwidth}
             \includegraphics[height=1.38in]{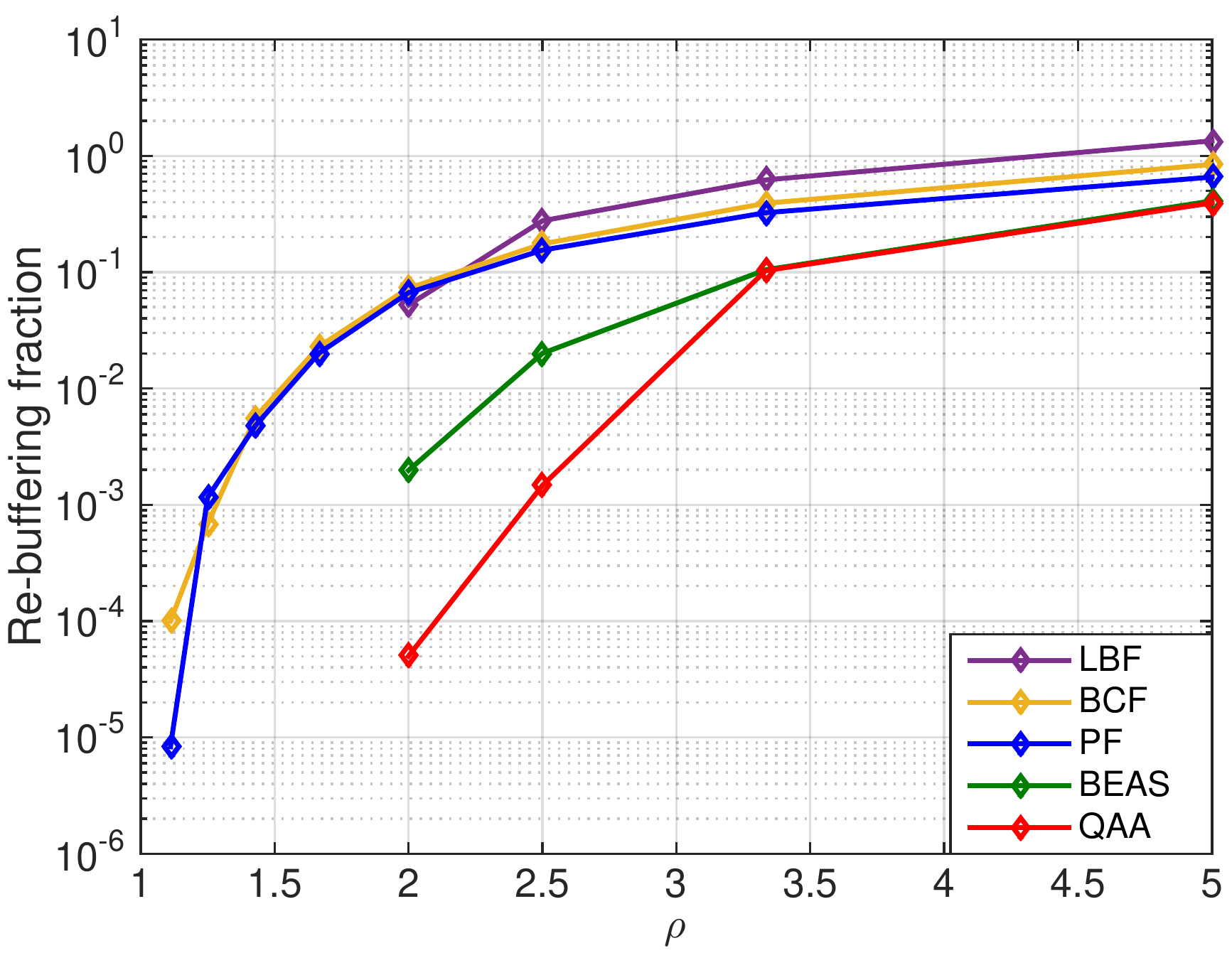}
             \caption{Re-buffering fraction}\label{fig:lowchan_rebuf}
         \end{subfigure}   \\
       \begin{subfigure}[h]{0.35\textwidth}
             \includegraphics[height=1.45in]{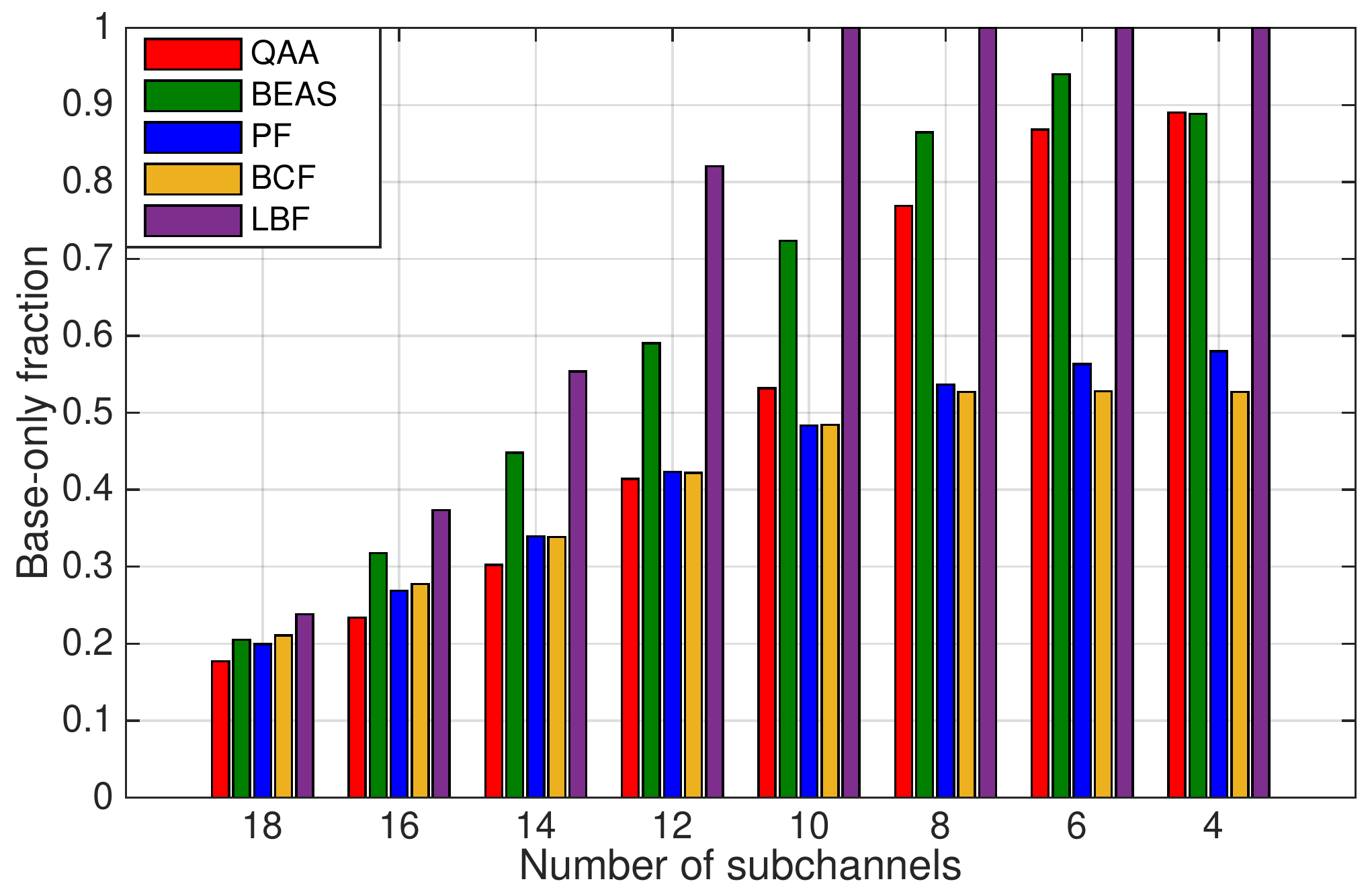}
             \caption{Base-only segment fraction}\label{fig:lowchan_layer}
        \end{subfigure}        	
\caption{20 homogeneous users streaming a 10 minute video with DBP-20s, $c_{avg} = 2.55$ Mbps.}
\label{fig:lowchan}         
\end{figure}

\subsection{Heterogeneous System}\label{sec:het}
Figures \ref{fig:diffchan} and \ref{fig:diffQ} show the performance of the scheduling schemes in non-homogeneous networks based on the discussion in Corollary \ref{cor:groups}. Figure \ref{fig:diffchan} shows the QoE metrics of a network with 20 users, out of which 10 experience $c_{avg} = 4.5$ Mbps and the other 10 experience $c_{avg} = 2.55$ Mbps. Similarly, Figure \ref{fig:diffQ} represents a network with $c_{avg} = 4.5$ Mbps and 20 users. Here, 10 users deploy a DBP-15s QA and the others use CBP. Similar to the homogenous cases, QAA and BEAS outperform the other algorithms in both sum reward and re-buffering. The general trend of the results is similar to the homogeneous case. In Figure \ref{fig:diffchan}, similar to \ref{fig:lowchan}, due to the presence of users in poor conditions, LBF degrades in performance as the load on the network increases. Also, in Figure \ref{fig:diffQ_layer}, we observe that more users are able to deliver full quality segments, which is due to the fact that unlike DBP which starts with downloading only base layers regardless of the channel conditions, CBP is very aggressive in requesting enhancement layers when the channel is in good condition.

\begin{figure}[h]
\hspace*{-0.6cm}
	\captionsetup[subfigure][h]{twoside,margin={0cm,0cm}}
	\centering
	\begin{subfigure}[h]{0.25\textwidth}
             \includegraphics[height=1.38in]{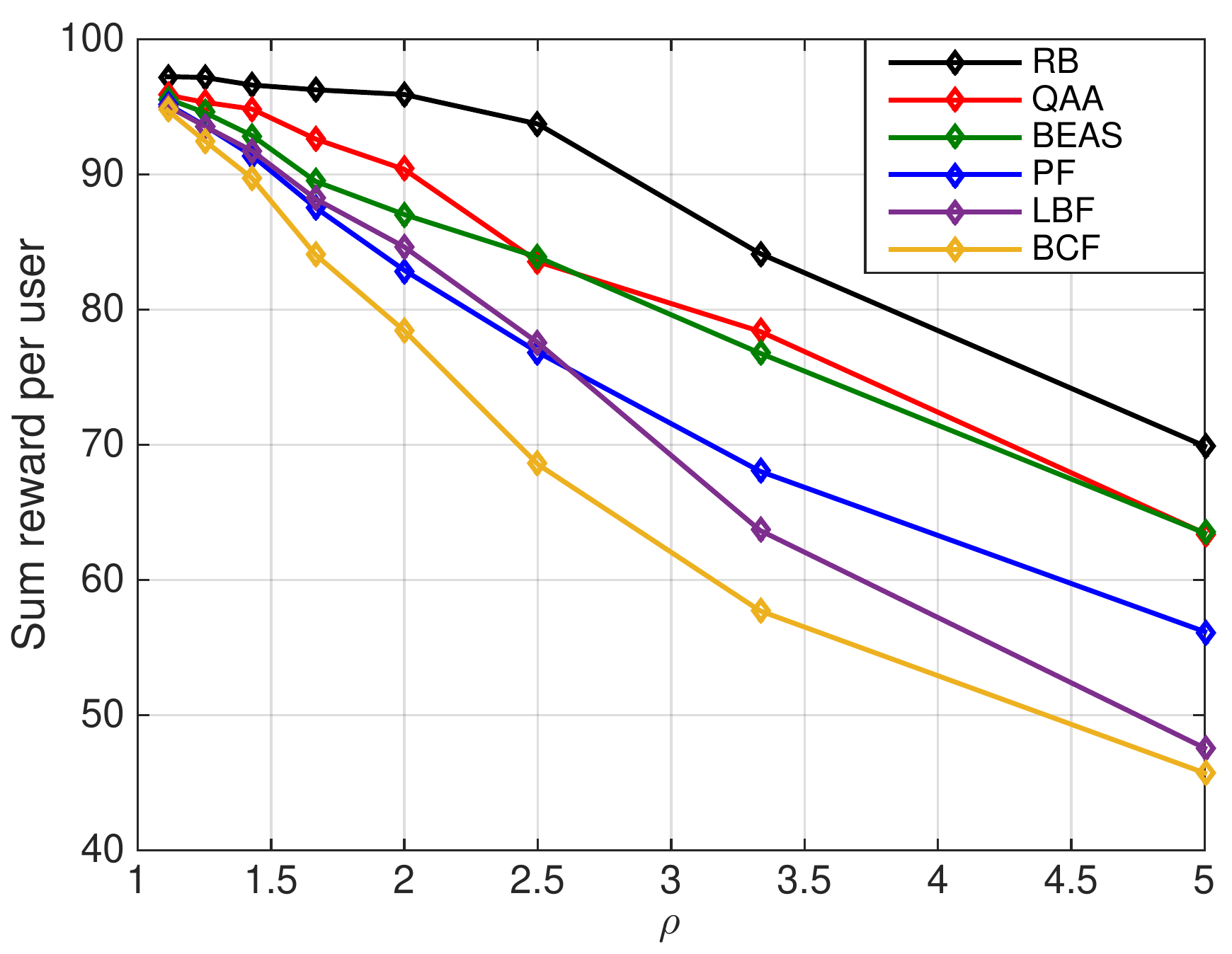}
             \caption{Reward per user}\label{fig:diffchan_reward}
        \end{subfigure}
        \begin{subfigure}[h]{0.25\textwidth}
             \includegraphics[height=1.38in]{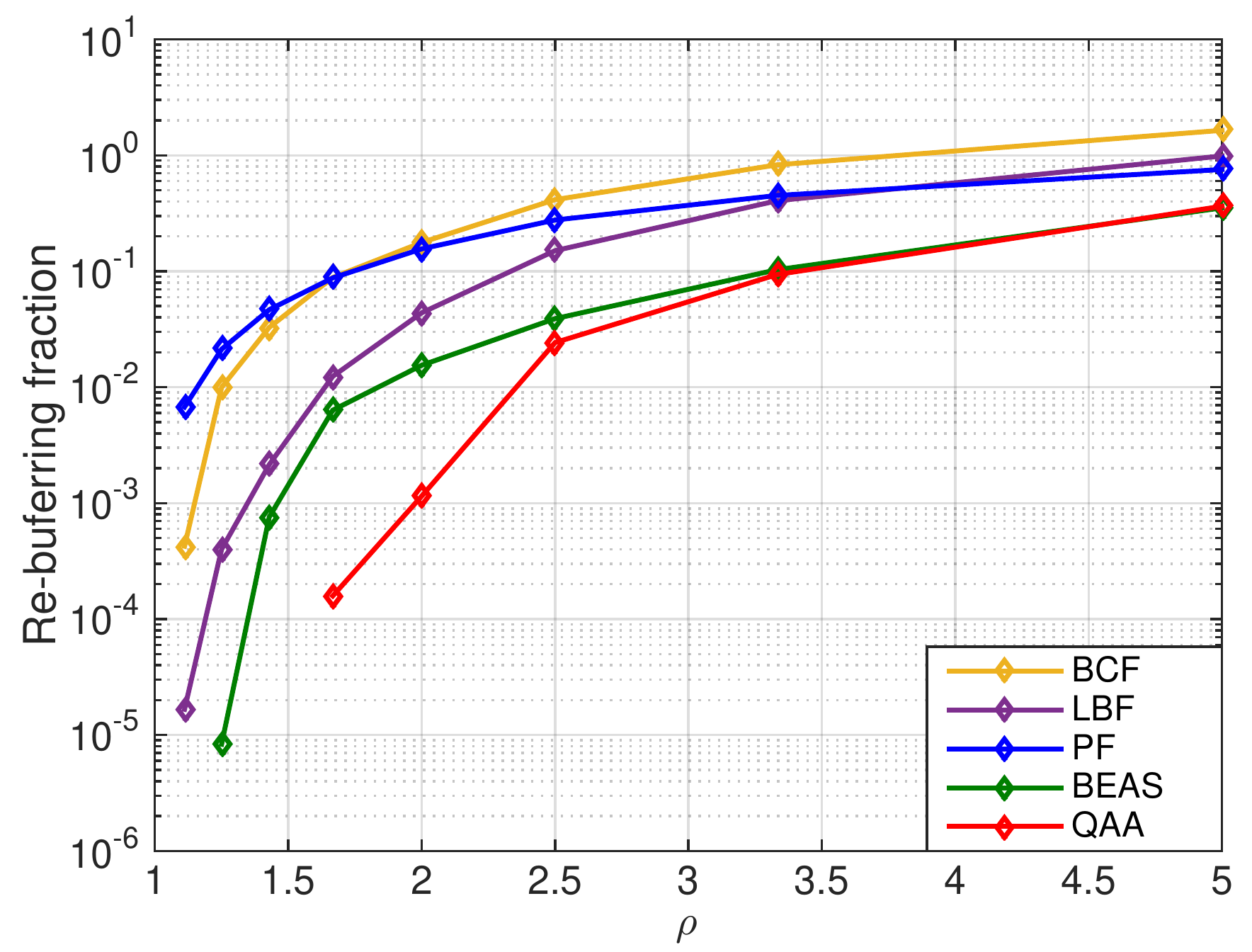}
             \caption{Re-buffering fraction}\label{fig:diffchan_rebuf}
         \end{subfigure}   \\
       \begin{subfigure}[h]{0.35\textwidth}
             \includegraphics[height=1.45in]{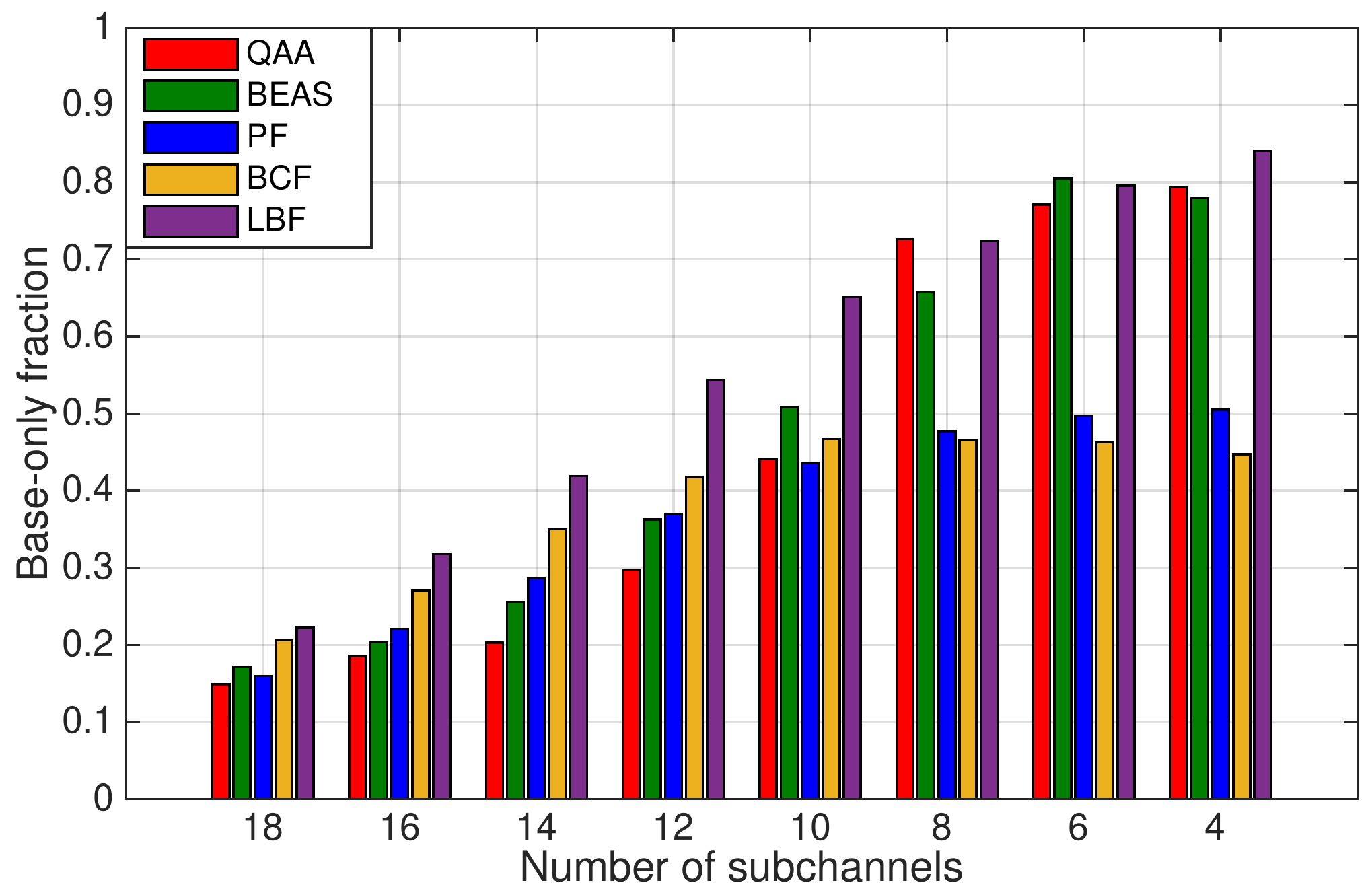}
             \caption{Base-only segment fraction}\label{fig:diffchan_layer}
        \end{subfigure}        	
\caption{Heterogeneous network with DBP-20s. Ten users experience $c_{avg} = 4.5$ Mbps  and the other ten have $c_{avg} = 2.55$ Mbps.}
\label{fig:diffchan}         
\end{figure}

\begin{figure}[h]
\hspace*{-0.6cm}
	\captionsetup[subfigure][h]{twoside,margin={0cm,0cm}}
	\centering
	\begin{subfigure}[h]{0.25\textwidth}
             \includegraphics[height=1.38in]{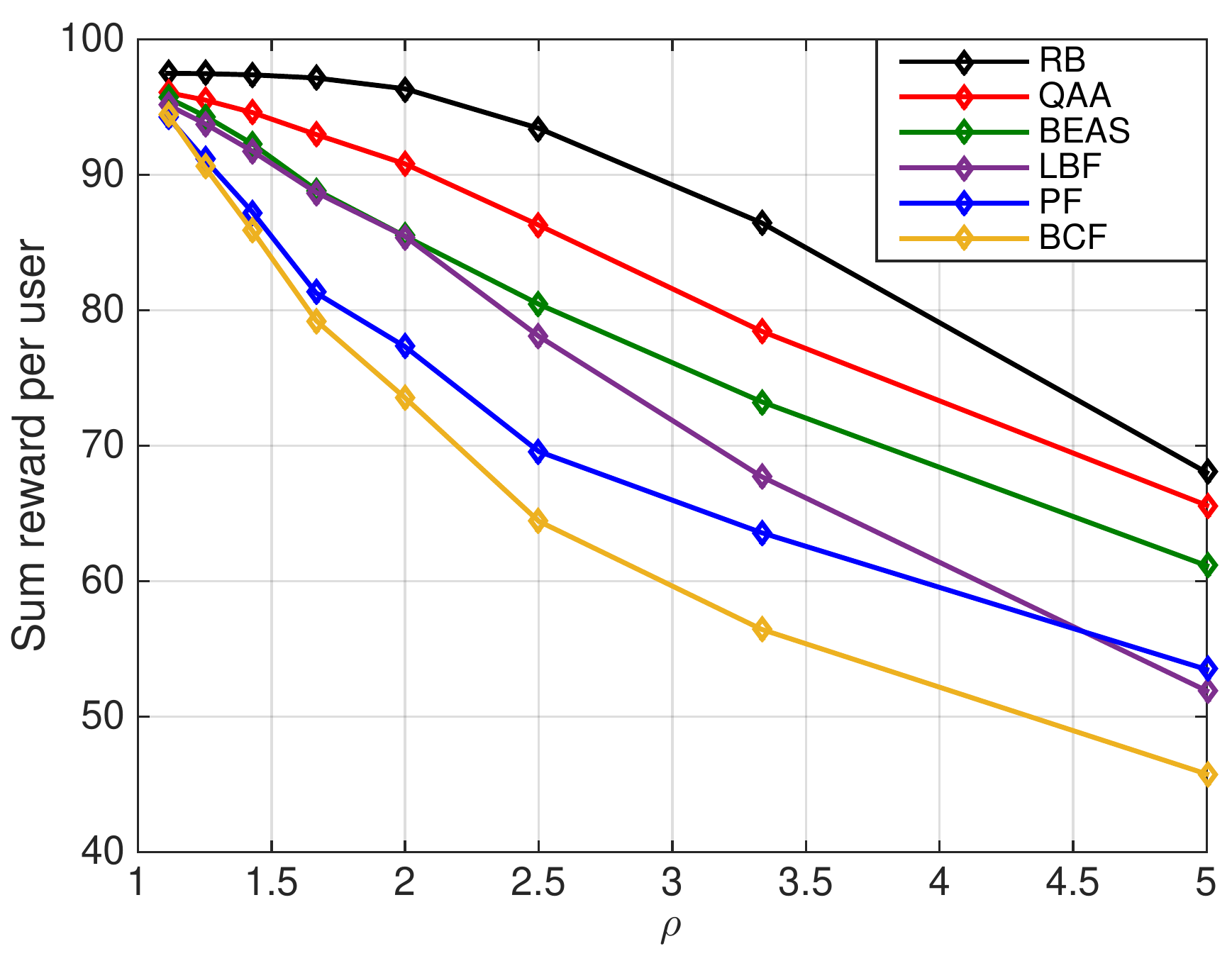}
             \caption{Reward per user}\label{fig:diffQ_reward}
        \end{subfigure}
        \begin{subfigure}[h]{0.25\textwidth}
             \includegraphics[height=1.38in]{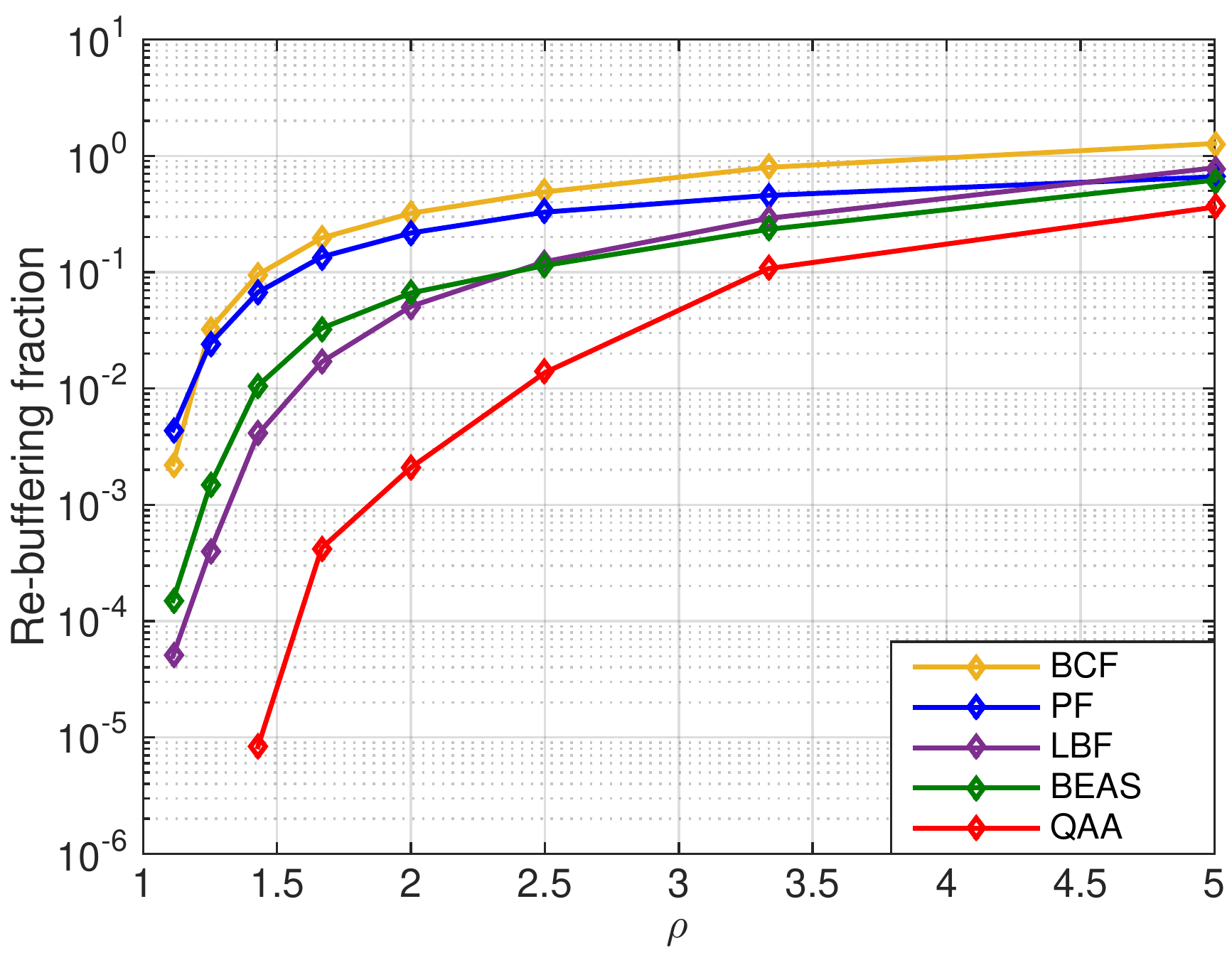}
             \caption{Re-buffering fraction}\label{fig:diffQ_rebuf}
         \end{subfigure}   \\
       \begin{subfigure}[h]{0.35\textwidth}
             \includegraphics[height=1.45in]{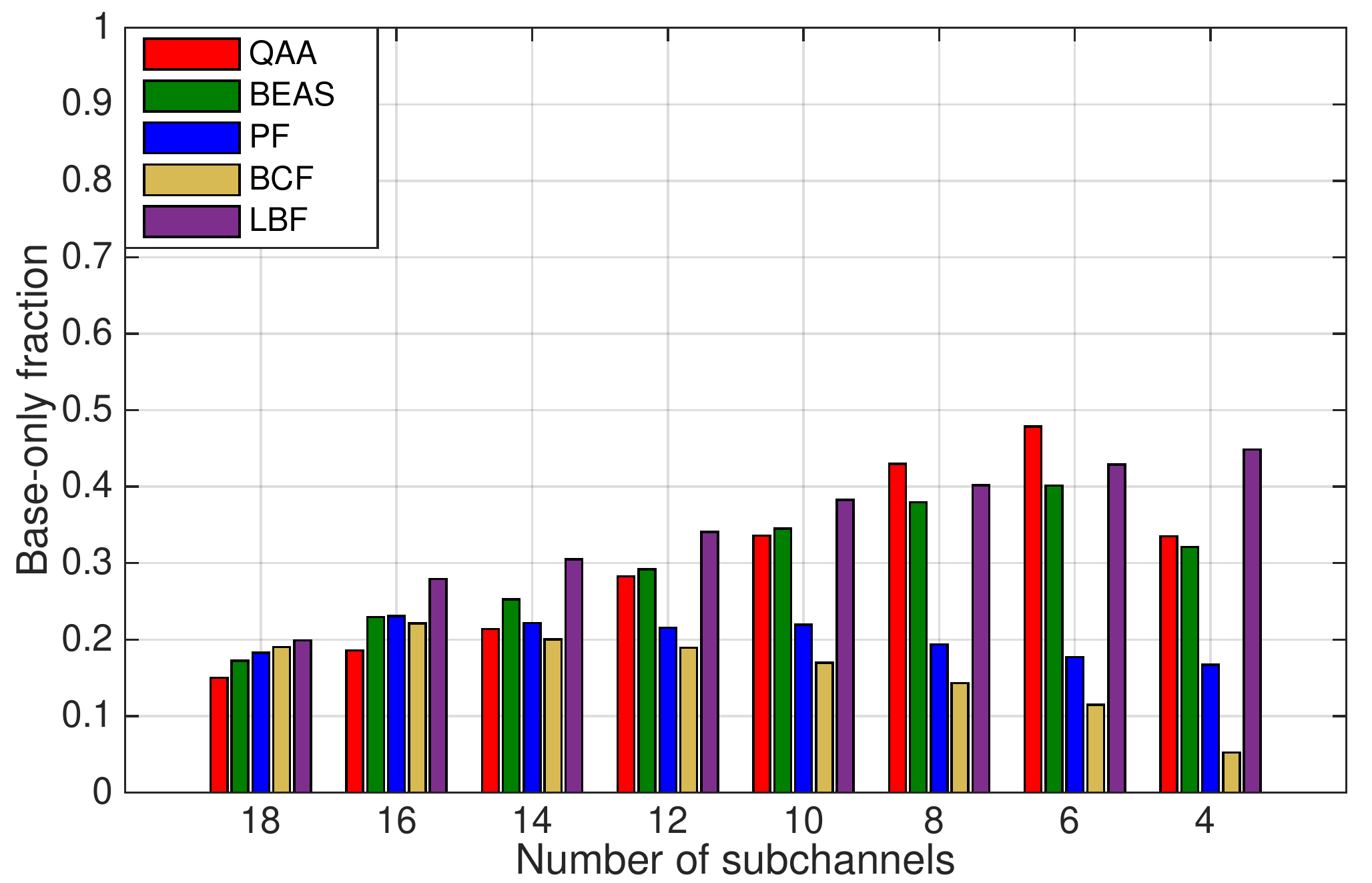}
             \caption{Base-only segment fraction}\label{fig:diffQ_layer}
        \end{subfigure}        	
\caption{Heterogeneous network with $c_{avg} = 4.5$ Mbps. Ten users use DBP-15s and the other ten use CBP for QA.}
\label{fig:diffQ}         
\end{figure}

%

\subsection{Discussion}
From the results in Sections \ref{sec:hom} and \ref{sec:het}, we can draw several conclusions. 

Since our QoE model rewards users based on their immediate playback output, it is desirable to always have non-zero segments in the buffer, preferably with as many layers as possible. Purely buffer based schemes (LBF) have an advantage in re-buffering due to the strict priority of users that are in higher risk of draining the buffer. However, for heavily loaded networks and lower data rates, the overall video quality drops because of poor spectrum usage. Therefore, the buffer level alone cannot be used as a reliable scheduling measure. On the other hand, purely channel dependent schemes (PF and BCF) do not need the buffer level as a scheduling measure since their goal is to increase network throughput. However, since the delay sensitivity of video is not taken into account in these scheduling schemes, they have poor re-buffering performance and hence, provide lower QoE. 

On the other hand, since purely buffer based schemes conservatively try to only avoid re-buffering, they fail at delivering high video quality, especially when the load is high. We can therefore conclude that by using buffer or channel alone, no scheduling policy can deliver satisfactory QoE. BEAS combines the desirable features of channel dependent and buffer dependent scheduling policies into a simple algorithm. By keeping track of the evolution of the buffer state, we can implicitly infer both the capacity and the load of the network. Whenever the buffer level for a user starts to diminish or if the user cannot build up an adequate buffer occupancy, users are scheduled based on the channel state to quickly fill the buffer and prevent re-buffering. If the buffer level grows, since there is no urgency for utilizing the channel efficiently, the scheduler prioritizes users with the lowest segments in the buffer. This also explains why in low capacity networks, where the buffer level is generally low, the gap between BEAS and channel dependent policies narrows.

Another important conclusion from these results is that, especially in wireless networks, even well designed end-to-end QA schemes cannot guarantee QoE if the underlying scheduling at the base station is not designed properly. Also, for a fixed QoE objective, BEAS can deliver up to 30\% more users on the same channel as compared to PF.

%% file: implementation.tex
\section{Testbed Implementation}\label{sec:implementation}
In this section, we describe some practical implications of BEAS followed by a testbed implementation.
For BEAS to run in a practical network, the scheduler needs to know the channel quality of each user as well as the state of its buffer. If HTTP is used as application layer protocol, the base station is able to extract data related to the next subsegment to be transmitted from each HTTP request packet that the user sends to the content provider. Thereby, it can accurately estimate the buffer level of each user. However, since more content providers are using HTTPS, this information is encrypted and cannot be retrieved by any intermediate node in the network including the base station. Recently, efforts are being made for estimating the buffer level on the users by measuring TCP/IP metrics. For instance, in \cite{krishnamoorthi2017buffest}, a machine learning-based traffic classification method is presented that aims at solving this problem. However, in the absence of these estimation techniques, the buffer state has to be fed back to the base station, in a manner similar to the CQI, as suggested in \cite{joseph2013nova, zahran2017sap,hosseini2016svc}.

We have implemented the scheduling algorithms on the \texttt{sandbox 4} network located in the \texttt{orbit}\cite{orbit} testbed. This experimental network consists of 9 nodes equipped with WiFi transceivers. The attenuation of the link between any two nodes can be manually altered from 0 to 63 dB. We use one of the nodes as a base station that contains all video segments and the other nodes act as streaming users. Then, we divide the 60 second long video into 1s long segments and encode them into a base layer and two enhancement layers using the \texttt{JSVM} encoder. We use temporal scalability where the frame rate of the temporal layers is 6, 12, and 24 frames per second. The QA deployed on all users is set  to DBP-5s.

For our experiment, in order to generate a heterogeneous wireless channel with fluctuating link capacities, we randomly change the value of the attenuation for each node every five seconds. For half of the nodes, the attenuation value is chosen randomly from 6dB, 9dB and 12dB. For the other half, the possible values for attenuation are 9dB, 12dB, and 15dB. Whenever a segment is fully retrieved by a user, the base station polls all users for their channel state which respond by sending their instantaneous channel state. For the buffer, we simplify the implementation by assuming knowledge of the duration of the segments and the layer index of the transmitted segment at the base station. Therefore, the base station can calculate all users' buffer state at any instance without the need of an explicit feedback.

Figure \ref{fig:imp} shows the performance comparison between the studied algorithms. It can be seen that similar to the simulation results, LBF has better re-buffering performance than BCF and PF, while delivering fewer enhancement layer segments. PF and BCF suffer from higher re-buffering but are able to deliver more enhancement layer segments. The benefits of both schemes are combined into BEAS which has the lowest re-buffering and while it is not always able to deliver many enhancement layers, it outperforms the other schemes in terms of total reward.

\begin{figure}[h]
\hspace*{-0.6cm}
	\captionsetup[subfigure][h]{twoside,margin={0cm,0cm}}
	\centering
	\begin{subfigure}[h]{0.25\textwidth}
             \includegraphics[height=1.38in]{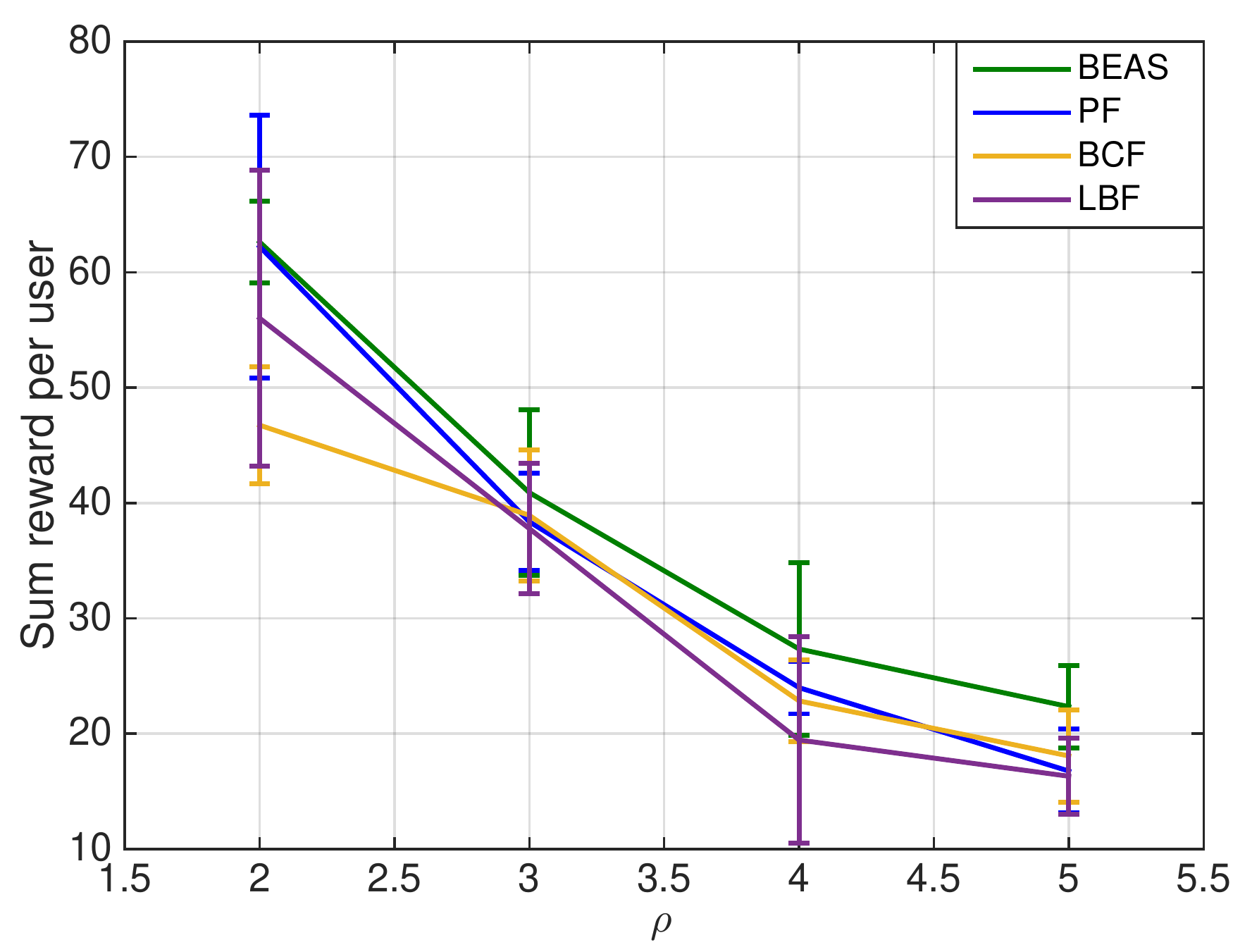}
             \caption{Reward per user}\label{fig:imp_reward}
        \end{subfigure}
        \begin{subfigure}[h]{0.25\textwidth}
             \includegraphics[height=1.38in]{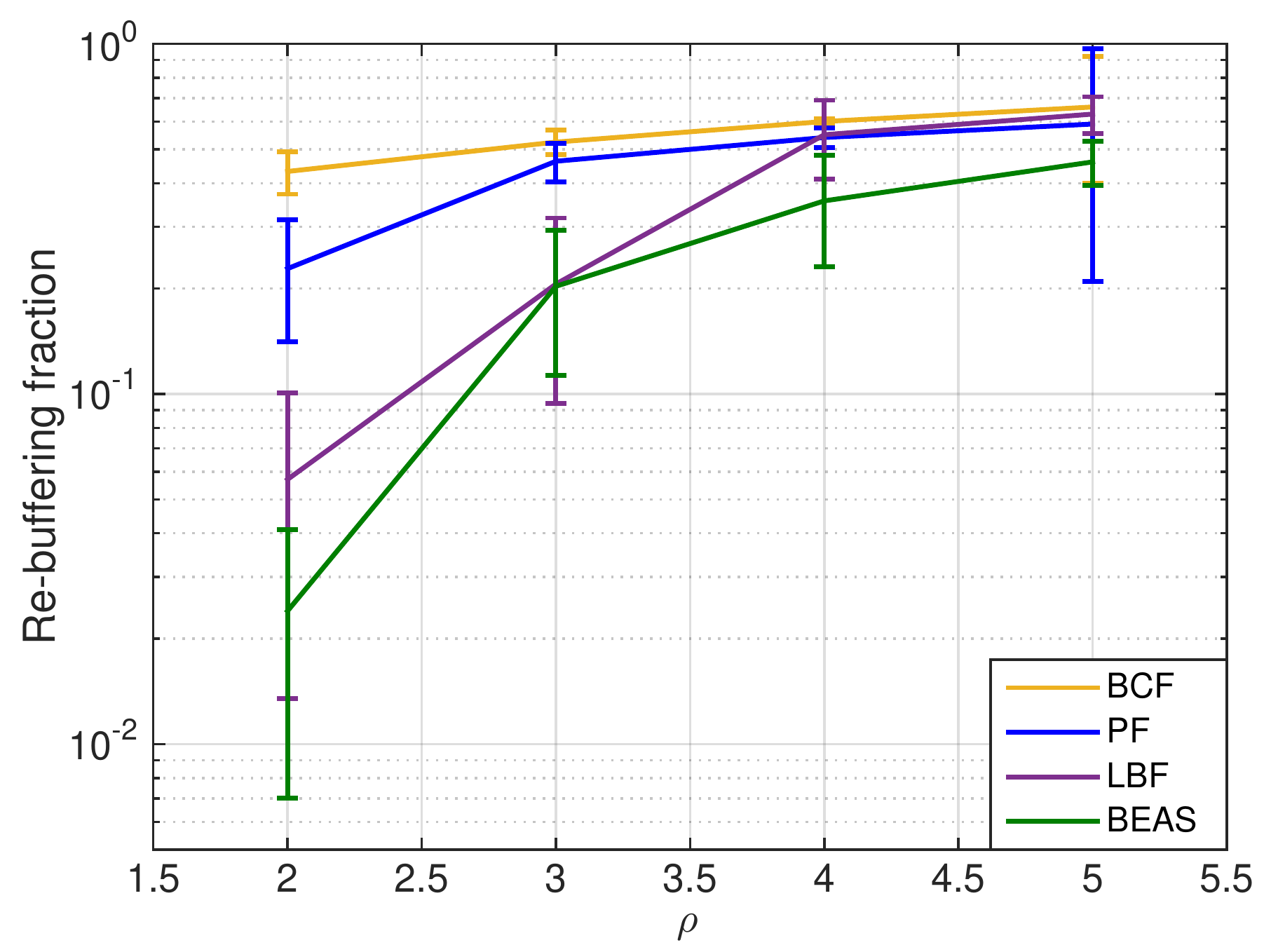}
             \caption{Re-buffering fraction}\label{fig:imp_rebuf}
        \end{subfigure} \\       
        	 \begin{subfigure}[h]{0.3\textwidth}
             \includegraphics[height=1.42in]{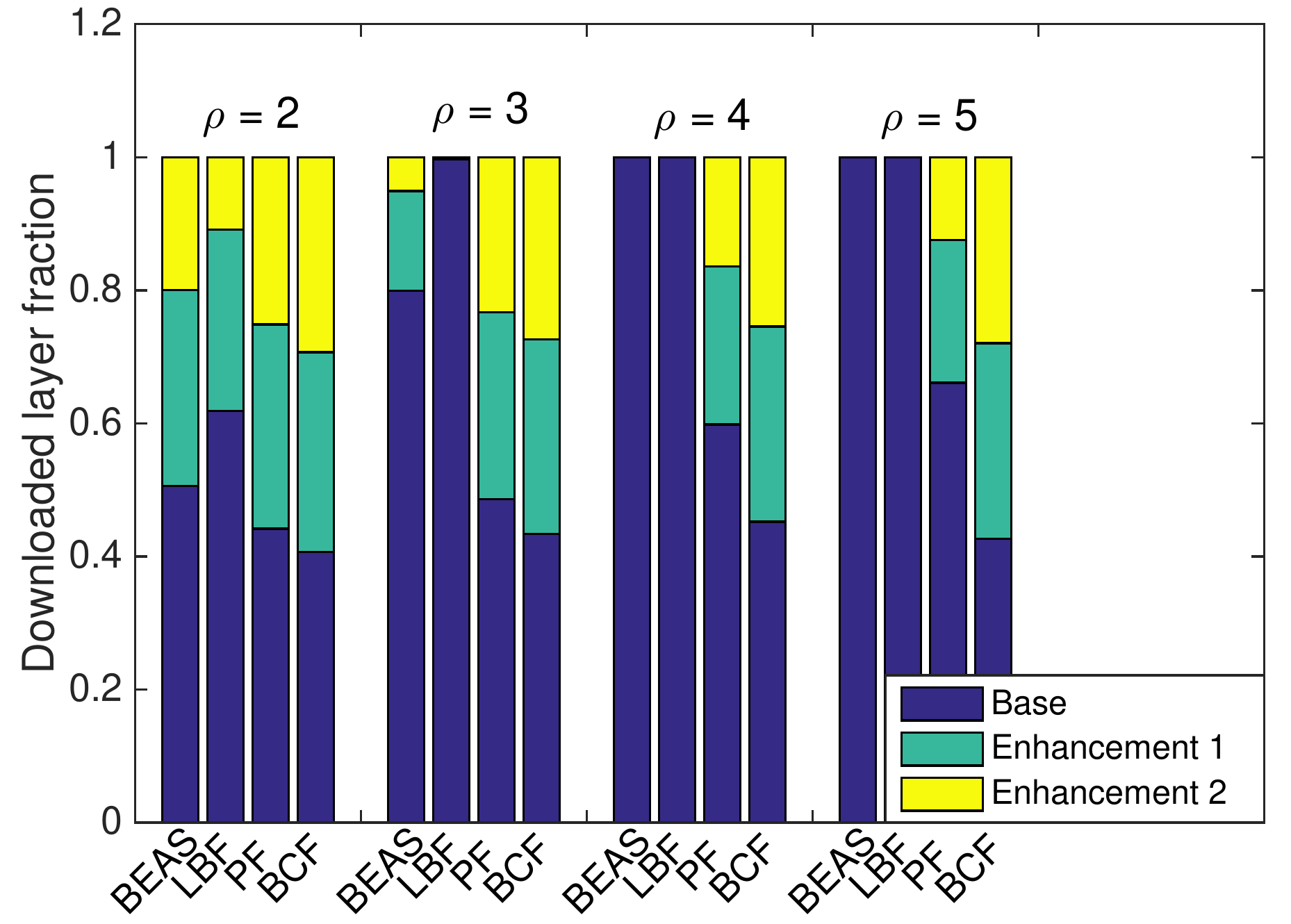}
             \caption{Fraction of downloaded layers}\label{fig:imp_layer_fraction}
        \end{subfigure}    
\caption{Testbed implementation of the scheduling algorithms. All users perform DBP-5s.}
\label{fig:imp}         
\end{figure}

%% file: conclusion.tex
\section{Conclusions and Future Work}\label{sec:con}

In this paper, we developed a framework for QA-adaptive SVC scheduling in wireless networks. We argue that instead of an overly complex and practically infeasible jointly optimized system, we should separate QA from the scheduler and adapt the scheduling policy to the underlying QA deployed on each user. Using the concept of RB, we formulate the problem as a linear program and solve it in order to obtain long term performance measures. We also formulate the jointly optimal problem as a MUSMDP in order to see the cost incurred by diverging from the jointly optimal scenario. We then develop a primal dual algorithm that performs scheduling in a QA-aware setting. By analyzing the outcome of this algorithm, we propose a heuristic scheduling algorithm that performs QA-blind scheduling with minimal complexity and signaling.

We also perform an extensive simulation study comparing the proposed scheduling algorithms with baseline schemes. Our results indicate that the optimal scheduler should have a joint buffer dependent and channel dependent behavior. By tracking the evolution of the buffer occupancy for each user, the scheduler should prioritize users that are draining the buffer and schedule them based on which user has a better channel quality. We also conclude that while QA schemes are designed to offer a good quality-delay trade-off, if the scheduler is not well designed, the end-to-end QA scheme cannot deliver high QoE in wireless networks. We finally evaluate the performance of the scheduling algorithm with a testbed implementation.

%% file: appendices.tex
\appendices
\section{}\label{app_1}
We first assume that there is an optimal allocation vector $\mathbf{x}^* = (\mathbf{x}^*_1,\cdots,\mathbf{x}^*_n)^T$ (where $\mathbf{x}^*_k = (\mathbf{x}^{0*}_k,\mathbf{x}^{1*}_k)^T$), in which the above proposition does not hold. Our goal is to show that if we replace the allocation vector $\mathbf{x}^{*}$ by the average allocation vector $\mathbf{\tilde{x}} = \frac{1}{N}\sum_{i\in\mathcal{N}}{\mathbf{x}_i^*}$, the new allocation is feasible and does not decrease the value of the objective function. 
First, we define $\mathbf{\tilde{x}}^1 = \frac{1}{N}\sum_{i\in\mathcal{N}}{\mathbf{x}_i^{1*}}$ as the average allocation vector for the active cases. Since $\mathbf{x}_i^{1*}$ ($\forall i\in\mathcal{N}$) is feasible and therefore non-negative, their average also satisfies the non-negativity constraint. Based on the definition of $\mathbf{\tilde{x}}^1 = \frac{1}{N}\sum_{i\in\mathcal{N}}{\mathbf{x}_i^{1*}}$, satisfying the resource constraint (\ref{eq:whittle}) becomes trivial. In order to check the feasibility of the new allocation for the polytope constraint (\ref{eq:depinoux}), we rewrite (\ref{eq:polytope}) as follows:
\begin{equation}\label{eq:constraint}
x_s^0 - \beta\sum_{l\in\mathcal{S}}h_{ls}^0x_{l}^0+x_s^1 - \beta\sum_{l\in\mathcal{S}}h_{ls}^1x_{l}^1 = \alpha_s,s\in\mathcal{S}.
\end{equation}
Therefore, for the optimal point, the set of constraints in (\ref{eq:depinoux}) can be represented as:
\begin{equation}
\mathbf{A}\mathbf{x}^*_i = \mathbf{A}^0\mathbf{x}_i^{0*} + \mathbf{A}^1\mathbf{x}_i^{1*} = \mathbf{\alpha} ~~ \forall i\in\mathcal{N},
\end{equation}
where $\mathbf{A}^0 = (\mathbf{I} - \beta \mathbf{H}^{0T})$ and $\mathbf{A}^1 = (\mathbf{I} - \beta \mathbf{H}^{1T})$ (Note that the user index is omitted for all matrices since they are identical for all users). In order to prove the feasibility of $\mathbf{\tilde{x}}$, we need to show that the following holds:
\begin{equation}\label{eq:proof_res}
\mathbf{A}^0\mathbf{\tilde{x}}^0 + \mathbf{A}^1\mathbf{\tilde{x}}^1 = \mathbf{\alpha}.
\end{equation}
If we plug in the values for $\mathbf{\tilde{x}}^0$ and $\mathbf{\tilde{x}}^1$, we have:

\begin{eqnarray}\label{eq:proof_res2}
\nonumber
\frac{1}{N}\sum_{i\in\mathcal{N}}\left(\mathbf{A}^0\mathbf{x}_i^{0*} + \mathbf{A}^1\mathbf{x}_i^{1*}\right) =& \frac{1}{N}\times N \mathbf\alpha = \mathbf\alpha,
\end{eqnarray}

and feasibility of $\mathbf{\tilde{x}}$ is concluded.
Finally, we check if the value of the objective function changes if we replace $\mathbf{x}^*$ with $\mathbf{\tilde{x}}$. Since the immediate rewards in each state is equal across the users, the optimal value of (\ref{eq:objective}) can be written as $\mathbf{R} \cdot\sum_{i\in\mathcal{N}}(\mathbf{x}^{0*}_i + \mathbf{x}^{1*}_i)$, where $\mathbf{R} = ( R_{s})_{1:|\mathcal{S}|}$. By replacing $\mathbf{x}_i^*$ ($\forall i \in \mathcal{N}$) with $\mathbf{\tilde{x}}$, it is easily verified that the objective value remains constant and the proof is complete.

\section{}\label{app_2}
The SMDP described by the objective function (\ref{eq:musmdp_new_objective}) can be turned into an equivalent MDP which an be solved using the Bellman equation described below for each user $n$ \cite{puterman2014markov}:

\begin{equation}\label{eq:equiv_mdp}
V_{s_n}^* = \max_{a_n} \left( \bar{r}_{s_n}^{a_n} + \sum_{{j_n}\in\mathcal{S}}M(j_n|s_n,a_n)V_{j_n}^*\right), 
\end{equation}
where $j_n$ is the state of user $n$ after action $a_n$ is fully executed, and the values for $\bar{r}_{s_n}^{a_n}$ and $M(j_n|s_n,a_n)$ are derived as follows:
\begin{eqnarray}
	\bar{r}_{s_n}^{a_n} &=& \mathbb{E}_s^a\left( \sum_{k=0}^{\tau_l}e^{-sk}R_{s_n}^k\right)\\\nonumber 
	&=& \sum_{t=0}^\infty\left(\sum_{k=0}^t\sum_{j\in\mathcal{S}_n}e^{-sk}R_{s_n}^kp(j|k,s_n,a_n)\right)f(t|s_n,a_n) \\\nonumber
	&=& \sum_{t=0}^\infty\left(\sum_{k=0}^te^{-sk}R_{s_n}^k\right)f(t|s,a) \\\nonumber
	&=& \sum_{t=0}^\infty \sum_{c\in\mathcal{C}}\left(\sum_{k=0}^te^{-sk}R_{s_n}^k f_{s_n}^{a_n}(t,c)\right). \\\nonumber
\end{eqnarray}
\begin{eqnarray} \label{eq:new_polytope}
M(j_n|s_n,a_n) &=& \sum_{k=0}^\infty e^{-sk}\mathbb{P}(k,j_n|s_n,a_n) = \mathbf{H}^{a_n}_{n,s_n,j_n}.
\end{eqnarray}
From (\ref{eq:new_polytope}) we can see that the transition probabilities of the equivalent MDP are obtained by the $\mathbf{H}^l_n$ matrices derived in (\ref{eq:hl}) and (\ref{eq:vl}). Therefore, similar to the polytope constraint derived in (\ref{eq:polytope}), we can derive the equivalent polytope constraint (\ref{eq:polytope_smdp}). 

For the resource constraint (\ref{eq:resource_smdp}), we need to show that both sides of the equation represent the expected discounted number of occupied subchannels. The right hand side is defined similar to the resource constraint (\ref{eq:simp_res_const}) and for the left hand side, we have:
\begin{eqnarray}\label{eq:constraint_deriv}
\nonumber
&&\mathbb{E}\left[\sum_{n\in\mathcal{N}}\sum_{l=1}^L\sum_{t=0}^\infty\left(e^{-st}I_{s_n}^l(t)\sum_{k=0}^{\tau_l}e^{-sk}\right)\right] = \\\nonumber
&&\sum_{n\in\mathcal{N}}\sum_{l=1}^L\mathbb{E}\left[\sum_{t=0}^\infty\left(e^{-st}I_{s_n}^l(t)\sum_{k=0}^{\tau_l}e^{-sk}\right)\right] = \\
&&\sum_{n\in\mathcal{N}}\sum_{l=1}^L\sum_{s_n\in\mathcal{S}_n}\left[\mathbb{E}_u\left(\sum_{t=0}^\infty e^{-st}I_{s_n}^l(t)\right)\mathbb{E}_\tau\sum_{k=0}^{\tau_{s_n}}e^{-sk}\right]. \\\nonumber
\end{eqnarray}
Similar to the derivation of (\ref{eq:equiv_mdp}), we conclude the following:

\begin{equation}\label{eq:exp_time}
\mathbb{E}_\tau\sum_{k=0}^{\tau}e^{-sk} = \sum_{t=0}^\infty \sum_{c\in\mathcal{C}}\left(\sum_{k=0}^te^{-sk}f_{s_n}^{a_n}(t,c)\right)
\end{equation}

By substituting (\ref{eq:exp_time}) into (\ref{eq:constraint_deriv}) and calling it $\bar\tau_{s_n}^l$, we will get the resource constraint 

\section{Calculating the average playback rate $\mu_{avg}$}\label{app_3}
Without loss of generality, we perform the derivation for the homogenous case. We can write $\mu_{avg}$ as:
\begin{equation}
\mu_{avg}=\sum_{l=0}^L\left(\sum_{i=0}^lq_i\right)\tau_l,
\end{equation}
where $q_i$ are defined as in Section \ref{sec:video_model}, with $q_0=0$ to represent a playback rate of zero for re-buffering. Also, $\tau_l$ is the fraction of total streaming time that segments with up to $l$ layers are played back according to the RB solution, which can be calculated as follows:
\begin{equation}
\tau_l=\frac{\sum_{s\in\mathcal{S}_l}(x^0_s + x^1_s)}{\sum_{s\in\mathcal{S}}(x^0_s + x^1_s)},
\end{equation}
where $x^0_s$ and $x^1_s$ are the optimal solutions of the RB for state $s$, and $\mathcal{S}_l$ is the set of all states for which up to $l$ layers are being played back $\mathcal{S}_l=\left\{s\in\mathcal{S}|b_i>0 ~~\forall i=(1,\cdots,l) \text{ and } b_i=0 ~~\forall i=(l+1,\cdots,L) \right\}$.

\section{}\label{app_4}
We denote the set of states for which $x^{*0 }_s=x^{*1}_s = 0$, as $\mathcal{S'}\subset\mathcal{S}$. First, we prove that all states in $\mathcal{S'}$ are unreachable under $u^*$. If for these states, we rewrite the constraints of RBOPT using (\ref{eq:constraint}), for the optimal points, we will have: 
\begin{equation}\label{eq:sim_constraint}
- \beta(\sum_{l\in\mathcal{S\backslash\mathcal{S'}}}h_{ls}^0x_{l}^{*0}+\sum_{l\in\mathcal{S\backslash\mathcal{S'}}}h_{ls}^1x_{l}^{*1}) = \alpha_s,s\in\mathcal{S'}.
\end{equation}
Let us first assume that $s$ is not an initial state ($\alpha_s = 0,\forall s\in\mathcal{S'}$). Since all variables on the left hand side are non-negative, (\ref{eq:sim_constraint}) can only hold if $h_{ls}^0x^{*0}_l=0$ and $h_{ls}^1x^{*1}_l = 0,\forall l\in\mathcal{S\backslash\mathcal{S'}}$, which, according to the above definition, means that state $s$ cannot be reached from any state in $\mathcal{S}\backslash\mathcal{S'}$. On the other hand, if $s$ is an initial state ($\alpha_s > 0$), then $s \notin \mathcal{S'}$, otherwise the two sides of (\ref{eq:sim_constraint}) cannot be equal. We conclude that the trajectory will start from an initial state that is a member of $\mathcal{S}\backslash\mathcal{S'}$ and that from any state belonging to this set, the optimal policy does not allow a transition from $\mathcal{S}\backslash\mathcal{S'}$ to $\mathcal{S'}$. Now, if state $s$ is unreachable, we have $h_{ls}^0x^{*0 } =h_{ls}^1x^{*1 }= 0, \forall l \in\mathcal{S}$ and also $\alpha_s = 0$. Therefore, it is easily verified that in order to have feasibility, $x_s^0=x_s^1=0$ and the proof is complete.

%% file: svc_mult_full.bbl
\begin{thebibliography}{10}
\providecommand{\url}[1]{#1}
\csname url@samestyle\endcsname
\providecommand{\newblock}{\relax}
\providecommand{\bibinfo}[2]{#2}
\providecommand{\BIBentrySTDinterwordspacing}{\spaceskip=0pt\relax}
\providecommand{\BIBentryALTinterwordstretchfactor}{4}
\providecommand{\BIBentryALTinterwordspacing}{\spaceskip=\fontdimen2\font plus
\BIBentryALTinterwordstretchfactor\fontdimen3\font minus
  \fontdimen4\font\relax}
\providecommand{\BIBforeignlanguage}[2]{{%
\expandafter\ifx\csname l@#1\endcsname\relax
\typeout{** WARNING: IEEEtran.bst: No hyphenation pattern has been}%
\typeout{** loaded for the language `#1'. Using the pattern for}%
\typeout{** the default language instead.}%
\else
\language=\csname l@#1\endcsname
\fi
#2}}
\providecommand{\BIBdecl}{\relax}
\BIBdecl

\bibitem{sandvine2016}
Sandvine, ``Sandvine global {I}nternet phenomena: {L}atin {A}merica and {N}orth
  {A}merica,'' Sandvine, Intelligent Boardband Networks, Tech. Rep., June 2016.

\bibitem{sanchez2011idash}
Y.~S{\'a}nchez de~la Fuente, T.~Schierl, C.~Hellge, T.~Wiegand, D.~Hong,
  D.~De~Vleeschauwer, W.~Van~Leekwijck, and Y.~Le~Bou{\'e}dec, ``i{D}{A}{S}{H}:
  improved dynamic adaptive streaming over {H}{T}{T}{P} using scalable video
  coding,'' in \emph{Proceedings of the Second annual ACM Conference on
  Multimedia systems}, 2011, pp. 257--264.

\bibitem{vidyo_2013}
``Vidyo and {G}oogle collaborate to enhance video quality within {WebRTC},''
  Press Release, August 2013,
  http://www.vidyo.com/company/news-and-events/press-releases/vidyo-and-google-collaborate-to-
  enhance-video -quality-within-webrtc/.

\bibitem{thomas2016applications}
E.~Thomas, M.~van Deventer, T.~Stockhammer, A.~C. Begen, M.-L. Champel, and
  O.~Oyman, ``Applications and deployments of server and network assisted
  {D}{A}{S}{H} ({S}{A}{N}{D}),'' 2016.

\bibitem{whittle1988restless}
P.~Whittle, ``Restless bandits: Activity allocation in a changing world,''
  \emph{Journal of Applied Probability}, pp. 287--298, 1988.

\bibitem{gittins1979bandit}
J.~C. Gittins, ``Bandit processes and dynamic allocation indices,''
  \emph{Journal of the Royal Statistical Society. Series B}, pp. 148--177,
  1979.

\bibitem{hosseini2015not}
S.~A. Hosseini, F.~Fund, and S.~S. Panwar, ``({N}ot) yet another policy for
  scalable video delivery to mobile users,'' in \emph{Proceedings of the 7th
  ACM International Workshop on Mobile Video}.\hskip 1em plus 0.5em minus
  0.4em\relax ACM, 2015, pp. 17--22.

\bibitem{andelin2012quality}
T.~Andelin, V.~Chetty, D.~Harbaugh, S.~Warnick, and D.~Zappala, ``Quality
  selection for dynamic adaptive streaming over {H}{T}{T}{P} with scalable
  video coding,'' in \emph{Proceedings of the 3rd ACM Multimedia Systems
  Conference}, 2012, pp. 149--154.

\bibitem{xiang2012adaptive}
S.~Xiang, L.~Cai, and J.~Pan, ``Adaptive scalable video streaming in wireless
  networks,'' in \emph{Proceedings of the 3rd ACM Multimedia Systems
  Conference}, 2012, pp. 167--172.

\bibitem{otwani2015optimal}
J.~Otwani, A.~Agarwal, and A.~K. Jagannatham, ``Optimal scalable video
  scheduling policies for real-time single-and multiuser wireless video
  networks,'' \emph{IEEE Transactions on Vehicular Technology}, vol.~64, no.~6,
  pp. 2424--2435, 2015.

\bibitem{kim2005optimal}
T.~Kim and M.~H. Ammar, ``Optimal quality adaptation for scalable encoded
  video,'' \emph{Selected Areas in Communications, IEEE Journal on}, vol.~23,
  no.~2, pp. 344--356, 2005.

\bibitem{kuschnig2010evaluation}
R.~Kuschnig, I.~Kofler, and H.~Hellwagner, ``An evaluation of {T}{C}{P}-based
  rate-control algorithms for adaptive internet streaming of
  {H}.264/{S}{V}{C},'' in \emph{Proceedings of the first annual ACM SIGMM
  conference on Multimedia systems}.\hskip 1em plus 0.5em minus 0.4em\relax
  ACM, 2010, pp. 157--168.

\bibitem{sieber2013implementation}
C.~Sieber, T.~Ho{\ss}feld, T.~Zinner, P.~Tran-Gia, and C.~Timmerer,
  ``Implementation and user-centric comparison of a novel adaptation logic for
  {D}{A}{S}{H} with {S}{V}{C},'' in \emph{Integrated Network Management (IM
  2013), IFIP/IEEE International Symposium on}, 2013, pp. 1318--1323.

\bibitem{muller2012using}
C.~M{\"u}ller, D.~Renzi, S.~Lederer, S.~Battista, and C.~Timmerer, ``Using
  scalable video coding for dynamic adaptive streaming over {H}{T}{T}{P} in
  mobile environments,'' in \emph{Signal Processing Conference (EUSIPCO), 2012
  Proceedings of the 20th European}.\hskip 1em plus 0.5em minus 0.4em\relax
  IEEE, 2012, pp. 2208--2212.

\bibitem{huang2015buffer}
T.-Y. Huang, R.~Johari, N.~McKeown, M.~Trunnell, and M.~Watson, ``A
  buffer-based approach to rate adaptation: Evidence from a large video
  streaming service,'' \emph{ACM SIGCOMM Computer Communication Review},
  vol.~44, no.~4, pp. 187--198, 2015.

\bibitem{freris2013distortion}
N.~M. Freris, C.-H. Hsu, J.~P. Singh, and X.~Zhu, ``Distortion-aware scalable
  video streaming to multinetwork clients,'' \emph{IEEE/ACM Transactions on
  Networking}, vol.~21, no.~2, pp. 469--481, 2013.

\bibitem{freris2010resource}
N.~M. Freris, C.-H. Hsu, X.~Zhu, and J.~P. Singh, ``Resource allocation for
  multihomed scalable video streaming to multiple clients,'' in
  \emph{Multimedia (ISM), 2010 IEEE International Symposium on}.\hskip 1em plus
  0.5em minus 0.4em\relax IEEE, 2010, pp. 9--16.

\bibitem{ji2009scheduling}
X.~Ji, J.~Huang, M.~Chiang, G.~Lafruit, and F.~Catthoor, ``Scheduling and
  resource allocation for {S}{V}{C} streaming over {O}{F}{D}{M} downlink
  systems,'' \emph{Circuits and Systems for Video Technology, IEEE Transactions
  on}, vol.~19, no.~10, pp. 1549--1555, 2009.

\bibitem{khan2016qoe}
N.~Khan and M.~G. Martini, ``{Q}o{E}-driven multi-user scheduling and rate
  adaptation with reduced cross-layer signaling for scalable video streaming
  over {L}{T}{E} wireless systems,'' \emph{EURASIP Journal on Wireless
  Communications and Networking}, vol. 2016, no.~1, p.~1, 2016.

\bibitem{zhao2015qoe}
M.~Zhao, X.~Gong, J.~Liang, W.~Wang, X.~Que, and S.~Cheng, ``{Q}o{E}-driven
  cross-layer optimization for wireless dynamic adaptive streaming of scalable
  videos over {H}{T}{T}{P},'' \emph{Circuits and Systems for Video Technology,
  IEEE Transactions on}, vol.~25, no.~3, pp. 451--465, 2015.

\bibitem{talebi2011quasi}
M.~S. Talebi, A.~Khonsari, M.~H. Hajiesmaili, and S.~Jafarpour, ``Quasi-optimal
  network utility maximization for scalable video streaming,'' \emph{arXiv
  preprint arXiv:1102.2604}, 2011.

\bibitem{zhang2010cross}
H.~Zhang, Y.~Zheng, M.~A. Khojastepour, and S.~Rangarajan, ``Cross-layer
  optimization for streaming scalable video over fading wireless networks,''
  \emph{IEEE Journal on Selected Areas in Communications}, vol.~28, no.~3,
  2010.

\bibitem{xiao2015optimal}
Y.~Xiao and M.~van~der Schaar, ``Optimal foresighted multi-user wireless
  video,'' \emph{Selected Topics in Signal Processing, IEEE Journal of},
  vol.~9, no.~1, pp. 89--101, 2015.

\bibitem{fu2009systematic}
F.~Fu and M.~van~der Schaar, ``A systematic framework for dynamically
  optimizing multi-user wireless video transmission,'' \emph{arXiv preprint
  arXiv:0903.0207}, 2009.

\bibitem{cicalo2014distortion}
S.~Cicalo and V.~Tralli, ``Distortion-fair cross-layer resource allocation for
  scalable video transmission in ofdma wireless networks,'' \emph{IEEE
  Transactions on Multimedia}, vol.~16, no.~3, pp. 848--863, 2014.

\bibitem{zahran2017sap}
A.~H. Zahran, J.~J. Quinlan, K.~Ramakrishnan, and C.~J. Sreenan, ``{S}{A}{P}:
  Stall-aware pacing for improved {D}{A}{S}{H} video experience in cellular
  networks,'' in \emph{Proceedings of the 8th ACM on Multimedia Systems
  Conference}.\hskip 1em plus 0.5em minus 0.4em\relax ACM, 2017, pp. 13--26.

\bibitem{chen2013scheduling}
J.~Chen, R.~Mahindra, M.~A. Khojastepour, S.~Rangarajan, and M.~Chiang, ``A
  scheduling framework for adaptive video delivery over cellular networks,'' in
  \emph{Proceedings of the 19th annual international conference on Mobile
  computing \& networking}.\hskip 1em plus 0.5em minus 0.4em\relax ACM, 2013,
  pp. 389--400.

\bibitem{georgopoulos2013towards}
P.~Georgopoulos, Y.~Elkhatib, M.~Broadbent, M.~Mu, and N.~Race, ``Towards
  network-wide {Q}o{E} fairness using openflow-assisted adaptive video
  streaming,'' in \emph{Proceedings of the 2013 ACM SIGCOMM workshop on Future
  human-centric multimedia networking}.\hskip 1em plus 0.5em minus 0.4em\relax
  ACM, 2013, pp. 15--20.

\bibitem{petrangeli2016qoe}
S.~Petrangeli, J.~Famaey, M.~Claeys, S.~Latr{\'e}, and F.~De~Turck,
  ``Qoe-driven rate adaptation heuristic for fair adaptive video streaming,''
  \emph{ACM Transactions on Multimedia Computing, Communications, and
  Applications (TOMM)}, vol.~12, no.~2, p.~28, 2016.

\bibitem{hu2012qoe}
H.~Hu, X.~Zhu, Y.~Wang, R.~Pan, J.~Zhu, and F.~Bonomi, ``{Q}o{E}-based
  multi-stream scalable video adaptation over wireless networks with proxy,''
  in \emph{IEEE International Conference on Communications (ICC)}, 2012, pp.
  7088--7092.

\bibitem{joseph2014nova}
V.~Joseph and G.~de~Veciana, ``{N}{O}{V}{A}: {Q}o{E}-driven optimization of
  {D}{A}{S}{H}-based video delivery in networks,'' in \emph{IEEE INFOCOM
  2014-IEEE Conference on Computer Communications}, 2014, pp. 82--90.

\bibitem{d1960probleme}
F.~d'Epenoux, ``Sur un probleme de production et de stockage dans
  l’al{\'e}atoire,'' \emph{Revue Fran{\c{c}}aise de Recherche
  Op{\'e}rationelle}, vol.~14, pp. 3--16, 1960.

\bibitem{weber1990index}
R.~R. Weber and G.~Weiss, ``On an index policy for restless bandits,''
  \emph{Journal of Applied Probability}, pp. 637--648, 1990.

\bibitem{puterman2014markov}
M.~L. Puterman, \emph{Markov decision processes: discrete stochastic dynamic
  programming}.\hskip 1em plus 0.5em minus 0.4em\relax John Wiley \& Sons,
  2014.

\bibitem{murty1983linear}
K.~G. Murty, \emph{Linear programming}.\hskip 1em plus 0.5em minus 0.4em\relax
  John Wiley \& Sons, 1983.

\bibitem{fund2013performance}
F.~Fund, C.~Wang, Y.~Liu, T.~Korakis, M.~Zink, and S.~S. Panwar, ``Performance
  of dash and webrtc video services for mobile users,'' in \emph{Packet Video
  Workshop (PV), 2013 20th International}.\hskip 1em plus 0.5em minus
  0.4em\relax IEEE, 2013, pp. 1--8.

\bibitem{kelly1998rate}
F.~P. Kelly, A.~K. Maulloo, and D.~K. Tan, ``Rate control for communication
  networks: shadow prices, proportional fairness and stability,'' \emph{Journal
  of the Operational Research Society}, vol.~49, no.~3, pp. 237--252, 1998.

\bibitem{yeh2014bandwidth}
Y.-H. Yeh, Y.-C. Lai, Y.-H. Chen, and C.-N. Lai, ``A bandwidth allocation
  algorithm with channel quality and {Q}o{S} aware for {I}{E}{E}{E} 802.16 base
  stations,'' \emph{International Journal of Communication Systems}, vol.~27,
  no.~10, pp. 1601--1615, 2014.

\bibitem{krishnamoorthi2017buffest}
V.~Krishnamoorthi, N.~Carlsson, E.~Halepovic, and E.~Petajan, ``Buffest:
  Predicting bufer conditions and real-time requirements of http (s) adaptive
  streaming clients,'' in \emph{Proceedings of the 8th ACM Multimedia Systems
  Conference}, 2017.

\bibitem{hosseini2016svc}
S.~A. Hosseini, Z.~Lu, G.~de~Veciana, and S.~S. Panwar, ``{S}{V}{C}-based
  multi-user streamloading for wireless networks,'' \emph{IEEE Journal on
  Selected Areas in Communications}, vol.~34, no.~8, pp. 2185--2197, 2016.

\bibitem{orbit}
``Orbit lab tutorial,'' \url{http://www.orbit-lab.org/wiki/Tutorials},
  accessed: 2017-06-22.

\end{thebibliography}
